%% file: main.tex
\documentclass[showkeys,floatfix,noshowpacs, preprintnumbers, twocolumn, aps,superscriptaddress,pra,10pt,nofootinbib]{revtex4-2}

\input{preamble}

\newcommand{\ulmshort}{Institute of Theoretical Physics, Ulm University, Ulm, Germany} 
\newcommand{\iqstshort}{Center for Integrated Quantum Science and Technology (IQST), 89081 Ulm, Germany} 

\begin{document}

\title{Zeno-Assisted Quantum Heat Engines}

\begin{abstract}
    Finite-time quantum heat engines (QHEs) typically extract less work than their quasistatic counterparts because fast driving generates coherences and non-adiabatic transitions during the work strokes, a phenomenon commonly referred to as quantum friction. Quantum lubrication denotes a broad class of strategies that use auxiliary systems or controls to mitigate this loss. In this work, we introduce a lubrication protocol based on the quantum Zeno dynamics (QZD). By coupling the working medium to an auxiliary lubricant system and frequently monitoring the lubricant, we confine the joint evolution to a Zeno subspace and obtain an effective shortcut to adiabaticity during the work strokes of a QHE running an Otto cycle. In the ideal Zeno limit, the protocol reproduces the transitionless dynamics required to preserve populations in the instantaneous energy basis and recover the Otto efficiency at finite stroke duration. We also analyze several implementation-dependent thermodynamic costs, including switching, driving, monitoring, and imperfect thermalization, in order to assess how these costs constrain the practical gains in efficiency and power. Our results identify QZD as a conceptually distinct route to quantum lubrication and highlight quantum heat engines as a useful setting in which to study the interplay between strong coupling, measurement, and quantum thermodynamic control.
\end{abstract}

\setlength{\skip\footins}{12pt}
\def\thefootnote{$\star$}\footnotetext{These authors contributed equally to this work.}\def\thefootnote{\arabic{footnote}}

\author{Selma Memi\'{c}$^\star$}
\email{selma.memic@uni-ulm.de}
\affiliation{\ulmshort}
\affiliation{\iqstshort}
\author{Rafael Wagner$^\star$}
\email{rafael.wagner@uni-ulm.de}
\affiliation{\ulmshort}
\affiliation{\iqstshort}
\author{Susana F. Huelga}
\email{susana.huelga@uni-ulm.de}
\affiliation{\ulmshort}
\affiliation{\iqstshort}
\author{Martin B. Plenio}
\email{martin.plenio@uni-ulm.de}
\affiliation{\ulmshort}
\affiliation{\iqstshort}

\maketitle

\tableofcontents

\section{Introduction}

Quantum heat engines (QHEs)~\cite{cangemi_quantum_2024,scovil_three-level_1959} lie at the core of research in quantum thermodynamics~\cite{deffner2019quantum,kosloff_quantum_1984,alicki1979quantumopensystemheatengines,alicki_introduction_2018}. Their study connects to a wide range of practical and foundational questions, from the design of microscopic refrigerators for quantum computation~\cite{aamir_thermally_2025} to the energetic requirements of quantum information processing~\cite{auffeves2022quantumenergy}. They also raise conceptual issues about how to define work and heat in a microscopic quantum setting~\cite{pusz1978passive,alicki1979quantumopensystemheatengines,talkner2007fluctuation,baumer2018fluctuating}, how those definitions impact the efficiency and power of nanoscale engines~\cite{rossnagel2014nanoscale,nelly2017surpassing}, and what are the essential differences between classical and quantum thermodynamics~\cite{woods2019maximumefficiencyof,lostaglio2020certifying,lostaglio2018quantumfluctuation,comar2025contextuality}.

A quantum heat engine is the quantum analog of a classical heat engine: the working medium is a quantum system, the reservoirs are modeled as large quantum systems in thermal states, and work and heat are usually defined through changes in expectation values of the relevant Hamiltonians. Here, we focus on \emph{stroke-based engines} (rather than continuous engines~\cite{kosloff_quantum_2014}) and in particular on the prototypical four-stroke Otto cycle~\cite{kosloff_quantum_2017}. In this setting the dynamics is separated into \emph{work strokes}, during which an external control drives the working medium, and \emph{heat strokes}, during which the working medium is coupled to a thermal bath.

Previous literature~\cite{kosloff_discrete_2002,feldmann2006quantum,feldmann_heat_1996,feldmann2003quantum} on both classical and quantum engines  has shown that \emph{finite-time} operation can degrade performance. In classical finite-time thermodynamics, this deterioration stems from various
limitations, such as finite heat transfer~\cite{curzon1975efficiency,vandenbroeck2005thermodynamic} or low dissipation in stochastic engines~\cite{schmiedl2007efficiency, espositio2010efficiency}. In the quantum setting, the main difficulty during finite-time work strokes is different: when the drive is not quasistatic,~\footnote{In quantum control,~``\emph{adiabatic}'' usually means that the Hamiltonian changes slowly enough to suppress transitions between instantaneous eigenstates~\cite{kato1950ontheadiabatic,wang2016necessary,albash2018adiabatic}. In quantum thermodynamics, the same term is often used to express that the system is \emph{closed}, so that no heat is exchanged with the environment. To avoid confusion, we will normally use ``adiabatic'' in its thermodynamic sense, and use ``\emph{quasistatic}'' or ``\emph{transitionless}''~\cite{berry2009transitionless} for the quantum-control notion. The standard expression ``\emph{shortcuts to adiabaticity}'' is retained because it is established terminology in the control literature.} it generally induces transitions between instantaneous energy eigenstates and creates coherences in that basis~\cite[Chapter 10]{griffiths2005introduction}. This effect is commonly known as \emph{quantum friction}~\cite{rezek_reflections_2010,kosloff_discrete_2002, delCampo2018friction-free} because it reduces the extractable work relative to the quasistatic ideal.

Several strategies have been proposed to mitigate quantum friction in finite-time QHEs. One route is to engineer an auxiliary dephasing environment that selectively removes coherences while preserving populations in the working medium; this idea has been referred to as \emph{quantum lubrication}~\cite{feldmann2006quantum}. For example, Weber \emph{et al.}~\cite{weber2023thermodynamiccostspuredephasing} considered a setup where the working system is coupled to an engineered lubricant comprising a harmonic oscillator that interacts with a dissipative thermal bath. This composite system, if correctly calibrated (as they show), yields a Markovian evolution that suitably drives the system, dumping coherences into the auxiliary lubricant and thereby avoiding unwanted transitions. Another approach is to optimize timing and stroke durations so that coherences generated during different parts of the protocol destructively interfere~\cite{camati2019coherence,rezek2009thequantumrefrigerator}. Yet another complementary family of techniques is known under the heading of \emph{shortcuts to adiabaticity}~\cite{del_campo_shortcuts_2013,alipour_shortcuts_2020,takahashi_shortcuts_2024,guery-odelin_shortcuts_2019}: by adding specially designed control terms, one enforces quasistatic-like state transformations in finite time~\cite{campo_more_2014,beau_scaling-up_2016, hou2025efficiency,shende2023otto,Deng2018superadiabatic,pedram2023quantum}. 

All of these methods explore variants of the same foundational question: What are the most feasible, effective, and physically well-motivated ways to assist a finite-time QHE in approaching the performance of the ideal (quasistatic) engine while operating at finite power? 

\textbf{Main results.} In this work, we propose and analyze a distinct form of lubrication based on the \emph{quantum Zeno effect}~\cite{misra1977zeno,greenfield_unified_2025,facchi2008quantum,vonNeumann1932mathematische}. We highlight our main results in~\firstbox. Rather than relying on dephasing to remove coherences, we trade dissipation for continuous monitoring that inhibits the formation of unwanted coherences. The key idea is to frequently monitor additional systems coupled to the working medium so that its evolution is constrained within certain \emph{Zeno subspaces}~\cite{facchi2002zeno}. This strategy, known as \emph{quantum Zeno dynamics} (QZD)~\cite{facchi2000quantumzeno,facchi2008quantum}, has been investigated in a variety of contexts and found to be useful for tasks ranging from error detection and correction~\cite{hacohen2018incoherent,erez_correcting_2004,cramer_repeated_2016,bluvstein_architectural_2025}, to adiabatic quantum computation~\cite{paz-silva_zeno_2012,berwald_grover_2024,berwald_zeno-effect_2025,zhang_optimal_2025,lewalle_multi-qubit_2023}, and other quantum control tasks~\cite{hacohen-gourgy_continuous_2020,dizaji_Hamiltonian_2024,lewalle2024optimal}.

Concretely, we consider a finite-time drive that, left unchecked, produces coherences which degrade the engine's work output. We then show how to assist the working medium during the work strokes of an Otto cycle by coupling it to an auxiliary (lubricant) system that is \emph{strongly coupled} to the working system during the drive and subject to a suitable monitoring protocol. The effective dynamics is then described by the interplay between \emph{two} different manifestations of the QZD: a strong-coupling drive~\cite{burgarth2019generalized,burgarth2022oneboundtorulethem,facchi2002zeno,schulman1998continuous} and a limiting sequence of strong projective measurements~\cite{misra1977zeno,facchi2001from,facchi2010quantumzenoeffect}, which together yield an effective time-dependent Hamiltonian that includes the original one and a counter-diabatic term inducing a shortcut to adiabaticity. Therefore, this dynamics (in the Zeno limit) preserves the populations of the working medium, inhibits the generation of unwanted coherences due to the finite drive, and allows the cycle to reach the Otto efficiency while operating at relatively short stroke duration. 

We also compare our results with those of Weber~\emph{et al.}~\cite{weber2023thermodynamiccostspuredephasing}. In their work, thermodynamic costs of their lubrication protocol were analyzed and shown to disappear in certain limits. Here, we demonstrate that a closely related system dynamics,  numerically found in Ref.~\cite{weber2023thermodynamiccostspuredephasing}, can be obtained from a Zeno-assisted QHE. Thus, we recover the transitionless drive that they obtain in analogous parameter regimes of their lubricant model.

\begin{mainresultsbox}{\hypertarget{box_1}{BOX 1: Main Results}}
\vspace{0.05cm}
{\parbox[h]{\dimexpr\linewidth\relax}{
\begin{center} \usym{1F441}~\textbf{Zeno-assisted scheme}\end{center}
Using the quantum Zeno dynamics, we engineer a lubrication scheme whose effective Hamiltonian contains a counter-diabatic term, yielding an emergent shortcut to adiabaticity. We simulate the complete dynamics, including the lubricant and its monitoring, for a single-qubit engine, while the theoretical result applies to any finite-dimensional system.
\begin{center} \textbf{\usym{1F501}}~\textbf{Comparison with Weber~et~al.}\end{center}
We compare our protocol, which relies on strong coupling and continuous monitoring, with their strong-dephasing protocol. We show that, in a limiting regime, our Zeno-assisted engine reproduces a working-medium dynamics closely related to one found numerically by Weber et al.~\cite{weber2023thermodynamiccostspuredephasing}.
\begin{center}\usym{2699}~\textbf{Thermodynamic costs}\end{center} We thoroughly investigate whether the
ideal protocol continues to perform well in specific regimes once detrimental effects
are taken into account. We find that the fundamental thermodynamic costs become negligible in the Zeno limit and in suitable parameter regimes.
}}\par
\vspace{0.4cm}
\end{mainresultsbox}

At the level of the idealized model, our protocol suggests that one can approach the Otto efficiency at finite stroke duration and, in that sense, alleviate the usual efficiency-power trade-off. In particular, if the work strokes are lubricated so that they remain effectively transitionless, and if the thermalization strokes are also made sufficiently fast, then the cycle can in principle deliver  large power while remaining close to the ideal Otto benchmark. In practice, however, this conclusion is constrained by implementation-dependent costs. Quantum heat engines are delicate nanoscale devices~\cite{oftelie2026impact}, and coupling the working medium to a lubricant can introduce additional energetic and entropic penalties. For that reason, we analyze several such costs in detail---including switching work cost~\cite{molitor2020stroboscopic}, control-driving cost~\cite{zheng2016counterdiabaticcost}, losses due to imperfect thermalization~\cite{rivas2012open, breuer2002theory}, cost of preparing an initial pure coherent state~\cite{abdelkhalek2016energy}, the cost of Zeno stabilization~\cite{abdelkhalek2016energy}, and entropy production associated with rare jumps between Zeno subspaces~\cite{landi2021irreversible,santos2019roleofcoherence,Elouard2017}---in order to determine which of them materially limit the advantages of lubrication.   

Our analysis shows that many of the costs considered are either negligible in the regimes of interest or are not specific to the lubrication of the work strokes themselves. The main exception is the cost of implementing the strong-coupling drive through an external control field~\cite{zheng2016counterdiabaticcost}, which can become significant in realistic devices. This indicates that the protocol is most promising either in architectures where the relevant strong coupling arises naturally, or in platforms for which the control overhead can be kept sufficiently small. Even with this caveat, the protocol is of clear theoretical interest and motivates further the study of thermodynamics in Zeno regimes~\cite{gherardini2018nonequilibrium,gherardini2021thermalization,zambon2023boundseffectivethermalizationzeno,barontini2025quantumzeno}.

\begin{figure*}[t]
    \centering
    \includegraphics[width=0.8\textwidth]{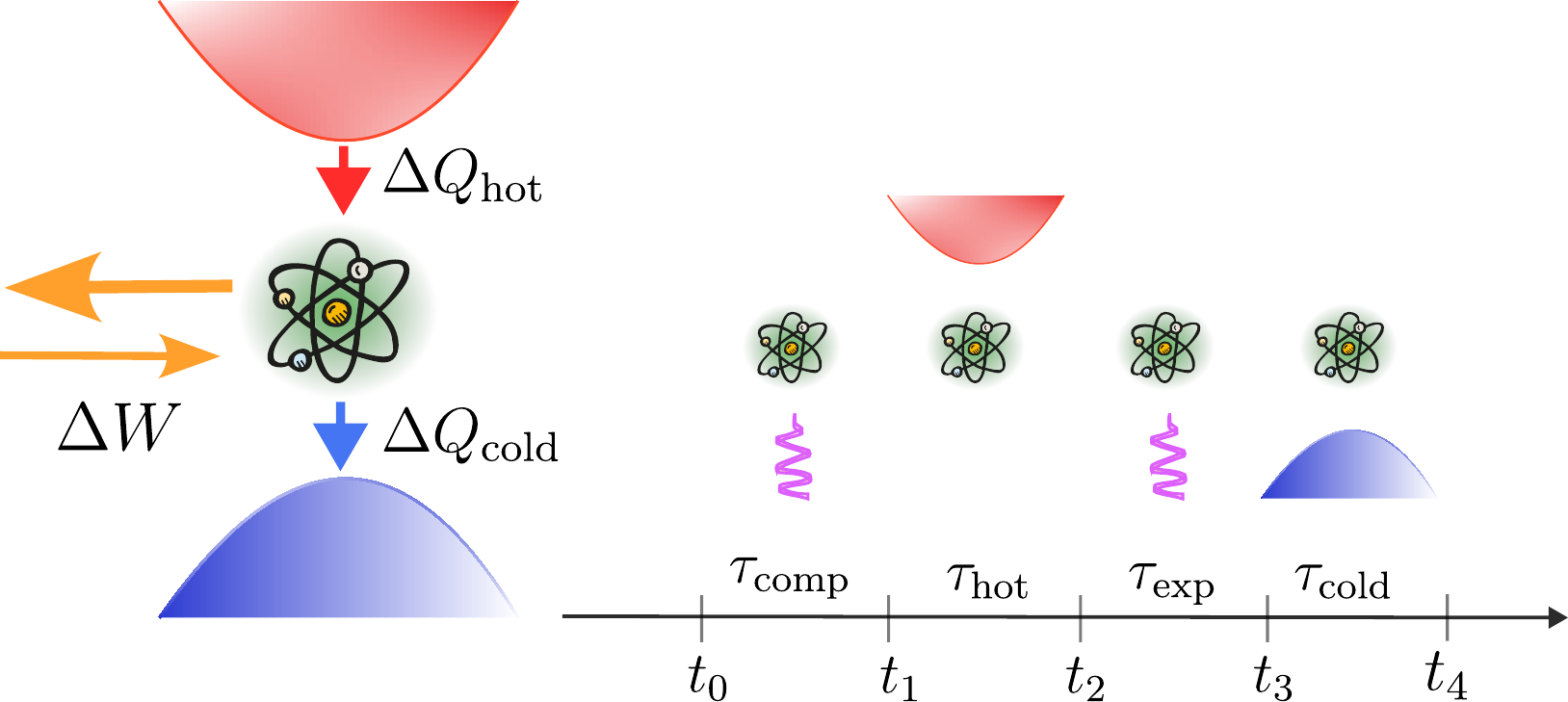}
    \caption{\textbf{Four-stroke quantum heat engine operating an Otto cycle.} (Left) A single working system $\mathcal{H}_S$  interacts alternately with two thermal baths, $\mathcal{H}_{B_{\rm h}}$ and $\mathcal{H}_{B_{\rm c}}$, at inverse temperatures $\beta_{\rm h}=T_{\rm h}^{-1}$ and $\beta_{\rm c}=T_{\rm c}^{-1}$, respectively. The labeled arrows indicate the average heat and work flows according to the sign convention adopted later in the text. (Right) The cycle consists of four strokes: an adiabatic compression, a hot isochoric thermalization, an adiabatic expansion, and a cold isochoric thermalization. During the work strokes the working medium is isolated from the baths and evolves unitarily; during each heat stroke it interacts with only one bath.} 
    \label{fig:QHE}
\end{figure*}

\textbf{Outline.} The structure of this work is as follows. In Sec.~\ref{sec:background_section} we describe the required background and notation for the presentation of our results in the following sections. The type of QHE we consider is presented in Sec.~\ref{sec:QHE_and_Otto_cycle}, and the known theory of efficiency, power, work extraction, optimality, and finite-time effects in these engines is reviewed in Sec.~\ref{sec:efficiency_power_and_finite_time}. These two sections can be skipped by readers with a solid background in quantum thermodynamics. We also describe the relevant background on the quantum Zeno effect in Sec.~\ref{sec:different_manifestations_zeno}, where we detail the known theory of the two manifestations of QZD that are relevant to us, namely those arising from frequent monitoring and strong coupling.

In Sec.~\ref{sec:Zeno_assisted_quantum_heat_engine} we begin presenting our results by describing the structure of our Zeno-assisted QHE. In that section we describe the device and specify how the Zeno-assistance takes place. We show explicitly how the Zeno-assistance provides an improvement via the emergence of a shortcut to adiabaticity. We then numerically verify the transitionless dynamics and show the effect of lubrication on efficiency, power, and extracted work, comparing finite-time non-lubricated and finite-time lubricated protocols.

Section~\ref{sec:relationship_between_decoherence_and_zeno} compares our Zeno-assisted method with the standard approach based on dephasing-assisted work strokes. In particular, we focus on the mechanism introduced in Ref.~\cite{weber2023thermodynamiccostspuredephasing}, as it allows for a direct comparison. In Sec.~\ref{sec:thermodynamic_costs}, we analyze the thermodynamic costs and limitations of our lubrication protocol. We conclude in Sec.~\ref{sec:discussion_and_outlook} with a discussion of our results, their relation to existing literature, and possible future research directions.

\section{Background}\label{sec:background_section}

\subsection{Quantum heat engines and the Otto cycle}\label{sec:QHE_and_Otto_cycle}

A \emph{quantum heat engine} is a quantum thermodynamic device operating a protocol capable of extracting work from the spontaneous heat flow between two baths at different temperatures. Standardly, a quantum thermodynamic heat device is defined by a multipartite quantum system, where each subsystem has a specific internal Hamiltonian, and the many subsystems have fixed allowed interactions. The QHE we consider is a device described by a tripartite quantum system 
\begin{equation}\label{eq:engines_system}
\mathcal{H}_{\rm QHE}:=\mathcal{H}_{B_{\rm c}}\otimes \mathcal{H}_S\otimes \mathcal{H}_{B_{\rm h}}.
\end{equation}
It is composed of two systems $\mathcal{H}_{B_{\rm h}}$ and $\mathcal{H}_{B_{\rm c}}$---referred to as \emph{heat baths}---in thermal states at temperatures $T_{\rm h}$ (`hot') and $T_{\rm c}$ (`cold') such that $T_{\rm h} $ is larger than $ T_{\rm c}$, and another system $\mathcal{H}_S$---called the \emph{working medium} (also known as  \emph{working system} or \emph{fluid}). This type of device is shown in Fig.~\ref{fig:QHE}. In standard QHEs the two baths never interact directly and energetic exchanges happen via the working medium viewed as an energetic mediator.

A \emph{protocol specification} is the description of the dynamical evolution of the engine system together with the specification of how one estimates energetics of the engine (i.e., as a specification of how work and heat are defined). We refer to the protocol as a \emph{cycle} whenever the specification is such that the system Hamiltonian of the working medium is cyclic, i.e. 
\begin{equation}
    H_S(t+\tau) = H_S(t),
\end{equation}
and we call $\tau$ the duration of the cycle. The same thermodynamic device can operate distinct protocols depending on the distinct specification of its dynamics and cycle operation (e.g. continuous dynamics~\cite{kosloff_quantum_2014} or a Carnot cycle~\cite{lee_Carnot_2017,bender2000quantummechanicalCarnot,geva_quantummechanical_1992,quan2007quantum}) and on the distinct specification of energetic regimes (e.g. engine or refrigerator regimes). 

In this work we consider QHEs operating an \emph{Otto cycle}~\cite{feldmann2003quantum,kieu2006quantum,feldmann2000performance,kosloff_discrete_2002}, a four-stroke protocol consisting of two work strokes and two heat strokes~\cite{solfanelli2020nonadiabatic}. We use the term \emph{stroke} for a stage of the cycle having a specified generator of the dynamics acting for a finite duration. During the work strokes the working medium is driven by a time-dependent Hamiltonian, while during the heat strokes its reduced dynamics is generated effectively by a bath-induced open system evolution. The distinction between work and heat strokes is therefore operational, even though each stroke is defined by its effective generator. Figure~\ref{fig:QHE} summarizes the four strokes considered here. Throughout this work we assume $\hbar = k_B = 1$ units.  

\subsubsection{Device specification and preparation}

We take the total Hilbert space of the engine (cf. Eq.~\eqref{eq:engines_system}) to be $$\mathcal{H}_{\rm QHE} = \bigotimes_k \mathcal{F}_k \otimes \mathbbm{C}^2 \otimes \bigotimes_k \mathcal{F}_k$$ so that the working medium is a single qubit $\mathcal{H}_S=\mathbbm{C}^2$~\cite{solfanelli2020nonadiabatic}, and each bath is modeled as an infinite-dimensional bosonic thermal reservoir $\mathcal{H}_{B_{\rm h}} = \mathcal{H}_{B_{\rm c}} =\bigotimes_k \mathcal{F}_k$, where  $\mathcal{F}_k$ is the bosonic Fock space associated with mode $k$.

The initial system Hamiltonian
is given by \begin{equation}\label{eq:initial_system_Hamiltonian}
H_{\rm c} = \frac{\omega}{2}Z = \frac{\omega}{2}\left(\vert 0 \rangle \langle 0 \vert - \vert 1 \rangle \langle 1 \vert\right)
\end{equation}
(the label ``c'' indicating ``cold'' will become clear shortly) where $Z$ denotes the third Pauli matrix (we will also use the first Pauli matrix $X$), and 
\begin{equation}\label{eq:initial_incoherent_state}
\rho_S(t_0) = p_0 \vert 0 \rangle \langle 0 \vert + p_1 \vert 1 \rangle \langle 1 \vert 
\end{equation}
is diagonal relative to the  eigenbasis of $H_{\rm c}$. For simplicity, we also assume that $\rho_S(t_0) = {e^{-\beta_{\rm c}H_{\rm c}}}/{\mathrm{Tr}[e^{-\beta_{\rm c}H_{\rm c}}]}$ is a thermal state of inverse temperature $\beta_{\rm c}$. 

We quantify work in the microscopic setting as an integral over the derivative of the time-dependent part of the system Hamiltonian $H_S(t)$,
\begin{equation}\label{eq:work}
\Delta W(t_f,t_i) :=\int_{t_i}^{t_f} \mathrm{d}s\, \mathrm{Tr}\left[\frac{\mathrm{d} H_{S}(s)}{\mathrm{d}s} \rho_S(s)\right].
\end{equation}
This expression is the standard microscopic definition of work used in the so-called \emph{semiclassical}~\cite{Dann2023} formulation of quantum thermodynamics. It follows from the first law at the level of expectation values, as developed by Pusz and Woronowicz~\cite{pusz1978passive}, Spohn and Lebowitz~\cite{Spohn1978}, and Alicki~\cite{alicki1979quantumopensystemheatengines}. In the same framework, the heat exchanged during the thermalization strokes is given by
\begin{equation}\label{eq:heat}
    \Delta Q(t_f,t_i) = \int_{t_i}^{t_f}\mathrm ds \,\mathrm{Tr}\left[H_{S}(s)\frac{\mathrm d \rho_S(s)}{\mathrm d s}\right],
\end{equation}
which attributes energy changes due to the state variation---rather than to explicit time dependence of the Hamiltonian---to heat.

\subsubsection{Four-stroke Otto cycle}

For our engine, the working qubit is driven by a cyclic time-dependent Hamiltonian given by a linear combination of the Pauli operators $X$ and $Z$, which, as we will see shortly, induces appreciable quantum friction when the working qubit is driven far from the quasistatic regime.~\footnote{Note that one \emph{can} drive a system with $\omega(t)Z$, in such a way that coherences are not generated and populations are not changed by this dynamics. Lubrication is only necessary whenever the drive is such that the time-dependent Hamiltonian does not commute at different times.} Let us choose $t_0 =0$ and modular variables $t=t \bmod \tau$ relative to the cycle time $\tau=t_4-t_0 \equiv t_4$ (as illustrated in Fig.~\ref{fig:QHE}). The total cyclic time-dependent system Hamiltonian $H_S(t)$ is therefore given by
\begin{equation}\label{eq:Hamiltonian_S}
    H_S(t) = \frac{\omega}{2}Z + \frac{f(t)}{2} X,
\end{equation}
where $X$ is the first Pauli matrix, and the time-dependent function $f(t)$ describes a linear drive acting on the qubit in the following way:
\begin{equation}\label{eq:function_f}
    f(t)=\left \{ \begin{array}{ll}   
    \Omega_0\frac{t}{t_1}& 0\, \leq t < t_1\\
    \Omega_0& t_1 \leq t < t_2\\ 
    \Omega_0\frac{t_3-t}{t_3-t_2}&t_2\leq  t < t_3\\
    0& t_3 \leq t < t_4.   \end{array}\right. 
\end{equation}
This definition naturally elucidates the duration of each stroke as seen in Fig.~\ref{fig:QHE}. In summary, the four-stroke engine runs in the following manner:
\begin{enumerate}
    \item \textbf{Adiabatic compression stroke:} In the first stroke we isolate the working medium from the baths and drive the qubit for the time $\tau_{\rm comp}=t_1$, thus \emph{increasing} its energy gap (see Fig.~\ref{fig:workstrokes}). 
    
    \item \textbf{Hot isochoric~thermalization~stroke:} Following the compression stroke we implement the first heat stroke by letting the working medium interact with the hot bath for $\tau_{\rm hot} = t_2-t_1$. The system Hamiltonian is maintained constant.
        
    \item \textbf{Adiabatic expansion stroke:} Once the hot thermalization is concluded, we uncouple the system from the hot bath and turn on the drive once more during the time interval $\tau_{\rm exp} = t_3-t_2$. Here, the qubit's energy gap is \emph{decreased} (see Fig.~\ref{fig:workstrokes}).
    
    \item \textbf{Cold isochoric thermalization stroke:} During the time interval $\tau_{\rm cold} = t_4-t_3$ the working system evolves similarly as during the hot thermalization stroke. We let the working medium interact with the cold bath, while the system Hamiltonian is, once more, maintained constant.
\end{enumerate}

\begin{figure}[t]
    \centering
    \includegraphics[width=0.35\textwidth]{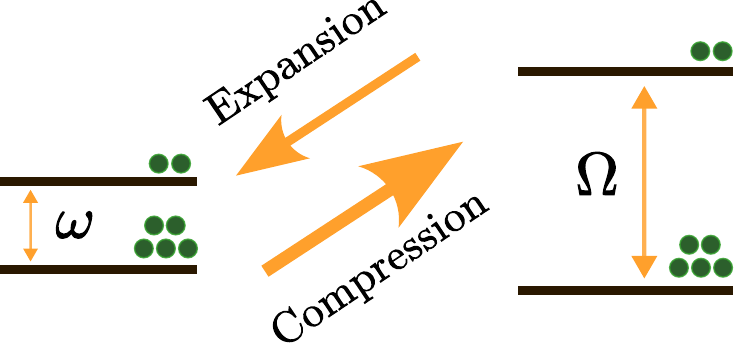}
    \caption{\textbf{Compression and expansion strokes.} During the engine cycle, the compression stroke begins with the system Hamiltonian $H_{\rm c}$ initialized at an energy gap $\omega$. The system Hamiltonian is then linearly driven to $H_{\rm h}$ such that the energy gap increases to $\Omega = \sqrt{\omega^2+\Omega_0^2} > \omega$. Ideally, the state remains unchanged in the instantaneous energy eigenbasis. This process extracts work from the system. The reverse process, in which the energy gap is reduced from $\Omega$ back to $\omega$, constitutes the expansion stroke and requires work to be performed on the system.}
    \label{fig:workstrokes}
\end{figure}

We provide further details on the full microscopic evolution during the compression and expansion strokes in Sec.~\ref{subsubsec:work_strokes}, as well as during the hot and cold thermalization strokes in Sec.~\ref{subsubsec:thermalisation_strokes}. Once the cold isochoric thermalization stroke is performed we re-initialize the compression stroke and continue implementing the same series of strokes. The total duration of the Otto cycle in our QHE is therefore 
\begin{equation}
\tau := \tau_{\rm comp} + \tau_{\rm hot} + \tau_{\rm exp} + \tau_{\rm cold}.
\end{equation}

\subsubsection{Work strokes}\label{subsubsec:work_strokes}

During the compression and expansion strokes  we isolate the working medium from the baths and drive the qubit such that it undergoes a unitary evolution
\begin{equation}
    U(t_f,t_i) = \mathcal{T}\left[\exp\left(-i\int_{t_i}^{t_f}H_S(s)\,\mathrm d s\right)\right],
\end{equation}
generated by the time-dependent Hamiltonian \eqref{eq:Hamiltonian_S}. This Hamiltonian can be rewritten as 
\begin{equation}
    H_S(t) = \frac{\varepsilon(t)}{2}R(t)
\end{equation}
where 
\begin{equation}\label{eq:eigenvalues_compression}
    \varepsilon(t) = \sqrt{\omega^2+ f^2(t)},
\end{equation}
and the diagonalized operator $R(t)$ in the (rotated) instantaneous eigenbasis $\{\vert 0_{t}\rangle, \vert 1_{t}\rangle \}$ is given by
\begin{equation}\label{eq:R(t)}
    R(t) = \cos(\theta_t)Z+\sin(\theta_t)X = \vert 0_t\rangle \langle 0_t\vert - \vert 1_t \rangle \langle 1_t \vert.
\end{equation}
Above, we have denoted the eigenvectors 
\begin{subequations}\label{eq:instantaneous_states}
\begin{align}
    \vert 0_t \rangle &= \cos\left(\frac{\theta_t}{2}\right)\vert 0 \rangle
    + \sin\left(\frac{\theta_t}{2}\right)\vert 1 \rangle, \label{eq:0t} \\
    \vert 1_t \rangle &= \sin\left(\frac{\theta_t}{2}\right)\vert 0 \rangle
    - \cos\left(\frac{\theta_t}{2}\right)\vert 1 \rangle, \label{eq:1t}
\end{align}
\end{subequations}
with the angle
\begin{equation}\label{eq:theta(t)}
    \theta_t=\arctan\left(\frac{f(t)}{\omega}\right),
\end{equation}
describing the rotation of working qubit's basis.

As mentioned, this Hamiltonian \emph{increases} the gap between energy levels in the working medium during the compression stroke, while \emph{decreasing} the gap during the expansion stroke. Therefore, it maps the system Hamiltonian from $H_{\rm c}$ to $H_{\rm h} = H_S(t_1)$ given by
\begin{equation}\label{eq:final_system_Hamiltonian}
H_{\rm h} = \frac{\omega}{2}Z+\frac{\Omega_0}{2}X = \frac{\Omega}{2}(\vert 0_{\rm h}\rangle \langle 0_{\rm h}\vert-\vert 1_{\rm h} \rangle \langle 1_{\rm h}\vert ),
\end{equation}
where $\{\vert 0_{\rm h}\rangle, \vert 1_{\rm h}\rangle \}$ is given in Eq.~\eqref{eq:instantaneous_states} with 
\begin{equation}
\theta_{\rm h} \equiv \theta_{t_1}=\arctan (\Omega_0/\omega), 
\end{equation}
and back from $H_{\rm h}$ to $H_{\rm c} = H_S(t_3)$, as seen in Fig.~\ref{fig:workstrokes}. For the cold eigenbasis, we simplify our notation by identifying $\{\vert 0_{\rm c}\rangle,\vert 1_{\rm c}\rangle\} \equiv \{\vert 0\rangle,\vert 1\rangle\}$. We use $\Omega$ (see Fig.~\ref{fig:workstrokes}) to denote the value of the energy gap $\varepsilon (t=t_1)$ at the end of the compression stroke (and at the beginning of the expansion stroke):
\begin{equation}
\varepsilon (t_1) \equiv \Omega = \sqrt{\omega^2+\Omega_0^2}.
\end{equation}

Using Eq.~\eqref{eq:work}, we can express the net work extracted from the system during the  compression stroke as (writing $\mathrm{Tr}[B\rho] \equiv \langle B \rangle_{\rho}$ for an arbitrary operator $B$)
\begin{align}
    \Delta W (t_1,0) &=\mathrm{Tr}\bigr[H_{\rm h}\rho_S(t_1)\bigr]-\mathrm{Tr}\bigr[H_{\rm c}\rho_S(0)\bigr] \nonumber \\
    &=\langle H_{\rm h}\rangle_{\rho_S(t_1)} - \langle H_{\rm c}\rangle_{\rho_S(0)}, \label{eq:work_compression_stroke}
\end{align}
and during the expansion stroke as 
\begin{align}
    \Delta W(t_3,t_2) &=\mathrm{Tr}\bigr[H_{\rm c}\rho_S(t_3)\bigr]-\mathrm{Tr}\bigr[H_{\rm h}\rho_S(t_2)\bigr] \nonumber \\
    &=\langle H_{\rm c}\rangle_{\rho_S(t_3)}-\langle H_{\rm h}\rangle_{\rho_S(t_2)} \label{eq:work_expansion_stroke}.
\end{align}
Because the working medium is isolated from the baths during these strokes, the evolution is unitary and no heat is exchanged (i.e. $\Delta Q_{\rm comp} = \Delta Q_{\rm exp} = 0$) with the environment. In that thermodynamic sense the stroke is \emph{adiabatic}.

\subsubsection{Thermalization strokes}~\label{subsubsec:thermalisation_strokes} 

During the hot and cold thermalization strokes, the energy gap of the working qubit is kept constant, having values $\Omega$ and $\omega$, respectively. The qubit interacts with one heat bath at a time, such that the effective evolution of the subsystem $\rho_S(t)$ is described by the Markovian master equation, whose solution admits the following form~\cite{rivas2012open,breuer2002theory} 
\begin{equation}
    \rho_S(t) = e^{(t-t_i)\mathcal{L}_{\rm x}}\rho_S(t_i),
\end{equation}
where $t_i \in \{t_1, t_3 \}$, $\mathrm{x}\in \{\rm h, c\}$, and the Lindblad superoperator $\mathcal{L}_{\rm x}$ which generates the effective dynamics is given by~\cite{weber2023thermodynamiccostspuredephasing}
\begin{align}
    \mathcal{L}_{\rm x}[\rho_S] = \gamma_{\rm x} n_{\rm x} \left(\sigma^+_{\rm x}\rho_S\sigma^-_{\rm x}-\frac{1}{2}\{\sigma_{\rm x}^-\sigma_{\rm x}^+,\rho_S\}\right)\nonumber \\
    +\gamma_{\rm x} (n_{\rm x}+1)\left(\sigma^-_{\rm x}\rho_S\sigma^+_{\rm x}-\frac{1}{2}\{\sigma_{\rm x}^+\sigma_{\rm x}^-,\rho_S\}\right), \label{eq:thermalization_Lindbladian}
\end{align}
with $\sigma^{+}_{\rm x} := \vert 0_{\rm x}\rangle \langle 1_{\rm x}\vert$ and $\sigma_{\rm x}^-=(\sigma_{\rm x}^+)^{\dagger} = \vert 1_{\rm x}\rangle \langle 0_{\rm x}\vert$ being raising and lowering operators relative to the eigenbasis $\{\vert 0_{\rm x} \rangle, \vert 1_{\rm x} \rangle\}$ of $H_{\rm x}$. Above, $n_{\rm x} = 1/(e^{\varepsilon_{\rm x}\beta_{\rm x}}-1)$ is the mean thermal photon number characterizing each thermal bath, where $\varepsilon_{\rm h}= \Omega$ and $\varepsilon_{\rm c}= \omega$, and $\gamma_{\rm x}$ is the dissipative coupling rate between the bath and the system. Note that $$\rho_{\beta_{\rm x}}^{\rm th} = \frac{e^{-\beta_{\rm x}H_{\rm x}}}{\mathrm{Tr}[e^{-\beta_{\rm x}H_{\rm x}}]}$$ is a steady state for $\mathcal{L}_{\rm x}$. According to this coupling, the system relaxes exponentially fast towards that steady state. 

During the heat strokes the working medium interacts dissipatively with a thermal reservoir while the system Hamiltonian is held fixed. We therefore treat these strokes as \emph{isochoric}: the Hamiltonian does not change in time, so the work contribution vanishes in Eq.~\eqref{eq:work}, leading to the energy exchange during hot, 
\begin{equation}\label{eq:hot_heat_stroke}
    \Delta Q(t_2,t_1) = \mathrm{Tr}\Bigr[H_{\rm h}\bigr(\rho_S(t_2)-\rho_S(t_1)\bigr)\Bigr] \equiv \Delta Q_{\rm hot},
\end{equation}
and cold, 
\begin{equation}\label{eq:cold_heat_stroke}
    \Delta Q(t_4,t_3)=\mathrm{Tr}\Bigr[H_{\rm c}\bigr(\rho_S(t_4)-\rho_S(t_3)\bigr)\Bigr] \equiv \Delta Q_{\rm cold},
\end{equation}
thermalization entirely accounted for as heat.

To conclude, we describe our sign convention: When we \emph{extract work} from the system we have $\Delta W<0$, and therefore when we \emph{input work} on the system we have $\Delta W>0$. With this sign convention, during the compression stroke we \emph{extract work}. This may appear counterintuitive because, in many  conventions, compression requires work input. The difference here comes from the symmetric spectrum $\left\{\pm \varepsilon (t)/2 \right\}$ of the driven Hamiltonian. As the gap increases, the excited level moves upward while the ground level moves downward (see Fig.~\ref{fig:workstrokes}). Since the state is initially cold and therefore mostly supported on the ground state, the internal energy of the working medium decreases during compression, which in our convention corresponds to extracted work.~\footnote{If one considers an asymmetric description of the Hamiltonian where the ground state is taken to have energy $\varepsilon_1 = 0$ with the excited state having energy $\varepsilon_0 = \varepsilon(t)$ the Hamiltonian becomes $H'(t) = \varepsilon(t) \vert 0_t \rangle \langle 0_t \vert$ and one extracts work during the expansion stroke instead.} 

\subsubsection{Total work and total heat}

Following from the two strokes, we have a total change in work given by
\begin{align}\label{eq:total_work_implicit}
    \Delta W_{\rm tot} &= \Delta W(t_1,t_0)+\Delta W(t_3,t_2),
\end{align}
and therefore, using Eqs.~\eqref{eq:work_compression_stroke} and~\eqref{eq:work_expansion_stroke}, 
\begin{align}\label{eq:total_work_explicit}
    \Delta W_{\rm tot} =&+ \langle H_{\rm h}\rangle_{\rho_S(t_1)}-\langle H_{\rm c}\rangle_{\rho_S(0)}\nonumber \\&+\langle H_{\rm c}\rangle_{\rho_S(t_3)}-\langle H_{\rm h}\rangle_{\rho_S(t_2)}.
\end{align}
Moreover, we have two isochoric heat thermalization strokes yielding a total heat exchange of 
\begin{equation}
    \Delta Q_{\rm tot} = \Delta Q(t_2,t_1)+\Delta Q(t_4,t_3).
\end{equation}
Using Eqs.~\eqref{eq:hot_heat_stroke} and~\eqref{eq:cold_heat_stroke} we conclude that
\begin{align}
    \Delta Q_{\rm tot} = &+\langle H_{\rm h} \rangle_{\rho_S(t_2)} -\langle H_{\rm h} \rangle_{\rho_S(t_1)} \nonumber\\
    &+\langle H_{\rm c} \rangle_{\rho_S(t_4)} -\langle H_{\rm c} \rangle_{\rho_S(t_3)}.
\end{align}

\subsection{Efficiency and power of the Otto cycle}\label{sec:efficiency_power_and_finite_time}

The \emph{power} of  an engine running a cyclic protocol of duration $\tau$ is quantified by the ratio of total output work to the cycle duration, i.e.
\begin{equation}\label{eq:power_cyclic_QHE}
    P_{\rm tot} = \frac{-\Delta W_{\rm tot}}{\tau},
\end{equation}
and the \emph{efficiency} of the engine is given by the ratio of the net work extracted to the heat absorbed during the hot thermalization stroke (which can be viewed as the amount of average heat that flows from the hot bath to the working medium that can later be extracted as work). Succinctly, 
\begin{equation}\label{eq:efficiency_cyclic_QHE}
    \eta = \frac{-\Delta W_{\rm tot}}{\Delta Q_{\rm hot}}.
\end{equation}
Recall that by our convention, we extract work when $\Delta W < 0$, therefore, we have included a minus sign for having power $P_{\rm tot}>0$ and $\eta > 0$ whenever we extract more work than we inject (assuming that $\Delta Q_{\rm hot} > 0$).

\subsubsection{Ideal Otto  efficiency}

Let us assume that the working medium fully thermalizes during the isochoric strokes, reaching the thermal states 
\begin{equation}
    \rho_{S}(0) = \frac{e^{-\beta_{\rm c}H_{\rm c}}}{\mathrm{Tr}[e^{-\beta_{\rm c}H_{\rm c}}]} \quad \text{and}\quad \rho_{S}(t_2) = \frac{e^{-\beta_{\rm h}H_{\rm h}}}{\mathrm{Tr}[e^{-\beta_{\rm h}H_{\rm h}}]}.
\end{equation}
\emph{Perfect} thermalization in our model occurs only in the limit $\tau_{\rm hot}, \tau_{\rm cold} \to \infty$. Let us also assume that during the adiabatic work strokes the energy gap changes quasistatically such that no coherences or excitations are created. In this case, it is straightforward to show that 
\begin{equation}\label{eq:ideal_otto_compression}
    \Delta W(t_1,0) = -\frac{\Omega-\omega}{2}\,\tanh\left(\frac{\omega}{2T_{\rm c}}\right)
\end{equation}
and
\begin{equation}\label{eq:ideal_otto_expansion}
    \Delta W(t_3,t_2) = \frac{\Omega-\omega}{2}\,\tanh\left(\frac{\Omega}{2 T_{\rm h}}\right). 
\end{equation}
From Eqs.~\eqref{eq:ideal_otto_compression} and~\eqref{eq:ideal_otto_expansion} we find the total extracted work to be
\begin{equation}
    \Delta W_{\rm tot} = \frac{\Omega-\omega}{2}\left(\tanh\left(\frac{\Omega}{2 \, T_{\rm h}}\right)-\tanh\left(\frac{\omega}{2 \, T_{\rm c}}\right)\right).
\end{equation}

The total heat exchange during the hot thermalization stroke is given by
\begin{equation}
    \Delta Q_{\rm hot} = \frac{\Omega }{2}\left(\tanh\left(\frac{\omega}{2 \, T_{\rm c}}\right)-\tanh\left(\frac{\Omega}{2 \, T_{\rm h}}\right)\right),
\end{equation}
and for the cold thermalization stroke is 
\begin{align}
\Delta Q_{\rm cold} &= -\frac{\omega}{2}\left(\tanh\left(\frac{\omega}{2 \, T_{\rm c}}\right)-\tanh\left(\frac{\Omega}{2 \, T_{\rm h}}\right)\right)\\&=-\frac{\omega}{\Omega}\Delta Q_{\rm hot}.
\end{align}
From that, we can calculate the (ideal) Otto efficiency to be
\begin{equation}\label{eq:Otto_eff_def}
    \eta_{\mathrm{Otto}} = 1-\frac{\omega}{\Omega}. 
\end{equation}
The Otto efficiency is the optimal efficiency for the engine we have described. 

Assuming that $\Delta W_{\rm tot} < 0$ (we extract work), which is obtained whenever 
\begin{align}
    \tanh\left(\frac{\Omega}{2T_{\rm h}}\right)<\tanh\left(\frac{\omega}{2T_{\rm c}}\right) \iff \frac{\omega}{\Omega}>\frac{T_{\rm c}}{T_{\rm h}},
\end{align}
we can show that the Otto efficiency is upper bounded by Carnot's efficiency, since  
\begin{equation}\label{eq:optimal_eff}
    \eta_{\mathrm{Otto}} = 1-\frac{\omega}{\Omega} \leq 1-\frac{T_{\rm c}}{T_{\rm h}} = \eta_{\rm Carnot}.
\end{equation}

Above, we have assumed that the drive is implemented quasistatically, which then implies that $\tau_{\rm comp}, \tau_{\rm exp} \to \infty$. Note that, in this case, while the efficiency is optimal, the power captured by Eq.~\eqref{eq:power_cyclic_QHE} goes to zero. In what follows we will start by focusing on the imperfect (finite-time) drive during the work strokes. We will later relax the assumption of perfect thermalization (see also Sec.~\ref{subsec:imperfect_therm}).  

\subsubsection{Effect of finite-time drive on efficiency and power}

An ideal work stroke should change the energy gap without generating additional entropy or redistributing populations in the instantaneous energy basis. In the present setting this means that the dynamics should remain unitary and, at the same time, avoid creating coherences in the instantaneous eigenbasis of the driven Hamiltonian. For example, during compression we would ideally realize
\begin{equation}\label{eq:condition_goal_1}
    p_0 \vert 0 \rangle \langle 0 \vert + p_1 \vert 1 \rangle \langle 1 \vert \mapsto p_0\vert 0_{\rm h} \rangle \langle 0_{\rm h} \vert + p_1 \vert 1_{\rm h} \rangle \langle 1_{\rm h} \vert,
\end{equation}
and analogously during expansion. Finite-time driving generally violates this condition because the Hamiltonian changes too quickly for the populations to follow the instantaneous eigenbasis without transitions. 

More explicitly, whenever the drive is not quasistatic, the coupling $\langle 0_t\vert \dot 1_t \rangle =\sfrac{\dot \theta_t}{2}$ is not approximately zero and the work output is significantly reduced by coherence-generated friction. We now derive the expression for this reduction. Starting from Eq.~\eqref{eq:work} and using the parametrization introduced above, one finds the work
\begin{align}
     \Delta W(t_1,0) &= \int_0^{t_1} \frac{\dot{\varepsilon} (t)}{2}(\,\rho_{0_t0_t}-\rho_{1_t1_t})\mathrm dt  \nonumber \\&-\int_0^{t_1}  \frac{\omega \Omega_0 \mathfrak{Re}\!\big[\rho_{0_t1_t}]} {\sqrt{(\omega \tau_{\rm comp})^2 + (\Omega_0 t)^2} }  \mathrm dt \label{eq:compression_non_quasi_static}.
\end{align}
Above, $\mathfrak{Re}[x+\mathrm{i}y]=x$ and $\rho_{i_tj_t} = \langle i_t \vert \rho_S(t)\vert j_t \rangle$, where $i,j\in \{0,1\}$. The second integral in that expression isolates the contribution
\begin{equation}
    \delta W_{\rm comp}^{\rm (fric)}(t) = -\frac{\omega \Omega_0 \mathfrak{Re}\!\big[\rho_{0_t1_t}]} {\sqrt{(\omega \tau_{\rm comp})^2 + (\Omega_0 t)^2} }
\end{equation}
associated with the coherences generated in the instantaneous energy basis. For the expansion stroke, the analogous coherence-contribution term is
\begin{equation}
    \delta W_{\rm exp}^{\rm (fric)} = \frac{\omega \Omega_0 \mathfrak{Re}\left[ \rho_{ 0_t  1_t} \right]}{\sqrt{(\omega \tau_{\rm exp})^2 + (\Omega_0 (\tau_{\rm exp} - t + t_2))^2}},
\end{equation}
where $\rho_{ 0_t 1_t} = \langle  0_t \vert \rho_{S}(t)\vert  1_t\rangle$. The derivation of these can be found in Appendix~\ref{app:friction_terms_work}.

\subsubsection{Effect of finite-time thermalization on efficiency and power}

Finite-time thermalization can also limit performance. For \emph{endoreversible}~\footnote{Endoreversible engines are composed of subsystems that interact reversibly internally, while exchanging energy irreversibly with their surroundings.} classical heat engines~\cite{hoffmann2005Endoreversible}, finite thermal conductance prevents the working medium from equilibrating perfectly with the reservoirs, which leads to well-known bounds such as the Curzon--Ahlborn efficiency~\cite{curzon1975efficiency} 
\begin{equation}
    \eta_{\rm max} \leq \eta_{\rm CA}=1-\sqrt{\frac{T_{\rm c}}{T_{\rm h}}},
\end{equation}
where $\eta_{\rm max}$ is the efficiency of the engine when running at maximum power. In the quantum setting studied here, the same basic issue appears: If the thermalization strokes are too short, the working medium fails to approach the relevant Gibbs state closely enough, and both efficiency and power can be affected.

The effective master equations we consider for the thermalization strokes (Eq.~\eqref{eq:thermalization_Lindbladian}) are derived under the usual approximations~\cite{breuer2002theory,rivas2012open} (Born--Markov, weak coupling, rotating wave approximation, etc.). Under these assumptions, the reduced system relaxes exponentially fast to the Gibbs state of the Hamiltonian (either relative to $H_{\rm c}$ during the cold stroke or $H_{\rm h}$ during the hot stroke), and the Gibbs state is the unique steady state. The perfect thermalization holds only asymptotically, but the distance to the Gibbs state typically decays exponentially with the spectral gap of the generator.

As a first step, we will focus on losses generated during the work strokes, because for the parameter regime emphasized later in Sec.~\ref{sec:Zeno_assisted_quantum_heat_engine} the dominant finite-time effect is the coherence generated by fast driving. This simplification is meant only to isolate the lubrication mechanism more clearly; it should not be read as a general claim that imperfect thermalization is always negligible. We return to the thermalization bottleneck in Sec.~\ref{sec:thermodynamic_costs} and analyze explicitly when it becomes the dominant limitation.

\subsection{Manifestations of the quantum Zeno dynamics}\label{sec:different_manifestations_zeno}

The quantum Zeno effect~\cite{misra1977zeno,greenfield_unified_2025, vonNeumann1932mathematische,facchi2008quantum} is commonly characterized by the total inhibition of dynamical evolution due to frequent monitoring via strong measurements. Fundamentally, it presents an intrinsic tension between the two possible evolutions predicted by quantum theory: while the system attempts to evolve unitarily, thus changing its state, a constant monitoring of the system is described via the measurement postulate which, then, with high probability, avoids any evolution, freezing the system to a single state. 

In the quantum Zeno \emph{dynamics}~\cite{machida1999reflection,facchi2000quantumzeno,facchi2008quantum,hacohen2018incoherent,burgarth2014exponential,signoles2014confined,bretheau2015quantum}, instead of freezing the system, its implementation can dynamically change the state of the system when restricted to a degenerate subspace defined by the measurement apparatus. Such subspaces are known as \emph{Zeno subspaces}~\cite{facchi2002zeno}. 

In what follows, it will be important for our lubrication protocol to consider two different manifestations of the QZD: the standard one that follows from frequent projective measurements and a continuous version arising from strong coupling between two systems.

\subsubsection{Frequent projective measurements}

Let $\{P_\ell\}_{\ell}$ be a projection-valued measure (PVM), and denote $$\mathcal{H}_{P_\ell} := \text{Range}(P_\ell) = \left\{P_\ell\vert \psi \rangle \mid \vert \psi \rangle \in \mathcal{H}\right\}$$ the relative subspace each Hermitian projector $P_\ell$ projects onto. This PVM induces a partition $\mathcal{H} = \bigoplus_\ell \mathcal{H}_{P_\ell}.$ Each subspace $\mathcal{H}_{P_\ell}$ is called a \emph{Zeno subspace}. For our purposes, we can focus on the simplest case $M=\{P_1,P_2\}$, with $P_1P_2=0$ and $\mathcal{H} = \mathcal{H}_{P_1} \oplus \mathcal{H}_{P_2}$. Suppose we have operationally the situation provided in Fig.~\ref{fig:zeno_drive}. We consider a drive $H(t)$ which generates the evolution $t \mapsto U(t)$ given by
\begin{equation}
    U(t_f,t_i) = \mathcal{T}\left[\exp \left(-i\int_{t_i}^{t_f} \mathrm{d}s \,H(s)\right)\right].
\end{equation}
Now, we let $t_i \equiv t_{(1)}$, $t_f \equiv t_{(n)}$ and divide the evolution so that $t_f-t_i = \sum_{k=1}^{n-1} \delta t_{(k)}$ where $\delta t_{(k)} = t_{(k+1)}-t_{(k)} = (t_{f}-t_i)/(n-1)\sim (t_f-t_i)/n$. We alternate between short unitary pulses $U(t_{(k+1)},t_{(k)})$ and dichotomic measurements $M$. If we assume that we start within the $P_\ell$ subspace $\mathcal{H}_{P_\ell}$, s.t. 
\begin{equation}
\rho(t_{(1)}) = P_\ell\, \rho(t_{(1)}) \,P_\ell,
\end{equation}
the effective evolution \emph{within this subspace}, in the limit of $n \to \infty$, is provided by a unitary
\begin{equation}
    U_{\rm Zeno}^{(\ell)}(t_f,t_i) = \mathcal{T}\left[\exp \left(-i\int_{t_i}^{t_f} \mathrm{d}s\, H^{(\ell)}_{\rm Zeno}(s)\right)\right]
\end{equation}
where
\begin{equation}
    H^{(\ell)}_{\rm Zeno}(t) := P_\ell \,H(t) \, P_\ell
\end{equation}
is the so-called (time-dependent) \emph{Zeno Hamiltonian}.
\begin{figure}[t]
    \centering
    \includegraphics[width=\columnwidth]{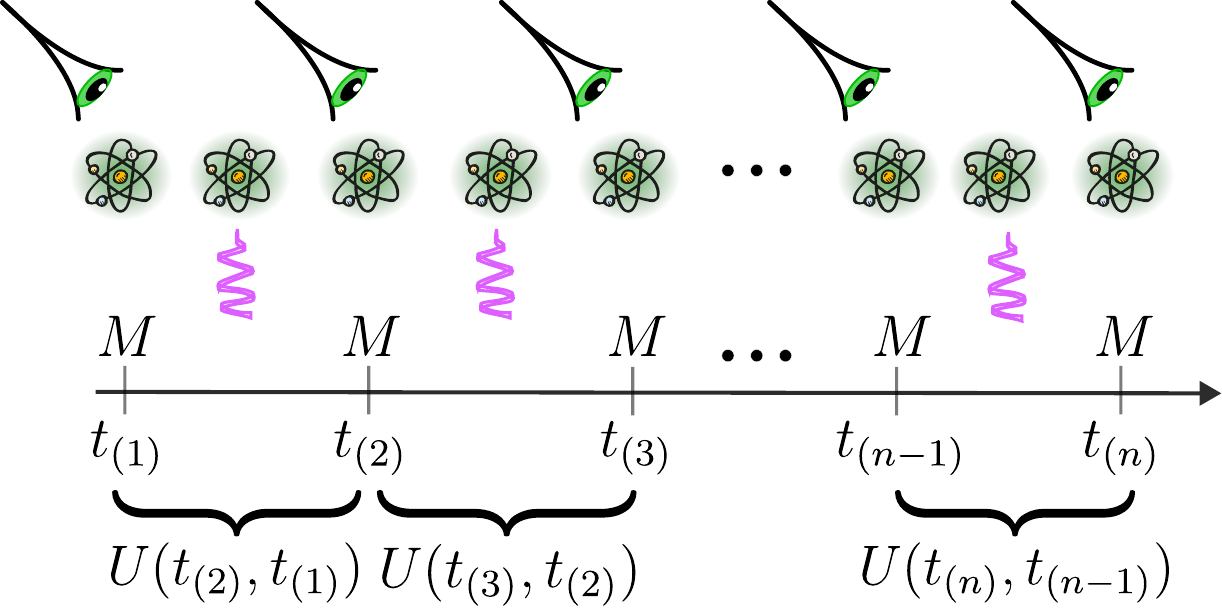}
    \caption{\textbf{ Quantum Zeno drive.} An interval $0\leq t_{(1)} \leq t \leq t_{(n)} $ is subdivided into $n-1$ intervals of duration $\delta t_{(k)} = t_{(k+1)}-t_{(k)}$, for $k = 1,\ldots,n-1$. At each instant $t_{(k)}$ a projective measurement $M$ is applied to the system. During the intervals $\delta t_{(k)}$ we drive the system according to a time-dependent Hamiltonian $H(t)$ yielding a unitary $U(t_{(k+1)},t_{(k)})$. These are sometimes referred to as \emph{unitary kicks} or \emph{pulses}~\cite{facchi2008quantum}. We thus alternate between unitary kicks of duration $\delta t_{(k)}$ and instantaneous projective measurements $M$. We refer to this as a time-dependent \emph{Zeno-drive}.}
    \label{fig:zeno_drive}
\end{figure}

\subsubsection{Strong coupling}

The effective dynamics can also be governed by a Zeno Hamiltonian in a distinct regime, where Zeno behavior emerges from \emph{strong coupling} rather than from frequent projective  measurements~\cite{schulman1998continuous,facchi2002zeno,burgarth2019generalized,burgarth2022oneboundtorulethem}. Let us consider a situation where we have a time-dependent Hamiltonian given by
\begin{equation}
    H_\Gamma(t) = \Gamma H_0(t)+H_1(t)
\end{equation}
and assume that $\Gamma \to \infty$. In this limit, the evolution is approximately given by~\cite{burgarth2022oneboundtorulethem}
\begin{equation}\label{eq:strong_coupling_effective}
    U_{\rm eff}(t) = \mathcal{T}\left[\exp\left(-i\int_0^t\mathrm{d}s \, H_{\rm eff}(s)\right)\right]
\end{equation}
generated by the effective Hamiltonian 
\begin{equation}
    H_{\rm eff}(t) := \Gamma  H_0(t)+A(t)+H_{\rm Zeno}^{(\rm str)}(t)
\end{equation}
where $A(t)$ is  the \emph{generator of the quasistatic transporter}~\cite{kato1950ontheadiabatic} defined by 
\begin{equation}\label{eq:adiabatic_transporter}
    A(t):=\frac{i}{2}\sum_\ell [\dot P_\ell^{(H_0)}(t),P_\ell^{(H_0)}(t)],
\end{equation}
$\{P_\ell^{(H_0)}(t)\}_{\ell}$ is the instantaneous eigenbasis of $H_0(t)$, and 
\begin{equation}\label{eq:strong_coupling_limit}
    H_{\rm Zeno}^{\rm (str)} (t) := \sum_\ell P_\ell^{(H_0)}(t) H_1(t)P_\ell^{(H_0)}(t)
\end{equation}
is the time-dependent Zeno Hamiltonian relative to $\{P_\ell^{(H_0)}(t)\}_\ell$. We refer to  Theorem~\ref{theorem:strong_coupling_theorem} in Appendix~\ref{app:generalization} for the precise statement of this result.

\section{Zeno-assisted quantum heat engines}\label{sec:Zeno_assisted_quantum_heat_engine}

In this section we introduce a variant of the engine from Sec.~\ref{sec:QHE_and_Otto_cycle} in which the working medium is augmented by an auxiliary system: the \emph{lubricant}. The resulting device still operates an Otto cycle, but the work strokes are now implemented through a Zeno-assisted protocol acting on the joint system of working medium and lubricant. The heat strokes remain essentially unchanged. As we show later in Sec.~\ref{sec:relationship_between_decoherence_and_zeno}, this construction also provides a useful point of contact with existing dephasing-assisted lubrication schemes.

\subsection{Lubricated device specification and preparation}

We start by engineering our device to be the system (cf.~Eq.~\eqref{eq:engines_system}) 
\begin{equation}
    \mathcal{H}_{\rm QHE} = \bigotimes_k \mathcal{F}_k\otimes \mathbbm{C}^2 \otimes \mathbbm{C}^2  \otimes \bigotimes_k \mathcal{F}_k
\end{equation}
so that the working medium $\mathcal{H}_S = \mathbbm{C}^2$ is augmented by a qubit lubricant  $\mathcal{H}_L = \mathbbm{C}^2$ (see Fig.~\ref{fig:QHE_zeno_assisted}).  As before, the two heat baths are given by infinite-dimensional bosonic thermal reservoirs where each $\mathcal{F}_k$ labels an infinite-dimensional bosonic Fock space.

We prepare the initial state of the whole system at time $t_0 = 0$ to be a product quantum state
\begin{align}
    \rho_{\rm QHE}(0) = \rho_{B_{\rm c},\beta_{\rm c}}^{\rm th}\otimes \rho_{SL}(0)\otimes \rho_{B_{\rm h},\beta_{\rm h}}^{\rm th},
\end{align}
where $\rho_{B_{\rm c},\beta_{\rm c}}^{\rm th}$ and $\rho_{B_{\rm h},\beta_{\rm h}}^{\rm th}$ are the thermal states for the cold and hot baths, respectively, and $\rho_{SL}(0)$ is given by
\begin{equation}\label{eq:initial_state_lubricated_engine}
    \rho_{SL}(0) = \bigr(p_0 \vert 0 \rangle \langle 0 \vert  + p_1\vert 1 \rangle \langle 1 \vert \bigr) \otimes  \vert + \rangle \langle + \vert,
\end{equation}
where $0\leq p_0\leq 1$ and $p_1=1-p_0$. The initial system Hamiltonian is $H_S(0) = \frac{\omega}{2}Z$ as in Eq.~\eqref{eq:initial_system_Hamiltonian}, and we let the lubricant's local Hamiltonian be time-independent and given by 
\begin{equation}\label{eq:lubricant_hamiltonian}
    H_L = \frac{\omega_L}{2}Z.
\end{equation}

\subsection{Adiabatic Zeno-assisted work strokes}

We implement the work strokes by isolating the composite system $\mathcal{H}_{SL} = \mathbbm{C}^2 \otimes \mathbbm{C}^2$ from the baths and subjecting it to a Zeno-drive described by a time-dependent Hamiltonian $H_{\rm tot}(t)$ and a local \emph{selective and non-destructive measurement} relative to the $X$ basis $\{\vert + \rangle, \vert - \rangle \}$ on the lubricant system. The total time-dependent Hamiltonian $H_{\rm tot}(t)$ is given by
\begin{equation}\label{eq:total_hamiltonian_SL}
    H_{\rm tot}(t) = H_S(t)\otimes \mathbbm{1}+\mathbbm{1}\otimes H_L+H_{SL}(t)
\end{equation}
where $H_S(t)$ is the cyclic Hamiltonian given by Eq.~\eqref{eq:Hamiltonian_S}  and $H_L$ by Eq.~\eqref{eq:lubricant_hamiltonian}. The time-dependent interaction Hamiltonian $H_{SL}(t)$ is given by
\begin{equation}\label{eq:H_SL(t)}
    H_{SL}(t) = \Gamma(t) \,R(t) \otimes X,
\end{equation}
where $R(t)$ is given by Eq.~\eqref{eq:R(t)}, and~\footnote{For simplicity, we choose interaction strength $\Gamma$ to be the same during compression and expansion strokes. However, it may be useful to consider $\Gamma_{\rm exp} \neq \Gamma_{\rm comp}$, since one of the work strokes may generate less coherence than the other and therefore require weaker lubrication.}
\begin{equation}
    \Gamma(t)=\left \{ \begin{array}{ll}
        
    \Gamma& \quad 0\, \leq t \leq  t_1 \, \cup \, t_2 \leq t \leq  t_3 \\
    0& \quad  \rm otherwise.  \end{array}\right. 
\end{equation}

\begin{figure}[t]
    \centering
    \includegraphics[width=\columnwidth]{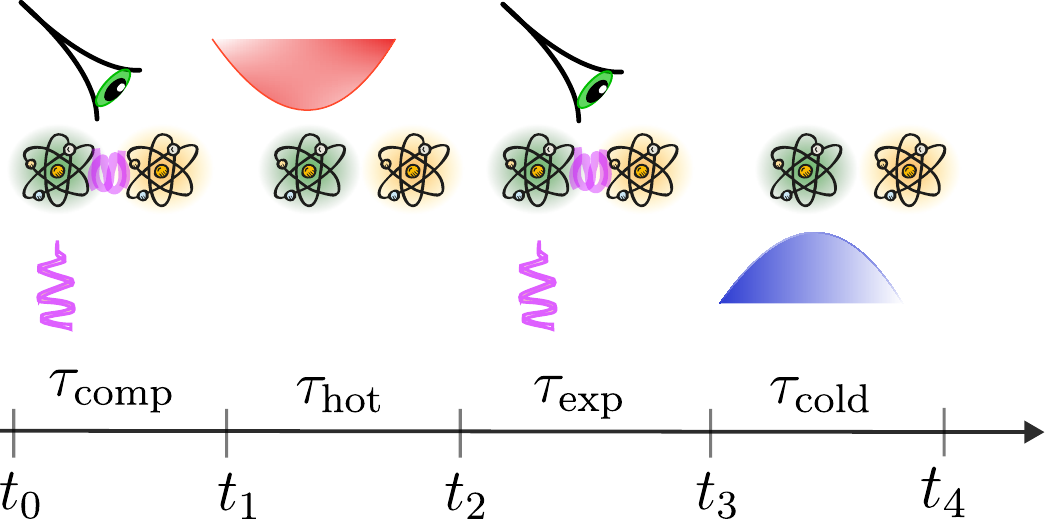}
    \caption{\textbf{Four-stroke Zeno-assisted quantum heat engine.} The engine shown in Fig.~\ref{fig:QHE} is modified by adding an auxiliary single-qubit system $\mathcal{H}_L=\mathbbm{C}^2$ to the working medium, referred to as the quantum lubricant. The engine is operated according to a four-stroke Otto cycle in which the unitary strokes are replaced by those generated by a Zeno-drive: strong time-dependent coupling between system and lubricant, together with frequent monitoring of the lubricant's state. The heat strokes remain unchanged, with only the system $\mathcal{H}_S$ brought into contact with the baths.}\label{fig:QHE_zeno_assisted}
\end{figure}

We now proceed to show how the QZD can lead to an effective shortcut to adiabaticity. Our full Zeno-assisted lubrication scheme will combine two manifestations of the QZD. The first ingredient is a strong-coupling regime,  $\Gamma \to \infty$ between the working medium and the lubricant. The second is a sequence of frequent selective measurements of the lubricant in the $X$ basis. The key point is that the strong coupling generates an effective Hamiltonian containing the quasistatic transporter, while the monitoring confines the dynamics to a suitable Zeno subspace. Throughout,  the \emph{ideal Zeno regime}, or equivalently the \emph{Zeno limit},  refers to the asymptotic parameter region where the monitoring interval is short compared with all intrinsic dynamical timescales and $\Gamma$ is the dominant energy scale.

We remark that another approach is to drive the system using frequent time-dependent (adaptive) non-destructive measurements $P(t)$ on the system~\cite{barontini2025quantumzeno}, an alternative that we will refer to as \emph{Zeno-dragging}. We shall comment on this alternative later on in Sec.~\ref{sec:discussion_and_outlook}.

Introducing both strong coupling and frequent measurements substantially changes the thermodynamic setting of the engine. In particular, the monitoring can generate fluctuations and entropy production, and the strong interaction can modify the energetics of work extraction. These issues are deferred to  Sec.~\ref{sec:thermodynamic_costs}. For now, we focus on the control-theoretic question of how the Zeno mechanism yields a shortcut to adiabaticity for the working medium.

\subsubsection{Strong coupling between working medium and lubricant}

Let us start with the compression stroke by calculating the effective strong-coupling unitary as given by Eq.~\eqref{eq:strong_coupling_effective}. Note that $R(t) \otimes X$ has an instantaneous eigendecomposition given by the product of those from $R(t)$ and $X$, thus we can write
\begin{align}
    P_+(t) &= \vert 0_t \rangle \langle 0_t \vert \otimes \vert+ \rangle \langle + \vert + \vert 1_t\rangle \langle 1_t \vert \otimes \vert - \rangle \langle - \vert \nonumber \\
    &\equiv  P_{0}(t) \otimes P_+ + P_{1}(t)\otimes P_-, \label{eq:P_+_time_dependent}\\    
    P_-(t) &=\vert 0_t \rangle \langle 0_t \vert \otimes \vert - \rangle \langle - \vert + \vert 1_t \rangle \langle 1_t \vert \otimes \vert + \rangle \langle + \vert \nonumber \\ 
    &\equiv P_0(t)\otimes P_-+P_1(t)\otimes P_+, \label{eq:P_-_time_dependent}
\end{align}
where $\vert 0_t \rangle$ and $\vert 1_t \rangle$ are the instantaneous eigenstates of $R(t)$ given by Eq.~\eqref{eq:instantaneous_states}.  The transporter $A(t)$ from Eq.~\eqref{eq:adiabatic_transporter} in this case becomes
\begin{align}
    A(t) &= \frac{i}{2}\big([\dot P_+(t),P_+(t)]+[\dot P_-(t),P_-(t)]\big)\nonumber \\
        &=\frac{i}{2}\big([\dot P_+(t),P_+(t)]+[-\dot P_+(t),\mathbbm{1}-P_+(t)]\big) \nonumber \\
        &=i\big[\dot P_+(t),P_+(t)\big], \label{eq:intermediate_At}
\end{align}
where we have used that $$\dot P_-(t) = \frac{\mathrm{d}}{\mathrm{d}t}(\mathbbm{1}-P_+(t))=-\dot P_+(t).$$ We can proceed by further simplifying $A(t)$ as given by Eq.~\eqref{eq:intermediate_At} using the explicit decomposition for $P_+(t)$ from Eq.~\eqref{eq:P_+_time_dependent}:
\begin{align}
    A(t) &= i \big[\dot P_+(t),P_+(t) \big] \nonumber\\
    &=i\big[\dot P_0 (t)\otimes \vert + \rangle \langle + \vert,P_0(t) \otimes \vert + \rangle \langle + \vert \big]\nonumber\\
    &\hspace{1cm}-i\big[\dot P_0(t) \otimes \vert - \rangle \langle - \vert,P_1(t)\otimes \vert - \rangle \langle - \vert \big]\nonumber\\
    &=i\big[\dot P_0(t),P_0(t)\big] \otimes \vert + \rangle \langle + \vert \nonumber\\
    &\hspace{1cm}- i\big[\dot P_0(t),\mathbbm{1}-P_0(t)\big] \otimes \vert - \rangle \langle - \vert \nonumber\\
    &=i\big[\dot P_0(t),P_0(t)\big]\otimes \big(\vert + \rangle \langle + \vert + \vert - \rangle \langle - \vert \big).
\end{align}
Therefore, we conclude that 
\begin{equation}\label{eq:AS_comp}
    A(t) = i \left[\frac{\mathrm{d}}{\mathrm{d}t} \Big(\vert 0_t \rangle \langle 0_t \vert \Big),\vert0_t \rangle \langle 0_t \vert\right] \otimes \mathbbm{1} \equiv A_S(t) \otimes \mathbbm{1}.
\end{equation}
Recall that 
\begin{equation}
    \vert \dot 0_t \rangle = -\frac{\dot \theta_t}{2}\sin(\theta_t/2)\vert 0 \rangle + \frac{\dot \theta_t}{2} \cos(\theta_t/2)\vert 1 \rangle =-\frac{\dot \theta_t}{2}\vert 1_t \rangle,
\end{equation}
(similarly $\vert \dot 1_t \rangle = \sfrac{\dot \theta_t}{2} \vert 0_t \rangle $), and that 
\begin{align}
    \frac{\mathrm d}{\mathrm dt}\big(\vert 0_t \rangle \langle 0_t \vert\big) &= \vert \dot 0_t \rangle \langle 0_t \vert + \vert 0_t \rangle \langle \dot 0_t \vert \\
    & =-\frac{\dot \theta_t}{2}\big(\vert 1_t \rangle \langle 0_t \vert + \vert 0_t \rangle \langle 1_t \vert \big).
\end{align}
Using this, $A_S(t)$ in Eq.~\eqref{eq:AS_comp} is
\begin{align}
    \frac{2}{i\dot \theta_t}A_S(t) &= -\big(\vert 1_t \rangle \langle 0_t \vert + \vert 0_t \rangle \langle 1_t \vert \big) \vert 0_t \rangle \langle 0_t \vert \nonumber\\ &\hspace{0.4cm}+\vert 0_t \rangle \langle 0_t \vert \big (\vert 1_t \rangle \langle 0_t \vert + \vert 0_t \rangle \langle 1_t \vert \big)  \\
    &=-\vert 1_t \rangle \langle 0_t \vert + \vert 0_t \rangle \langle 1_t \vert,
\end{align}
and thus
\begin{equation}
    A_S(t) = i\frac{\dot \theta_t}{2}\big(\vert 0_t \rangle \langle 1_t \vert - \vert 1_t \rangle \langle 0_t\vert \big).
\end{equation}

The Zeno Hamiltonian term coming from a strong-coupling limit is now given by (cf.~Eq.~\eqref{eq:strong_coupling_limit})
\begin{align}
    H_{\rm Zeno}^{\rm (str)}(t) &= \sum_{\ell \in \{+,-\}} P_\ell(t)\big(H_S(t) \otimes \mathbbm{1} + \mathbbm{1}\otimes H_L\big)P_\ell(t) \label{eq:strong Zeno Hamiltonian} \\
    &= H_S(t)\otimes \mathbbm{1}. \label{eq:strong Zeno Hamiltonian2}
\end{align}
This calculation can be split into two terms, as $\sum_{\ell \in \{+,-\}} P_\ell(t)\big(H_S(t) \otimes \mathbbm{1} \big) P_\ell(t)$ and $\sum_{\ell \in \{+,-\}} P_\ell(t) \big(\mathbbm{1}\otimes H_L \big)P_\ell(t)$. Let us calculate the first contribution. Using Eqs.~\eqref{eq:P_+_time_dependent} and~\eqref{eq:P_-_time_dependent} we get 
\begin{widetext}
\begin{equation}\label{eq:equation_to_be_referenced_later}
    \begin{split}
        \sum_{\ell \in \{+,-\}} P_\ell(t)H_S(t) \otimes \mathbbm{1} P_\ell(t) &= \Big( P_0(t) \otimes P_+ + P_1(t) \otimes P_- \Big) \Big( H_S(t) \otimes \mathbbm{1} \Big) \Big( P_0(t) \otimes P_+ + P_1(t) \otimes P_- \Big) \\
    &\;\;\;\;\;+ \Big( P_0(t) \otimes P_- + P_1(t) \otimes P_+ \Big) \Big( H_S(t) \otimes \mathbbm{1} \Big) \Big( P_0(t) \otimes P_- + P_1(t) \otimes P_+ \Big).
    \end{split}
\end{equation}
After multiplying all terms, we obtain
\begin{align}
    &\stackrel{\eqref{eq:equation_to_be_referenced_later}}{=}\Bigr(P_0(t)H_S(t)P_0(t) + P_1(t)H_S(t)P_1(t) \Bigr)\otimes P_+ + \Bigr(P_0(t)H_S(t)P_0(t) + P_1(t)H_S(t)P_1(t) \Bigr)\otimes P_- \nonumber \\
    &=H_S(t) \otimes \bigr( P_+ + P_- \bigr) = H_S(t) \otimes \mathbbm{1}.
\end{align}
\end{widetext}
Above, we have used that $P_+P_- = P_-P_+ = 0$, and similarly $P_0(t)P_1(t) = P_1(t)P_0(t)=0$ for all $0\leq t < t_1$ and $t_2 \leq t < t_3$. Between the second and third rows we have also used $P_{\pm}^2=P_{\pm}$, $P_+ + P_- = \mathbbm{1}$, and
\begin{align*}
    H_S(t) &= P_0(t) H_S(t) P_0(t) +  P_1(t) H_S(t) P_1(t).
\end{align*}  
Thus, we can rewrite
\begin{equation}
    \begin{split}
        &\sum_{\ell \in \{+,-\}} P_\ell(t) \big(H_S(t) \otimes \mathbbm{1}\big) P_\ell(t) \\
        &\;\;\;\;\;\;= P_0(t)H_S(t)P_0(t) \otimes \mathbbm{1} + P_1(t)H_S(t)P_1(t) \otimes \mathbbm{1} \\
        &\;\;\;\;\;\;= H_S(t) \otimes \mathbbm{1}.
    \end{split} 
\end{equation}
The second part of Eq.~\eqref{eq:strong Zeno Hamiltonian} is calculated similarly, yielding
\begin{equation}
    \sum_{\ell \in \{+,-\}} P_\ell(t) \big( \mathbbm{1} \otimes H_L \big) P_\ell(t) = 0,
\end{equation}
since the projectors on the lubricant are onto the $X$-basis while its local Hamiltonian is given by the third Pauli matrix $Z$. Thus the validity of Eq.~\eqref{eq:strong Zeno Hamiltonian2} is proved. Putting it all together, the effective Hamiltonian evolution in the strong coupling limit becomes
\begin{equation}\label{eq:strong_FINAL_approx}
    H_{\rm eff}(t) = \Gamma (t) R(t) \otimes X + \big( A_S(t)+H_S(t) \big) \otimes \mathbbm{1}.
\end{equation}
In Appendix~\ref{app:generalization} we discuss a generalization of this shortcut to adiabaticity induced by Zeno driving to any time-dependent Hamiltonian drive  $H_S(t)$ of any finite-dimensional quantum system $\mathcal{H}_S$.

It is important to emphasize that, although the evolution generated by $H_{\mathrm{eff}}(t)$ can be made arbitrarily close to the evolution generated by $H_{\rm tot}(t)$ in the strong-coupling limit, the Hamiltonians themselves need not be close in operator norm. In fact, for this case we see that
\begin{equation}
    \Vert H_{\rm tot}(t) - H_{\rm eff}(t) \Vert = \left \Vert \frac{\omega_L}{2} \mathbbm{1} \otimes Z - A_S(t) \otimes \mathbbm{1}\right\Vert 
\end{equation}
which is independent of $\Gamma$. As explained in detail in Ref.~\cite{burgarth2022oneboundtorulethem}, closeness of propagators does not imply closeness of generators. Thermodynamically, this matters because one \emph{cannot} simply replace the physical Hamiltonian $H_{\mathrm{tot}}(t)$ by the effective Hamiltonian $H_{\mathrm{eff}}(t)$ when evaluating work and heat.

It is worth noticing that Eq.~\eqref{eq:strong_FINAL_approx} is sufficient for lubrication (see Fig.~\ref{fig:instantaneous_full}(a)). In the strong-coupling regime, the effective dynamics suppresses the coherences that would otherwise be generated in the instantaneous energy basis of the working medium.   If the joint system is initialized in the product state $\rho_{SL}(0)$ given by Eq.~\eqref{eq:initial_state_lubricated_engine}, then evolving with the effective Hamiltonian produces the transitionless behavior required of the working medium, even before the measurement step is added. We will later compare this strong-coupling-only regime with the full Zeno-driven protocol, and in Sec.~\ref{sec:thermodynamic_costs} we will show that the two differ sharply once switching work costs are taken into account. To counter that, we combine a strong-coupling manifestation of the Zeno effect with the standard frequent monitoring which will yield an effective Zeno Hamiltonian governing the drive within a certain Zeno subspace as we now discuss.

\subsubsection{Frequently monitoring the lubricant's state}

We now add frequent selective measurements of the lubricant in the $X$ basis, described by the projectors $\mathbbm{1}\otimes \vert \ell \rangle \langle \ell \vert$ on $\mathcal{H}_{SL}$ with $\ell \in \{+,-\}$. Here \emph{selective} means that, after each measurement, the state is updated according to the actual outcome obtained in that run. We do \emph{not} postselect on a preferred sequence of outcomes or discard runs that leave a chosen subspace. Rather, each measurement record defines a trajectory, and the Zeno limit makes trajectories that remain in the same subspace overwhelmingly likely.

In the idealized protocol we understand the Zeno-assisted limit sequentially: first the strong-coupling limit $\Gamma\to\infty$ is taken at fixed stroke duration, yielding the effective Hamiltonian containing the quasistatic transporter; the frequent-monitoring limit is then taken within this effective description to confine the lubricant to a fixed $X$-measurement sector (we also note in App.~\ref{app:commutativity} that in some situations one cannot swap the order of these procedures).

The measurement $\{\mathbbm{1} \otimes \vert \ell \rangle \langle \ell \vert\}_{\ell \in \{+,-\}}$ induces a partition $\mathcal{H}_{SL} = \mathcal{H}_S \otimes \mathcal{H}_{\vert + \rangle }\oplus\mathcal{H}_S \otimes \mathcal{H}_{\vert - \rangle}$, where $\mathcal{H}_{\vert \ell \rangle} = \mathrm{span}(\{\vert \ell \rangle \})$. This leads to an evolution guided by the following Zeno Hamiltonian
\begin{equation}
    H_{\rm Zeno}^{(\ell)}(t) = \mathbbm{1} \otimes \vert \ell \rangle \langle \ell \vert H_{\rm eff}(t)  \mathbbm{1} \otimes \vert \ell \rangle \langle \ell \vert
\end{equation}
within the subspace $\mathbbm{1} \otimes \vert \ell \rangle \langle \ell \vert$, where $H_{\rm eff}(t)$ is given by Eq.~\eqref{eq:strong_FINAL_approx}. This, in turn, leads to
\begin{equation}\label{eq:Zeno_Hamiltonian}
    H_{\rm Zeno}^{(\ell)}(t) = \ell\, \Gamma (t)  R(t)\otimes \vert \ell \rangle \langle \ell \vert +\big( A_S(t)+H_S(t) \big) \otimes \vert \ell \rangle \langle \ell \vert. 
\end{equation}
Therefore, we conclude that within the $\mathcal{H}_S \otimes  \mathcal{H}_{\vert \ell \rangle}$ subspaces, the \emph{composite} system and lubricant evolve according to the unitary evolution \begin{equation}
U_{\rm Zeno}^{(\ell)}(t) = \mathcal{T}\left[\exp\left(-i\int H_{\rm Zeno}^{(\ell)}(s)\, \mathrm{d}s\right)\right].
\end{equation}
Within a fixed Zeno subspace---e.g., $\mathcal{H}_S \otimes \mathcal{H}_{\vert + \rangle }$---the effective evolution of the working medium is that of a transitionless drive generated by $H_{\mathrm{S}}(t)+A_{\mathrm{S}}(t)$. In the ideal Zeno limit, and for the dominant trajectory that remains in the fixed Zeno subspace throughout, the lubricant returns to its initial state and merely mediates the control of the working medium. In that limited sense it behaves as a reusable auxiliary system. More precisely, starting from $\rho_{SL}(0)=\rho_S(0)\otimes\vert +\rangle\langle + \vert $, the idealized Zeno protocol yields $\rho_{SL}(t_1)=\rho_S(t_1)\otimes|+\rangle\langle+|$.
\begin{figure*}[t]
    \centering
    \includegraphics[width=1\textwidth]{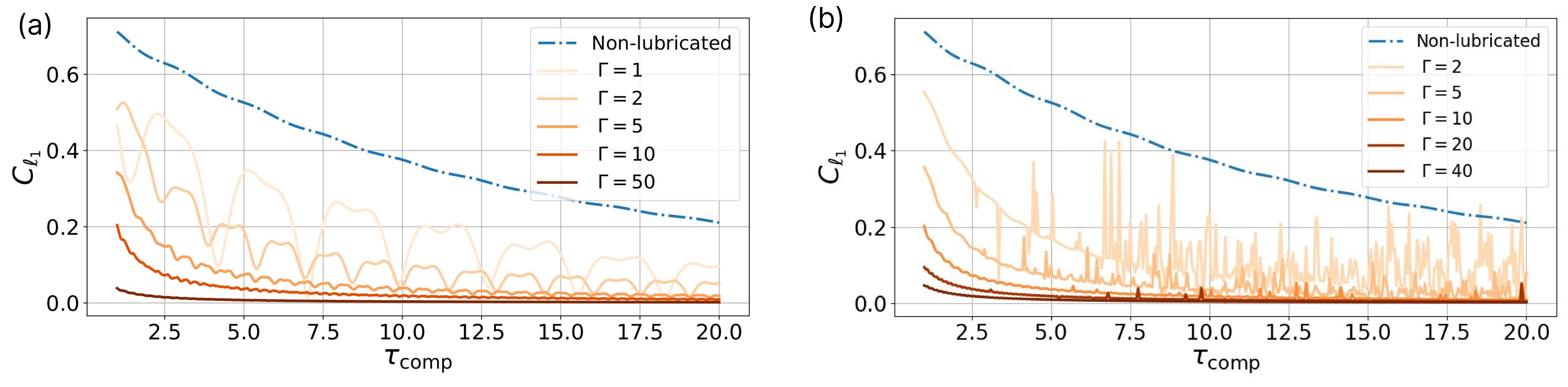}
    \caption{\textbf{Instantaneous coherence in the working system after strong-coupling and Zeno-assisted drives.} (a) Coherence quantifier $C_{\ell_1}(\rho_S(t_1)\vert H_S(t_1))$ of the final state $\rho_S(t_1)$ at the end of the compression stroke, relative to the final Hamiltonian $H_S(t_1)=H_{\rm h}$, as a function of the compression-stroke duration $\tau_{\rm comp}$. The curve is obtained from simulated values sampled every 0.05 in $\tau_{\rm comp}$. The total evolution of the system and the lubricant is generated unitarily by $H_{\rm tot}(t)$ in Eq.~\eqref{eq:total_hamiltonian_SL}; the coupling strength is varied by changing $\Gamma$ from 0 (the non-lubricated case) to 50, approaching the strong-coupling limit. (b) Same coherence quantifier for a finite-time Zeno drive. The curve is obtained from a single stochastic trajectory, sampled every 0.05 in $\tau_{\rm comp}$. Because the Zeno drive is based on selective measurements, different measurement records lead to different realizations, so the dynamics are inherently trajectory dependent. Parameters: $\omega=\omega_L=1$, $\Omega_0=5$, $T_{\rm c}=0.5$, and $n=100$.} 
    \label{fig:instantaneous_full}
\end{figure*}

\subsubsection{Comparison with counter-diabatic drive}

It remains to identify the additional term $A_{\mathrm{S}}(t)$ with the standard counter-diabatic correction~\cite{berry2009transitionless,demirplak2003adiabatic,demirplak2005assisted}. Once this is done, the meaning of the Zeno-assisted protocol becomes transparent: within each Zeno subspace, the working medium evolves as though it were driven by the counter-diabatic Hamiltonian associated with $H_{\mathrm{S}}(t)$.

To see this, recall that given a time-dependent drive $H(t)$, with instantaneous eigenbasis given by $\{\vert n(t) \rangle \}_n$, one has a shortcut to adiabaticity via counter-diabatic drive~\cite{del_campo_shortcuts_2013} whenever we can drive the system with the different effective (counter-diabatic) Hamiltonian $H_{\rm CD}(t) = H(t)+\tilde H(t)$ where 
\begin{equation}
    \tilde H(t) =  i\sum_{n}\big( \vert \dot n_t\rangle \langle n_t\vert -\langle n_t\vert \dot n_t\rangle \vert n_t \rangle \langle n_t \vert\big).
\end{equation}
Driving the system according to $H_{\rm CD}(t)$ achieves precisely the features we are aiming for, i.e. coherences in the instantaneous eigenbasis are not generated and populations are preserved.

Translating the condition to our case, we see that $\tilde H (t)$ for the drive during the work strokes $H_S(t)$ is precisely given by $A_S(t)$ since 
\begin{align}
    \tilde H := \,& i\sum_{\ell \in \{0,1\}} \big( \vert \dot \ell_t \rangle \langle \ell_t \vert -\langle \ell_t\vert \dot \ell_t \rangle \vert \ell_t \rangle \langle \ell_t \vert \big) \nonumber \\
    =\,&i \big( \vert \dot 0_t \rangle \langle 0_t \vert -\langle 0_t\vert \dot 0_t \rangle \vert 0_t \rangle \langle 0_t \vert \nonumber\\
    & +\vert \dot 1_t \rangle \langle 1_t \vert -\langle 1_t\vert \dot 1_t \rangle \vert 1_t \rangle \langle 1_t \vert\big) \nonumber\\
    =& \,i \frac{\dot \theta_t}{2}\big( -\vert 1_t \rangle \langle 0_t \vert + \vert 0_t \rangle \langle 1_t \vert \big) = A_S(t).
\end{align}
The result holds for both work strokes with appropriate definitions of $H_S(t)$, given by Eq.~\eqref{eq:Hamiltonian_S} and $\theta_t$ from Eq.~\eqref{eq:theta(t)}. Therefore, the combined action of the two Zeno regimes induces an effective counter-diabatic Hamiltonian that generates a quasistatic unitary evolution confined to a specific Zeno subspace. Moreover, the transporter $A(t)$ emerging from the strong-coupling Zeno regime coincides with the quasistatic gauge potential associated
with the driven Hamiltonian~\cite{sels2017minimizing}. In this sense, our Zeno-assisted protocol provides a \emph{dynamical}
mechanism for generating the geometric connection underlying counter-diabatic driving~\cite{guery-odelin_shortcuts_2019}.

\subsubsection{Numerical simulations for a single lubricated work stroke}

Recall that $\{\vert 0_{\rm h}\rangle , \vert 1_{\rm h}\rangle \}$ is the spectral basis for $H_{S}(t_1) = H_{\rm h}$, the local Hamiltonian of the system at the end of  the compression stroke. Let us start by investigating numerically the coherence relative to $H_{\rm h}$ generated during a finite-time compression of duration $\tau_{\rm comp}=t_1-t_0\equiv t_1$. For that, we use the $C_{\ell_1}$ coherence monotone~\cite{baumgratz2014quantifying,streltsov2017colloquium} which in this case is equal to
\begin{equation}
    C_{\ell_1}\big(\rho_S(t_1) \mid H_{\rm h} \big) = 2\vert \langle 0_{\rm h}\,\vert\, \rho_S(t_1)\,\vert \,1_{\rm h}\rangle \vert.
\end{equation}
We choose $T_{\rm c} = 0.5$ such that $p_0 \approx 0.1192$ in Eq.~\eqref{eq:initial_state_lubricated_engine}. Then, we evolve the system according to $H_{\rm tot}(t)$ as given by Eq.~\eqref{eq:total_hamiltonian_SL}, so that $$\rho_{SL}(t_1) = U_{\rm tot}(t) \big(\rho_{S}(t_0)\otimes \vert +\rangle \langle + \vert \,\big)U_{\rm tot}(t)^\dagger$$ where $U_{\rm tot}(t)$ is the unitary evolution generated by $H_{\rm tot}(t)$. We then calculate $C_{\ell_1}\big(\rho_S(t_1) \vert H_{\rm h}\big)$ for  increasing values of $\Gamma$ and different values of $t_1 = \tau_{\rm comp}$ and show the results in Fig.~\ref{fig:instantaneous_full}(a).~\footnote{The complete replication material for all simulations presented in this work can be found in Ref.~\cite{zenodo2026zenoassisted}.} The remaining parameters are listed in the figure caption.

In the non-lubricated protocol from Sec.~\ref{sec:background_section}, short compression times generate a substantial amount of coherence in the instantaneous eigenbasis at the end of the stroke. As $\tau_{\mathrm{comp}}$ increases, the protocol approaches the quasistatic regime and this coherence decreases. By contrast, when we drive the system according to $H_{\rm tot}(t)$ and let $\Gamma$ increase, the dynamics approaches a shortcut to adiabaticity: the final coherence is strongly suppressed even at comparatively short stroke durations.

The evolution of coherences shown in Fig.~\ref{fig:instantaneous_full}(a) is deterministic and the duration of the unitary evolutions (which we will also refer to as a \emph{pulse}) is given by the entire duration of $\tau_{\rm comp}$. Once we include $n$ measurements of the lubricant, the unitary evolutions are divided and implemented in intervals $\delta t \sim \tau_{\rm comp}/n$. We end up with a sequence of short pulses of duration $\delta t$, contrasted with (selective) non-destructive $X$ basis measurements of the lubricant as shown in Fig.~\ref{fig:zeno_drive}. 

Figure~\ref{fig:instantaneous_full}(b) shows the corresponding simulation when the drive is implemented through frequent measurements as well as strong coupling (here with $n= 100$ measurements and system-lubricant coupling up to $\Gamma = 40$). In this Zeno drive case the outcome is trajectory dependent: each realization of the measurement record produces a different final reduced state of the working medium, and therefore a different value of the coherence. Outside the ideal Zeno regime one observes fluctuations associated with occasional jumps between Zeno subspaces. These fluctuations become smaller when the measurements are sufficiently frequent, i.e. when $\delta t\sim \tau_{\mathrm{comp}}/n$ is small.

\begin{figure*}[t]
    \centering
    \includegraphics[width=1\linewidth]{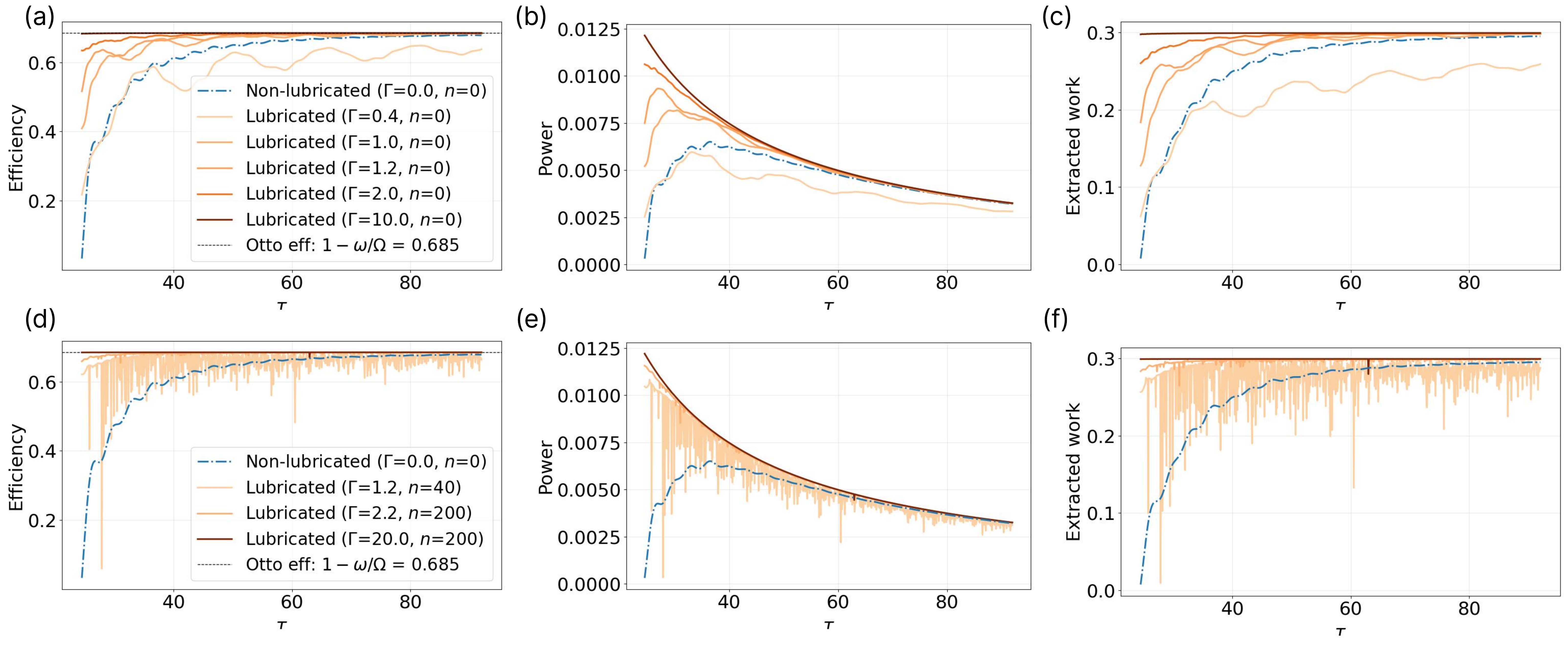}
    \caption{\textbf{Efficiency, power, and extracted work gains from lubrication via the Zeno effect without accounting for thermodynamic costs.} The non-lubricated case is shown for comparison, illustrating the detrimental effects of coherence when the work strokes are fast. Panels (a)-(c) show the gains obtained from lubrication in the strong-coupling limit. Panels (d)-(f) show the gains obtained when strong coupling is combined with frequent monitoring of the lubricant. In these, a single trajectory with $n$ measurements is shown for each choice of lubrication parameter. The duration of the cold thermalization stroke is $\tau_{\rm cold}=12$, and that of the hot thermalization stroke is $\tau_{\rm hot}=5$. The duration of the compression stroke varies in the range $5 \leq \tau_{\rm comp} \leq 50$, while the expansion stroke is set to half of the compression duration, $\tau_{\rm exp}=\tau_{\rm comp}/2$. {We always consider $\Gamma_{\rm comp} = \Gamma_{\rm exp}=\Gamma$. The curves are obtained from simulated values at intervals of $0.05$ in $\tau_{\rm comp}$. The net extracted work is plotted as $-\Delta W_{\rm tot}$, which for the chosen parameters equals $\Delta W_{\rm tot} \approx -0.30107$ in the ideal  limit of transitionless drive and ideal thermalization. In both lubrication scenarios, the Zeno limit suppresses instantaneous coherences and fluctuations, and the engine approaches the ideal quasistatic behavior.} Parameters: $\gamma_{\rm h}=\gamma_{\rm c}=0.5$, $\omega=\omega_L=1$, $\Omega_0=3.01105$, $T_{\rm c}=0.5$, and $T_{\rm h}=3$.}\label{fig:full_figure_test}
\end{figure*}

\subsection{Isochoric thermalization strokes}

Before simulating a full engine cycle, we must specify how the lubricant is treated during the heat strokes. During those strokes we \emph{switch off} the interaction between working medium and lubricant and allow only the working medium to interact with the thermal bath. Operationally, the lubricant is treated as an auxiliary control degree of freedom that is used during the work strokes and removed during thermalization.

If the work stroke is implemented by the full Zeno protocol, the ideal final state is a product state of the form $\rho_S(t_1)\otimes|\ell\rangle\langle\ell|$, with $\ell\in\{+,-\}$ determined by the trajectory. For sufficiently frequent measurements, trajectories that remain in the same Zeno subspace dominate. In the strong-coupling-only protocol, by contrast, the final state is the joint state $\rho_{SL}(t_1)$, which may contain correlations between working medium and lubricant. In both cases, the state brought into contact with the bath is the reduced state $\rho_S(t_1)=\mathrm{Tr}_L[\rho_{SL}(t_1)].$ At the beginning of the next work stroke we reinitialize the lubricant in the state $|+\rangle\langle+|$.

For the numerical illustrations in the next section, we approximate the heat strokes as sufficiently long to prepare the working medium close to the corresponding thermal state. This is an idealization: in the Markovian model used here, exact Gibbs states are reached only asymptotically. The approximation is nevertheless useful because it isolates the role of the lubricated work strokes. Under this assumption, deviations from ideal performance can be attributed primarily to the driving stage, allowing us to assess more cleanly how much improvement is produced by the Zeno-assisted protocol. As mentioned before, the limitations of this approximation are revisited in Sec.~\ref{sec:thermodynamic_costs}.

\subsection{Numerical simulations of a complete cycle of the lubricated engine}

We now evaluate the heat, work, efficiency, and power of the lubricated engine using the reduced state of the working medium and the definitions introduced in Secs.~\ref{sec:QHE_and_Otto_cycle} and~\ref{sec:efficiency_power_and_finite_time}. We show these in Figure~\ref{fig:full_figure_test}. To highlight the dynamical mechanisms behind lubrication, we compare three cases: (i) non-lubricated (reference) dynamics, (ii) strong-coupling, and (iii) combined strong coupling with frequent monitoring of the lubricant (Zeno drive). 

Figures~\ref{fig:full_figure_test}(a)-\ref{fig:full_figure_test}(c) show that, as $\Gamma$ increases, the strong-coupling protocol approaches the transitionless benchmark at shorter cycle times: the extracted work moves toward the quasistatic value and the efficiency approaches the efficiency of the quasistatically driven Otto engine. In this regime, the work strokes can be made arbitrarily fast, leaving the bottleneck for the power generation to be given by the duration of the thermalization strokes. We also remark that numerical results indicate a \emph{non-monotonic} effect of implementing the drive together with the interaction with the lubricant system on the engine performance. In the parameter regime considered here, the presence of the lubricant actually deteriorates the engine operation for $0 < \Gamma < 1$. Only for $\Gamma > 1$ does the inclusion of the lubricant system become beneficial. 

For the full Zeno-driven protocol, the work strokes are implemented as a sequence of short unitary pulses generated by $H_{\rm tot}(t)$, interspersed with selective measurements of the lubricant in the $X$ basis. Each simulation therefore produces a single stochastic trajectory. Figures~\ref{fig:full_figure_test}(d)-\ref{fig:full_figure_test}(f) show a  representative trajectory for different pairs $(\Gamma, n)$ of coupling strength $\Gamma$ and number of measurements $n$ for a Zeno drive. Even though these parameters are independent, they work together as increasing each aids lubrication. As both parameters increase, the fluctuations decrease and the dynamics approaches the transitionless limit.

We remark that the work difference $\Delta W_{\rm tot}$ given by Eq.~\eqref{eq:total_work_explicit} is negative (so that the net extracted work $- \Delta W_{\rm tot}$ shown in Figs.~\ref{fig:full_figure_test}(c) and (f) is positive) for the parameter ranges considered in this work. In our numerical studies we typically set $T_{\rm c}=0.5,\; T_{\rm h}=3,\; \omega=\omega_L=1$. Recalling that $\Omega=\sqrt{\omega^2+\Omega_0^2}\,, $ one finds that, assuming perfect thermalization during the isochoric hot and cold strokes, obtaining $\Delta W_{\rm tot}<0$ would require choosing $\Omega_0$ in the range $0<\Omega_0\lesssim 5.9$. To maximize the useful work extracted from the quantum heat engine for the temperatures and frequencies above, the numerically optimal choice lies near $\Omega_0\approx 3.01$.

In all these plots we show, for comparison, the  non-lubricated predictions. In the limit where $\tau_{\rm comp}$ is large (in this case, shown to be the case when the full cycle time $\tau \geq 60$) the non-lubricated engine approaches the quasistatic regime and both lubricated and non-lubricated engines converge.

These numerical results show that the protocol can substantially suppress coherence generation during the work strokes and thereby improve the apparent efficiency and power of the engine within the idealized model used so far. They also motivate a more careful analysis of implementation costs done in Sec.~\ref{sec:thermodynamic_costs}. Before turning to that analysis, it is useful to compare our approach with a related scheme that leads to similar conclusions and has already been examined from a similar perspective, namely the dephasing-assisted QHEs studied recently by Weber \emph{et al.}~\cite{weber2023thermodynamiccostspuredephasing}. Such lubrication mechanisms are among the most extensively studied and prominent in the literature~\cite{rezek_reflections_2010,camati2019coherence,kosloff_discrete_2002,feldmann2006quantum}.

\section{On the comparison between Zeno- and dephasing-assisted QHEs}\label{sec:relationship_between_decoherence_and_zeno}

\begin{figure*}[t]
    \centering
    \includegraphics[width=0.9\textwidth]{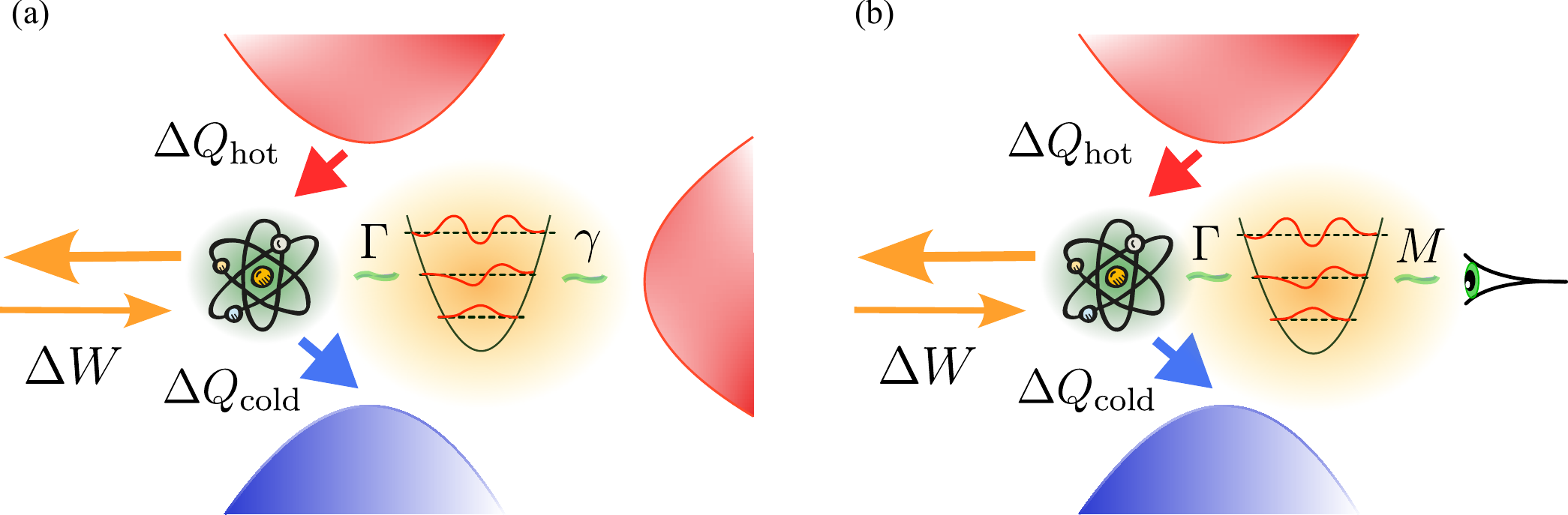}
    \caption{\textbf{Comparison between dephasing-assisted and Zeno-assisted quantum heat engine.} (a) Dephasing-assisted QHE from Ref.~\cite{weber2023thermodynamiccostspuredephasing}. The working system is provided by a single qubit. The lubricant is a composite system of a harmonic oscillator and a thermal bath at some temperature $T_{L^{(2)}}$. The working qubit couples with the harmonic oscillator, with strength $\Gamma$, which in turn couples to a thermal bath with strength $\gamma$. (b) An analogous Zeno-assisted engine where the harmonic oscillator strongly couples to the working system and is frequently measured.}\label{fig:dephasing_vs_zeno}
\end{figure*}

\subsection{Dephasing-assisted QHE from Weber \emph{et al.}}

Reference~\cite{weber2023thermodynamiccostspuredephasing} studied a dephasing-assisted lubrication scheme for essentially the same QHE introduced in Sec.~\ref{sec:QHE_and_Otto_cycle}. Their protocol follows the general idea of quantum lubrication proposed in  Ref.~\cite{feldmann2006quantum}, namely to couple the working medium to an auxiliary system that suppresses the harmful coherences generated during fast driving. Because our Zeno-assisted protocol plays a similar conceptual role, it is useful to compare the two mechanisms directly. Their setup is shown in Fig.~\ref{fig:dephasing_vs_zeno}(a). The lubricant is taken to be a system $\mathcal{H}_{L}^{(1)}\otimes \mathcal{H}_L^{(2)}$, composed of a quantum harmonic oscillator $\mathcal{H}_L^{(1)}$---which strongly interacts with the single-qubit working medium forming a typical spin-boson composite~\cite{Anto_Sztrikacs2021strong}---and a bosonic thermal bath $\mathcal{H}_L^{(2)}$ at temperature $T_{L^{(2)}}$ ($\beta_{L^{(2)}}=T_{L^{(2)}}^{-1}$). The bosonic bath and the working medium never interact directly. They refer to the composite lubricant system $\mathcal{H}_L^{(1)}\otimes\mathcal{H}_L^{(2)}$ as the \emph{dephasing bath}. 

Weber \emph{et al.} considered two specific models for this dephasing bath, motivated by distinct simulation techniques: one which is targeted for simulations using the Time Evolving Density matrices using Orthogonal Polynomials Algorithm (TEDOPA)~\cite{prior2010efficient,chin2010exact}, and another based on the Dissipation Assisted Matrix Product factorization (DAMPF) approach~\cite{somoza2019dissipation,mascherpa2020optimized}. In what follows, we consider only their DAMPF case, which is also the one that fits better with recent investigations on the role of strong couplings in thermodynamic devices~\cite{Anto_Sztrikacs2021strong,ivander2022strong,albarelli2024pesudomode}.

Their dephasing model is characterized by a full time-dependent Hamiltonian
\begin{equation}
    H_{SL^{(1)}}(t)=H_S(t)\otimes\mathbbm{1}+ \mathbbm{1}\otimes H_{L^{(1)}} + H_{\rm int}(t),
\end{equation}
where 
\begin{equation}
H_{L^{(1)}}=\omega_0 b^\dagger b\end{equation} 
(with $b$ and $b^\dagger$ being the annihilation and creation operators of the harmonic oscillator), $H_S(t)$ is the system Hamiltonian given by Eq.~\eqref{eq:Hamiltonian_S}, and $H_{\rm int}(t)$ is
\begin{equation}
    H_{\rm int}(t)=\Gamma(t) \,R(t)\otimes(b^\dagger+b).
\end{equation}
They consider the full dynamics given by the unitary part together with the oscillator damping:
\begin{equation}
    \frac{\mathrm d}{\mathrm d t}\rho = -i\big[H_{SL^{(1)}}(t),\rho\big] + \gamma\,\mathcal{D}_{\beta_{L^{(2)}}}[\rho],
\end{equation}
with
\begin{align}
    \mathcal{D}_{\beta_{L^{(2)}}}[\rho] &= \big(1+n_{L^{(2)}}\big)\left(b\,\rho\,b^\dagger - \tfrac{1}{2}\{b^\dagger b,\rho\}\right) \nonumber\\
    &\quad + n_{L^{(2)}}\left(b^\dagger\,\rho\,b - \tfrac{1}{2}\{b b^\dagger,\rho\}\right),\label{eq:time_dependent_dissipator_weber}
\end{align}
where $n_{L^{(2)}}=1/(e^{\omega_0\beta_{L^{(2)}}}-1)$ is the mean thermal photon number. Numerically, they find that the optimal lubrication regime is one in which $\Gamma,\gamma\gg\omega,\Omega_0,\gamma_{\rm h},\gamma_{\rm c}$. In other words, the working system is strongly coupled with $\mathcal{H}_L^{(1)}$, which in turn is strongly damped by $\mathcal{H}_L^{(2)}$, engineered so that the ratio $\Gamma/\gamma$ remains finite. In this regime, the full dephasing bath $\mathcal{H}_L^{(1)}\otimes\mathcal{H}_L^{(2)}$ can be described analytically by a temperature-dependent Lorentzian spectral density,
\begin{align}
    J_{\beta_{L^{(2)}}}(\omega) &= \tfrac{1}{2}J(\omega)\!\left(\coth\frac{\beta_{L^{(2)}}\omega_0}{2}+1\right) \nonumber\\
    &\quad + \tfrac{1}{2}J(-\omega)\!\left(\coth\frac{\beta_{L^{(2)}}\omega_0}{2}-1\right),
\end{align}
with
\begin{equation}
    J(\omega)=\frac{2\Gamma^2\gamma}{\gamma^2+4(\omega-\omega_0)^2}.
\end{equation}
They further derive an effective (for the case when the temperature of the lubricant thermal bath $T_{L^{(2)}} \approx 0$), controllable dephasing rate
\begin{equation}
    \Gamma_{\rm eff}=\frac{8\Gamma^2\gamma}{\gamma^2+4\omega_0^2},
\end{equation}
and engineer the parameters such that $\gamma/\omega_0\ll 1$. In this regime, they are able to approach the Otto efficiency $\eta_{\rm Otto}$~\eqref{eq:Otto_eff_def} at finite (in fact, arbitrarily high) power relative to the working system.

Through a series of numerical experiments they observe that, with good approximation, the joint state of $\mathcal{H}_S\otimes\mathcal{H}_L^{(1)}$ during work and heat strokes takes the form
\begin{equation}
    \rho_{SL^{(1)}}(t) \approx p_0\,|0_t\rangle\langle 0_t|\otimes |-\alpha\rangle\langle -\alpha| + p_1\,|1_t\rangle\langle 1_t|\otimes |\alpha\rangle\langle\alpha|,\label{eq:Weber_ansatz}
\end{equation}
and they compute the displacement $\alpha$ for which the working medium in contact with the strong dephasing bath is close to a steady state. The value
\begin{equation}\label{eq:steady_state_alpha}
    \alpha \;=\; 2\Gamma\,\frac{2\omega_0 + i\gamma}{4\omega_0^2+\gamma^2} \approx \frac{\Gamma}{\omega_0}
\end{equation}
is found to provide a parametrization for the steady states (note that we have used $\gamma \ll \omega_0$).

As a result, the evolution within the system's subspace is given by the state
\begin{equation}
    \rho_S(t) = \mathrm{Tr}_{L^{(1)}}\left[\rho_{SL^{(1)}}(t)\right] \approx p_0 \vert 0_t \rangle \langle 0_t \vert + p_1 \vert 1_t \rangle \langle 1_t \vert
\end{equation}
which is exactly the type of transitionless drive required for a successful quantum lubrication procedure.

\subsection{A Zeno-assisted reformulation}

We now show that the same effective transitionless drive can be recovered from a Zeno-assisted description of the setup in  Fig.~\ref{fig:dephasing_vs_zeno}-(b). Assume that the harmonic-oscillator lubricant is prepared within the subspace spanned by $\{\vert \alpha \rangle, \vert -\alpha \rangle\}$, with $\alpha$ given approximately by Eq.~\eqref{eq:steady_state_alpha}. In the strong-coupling regime this displacement is large enough that the two coherent states are nearly orthogonal, which allows us to interpret that subspace as an effective two-level system. Recall that $$\langle \alpha \vert (b+b^\dagger) \vert \alpha \rangle = 2 \mathfrak{Re}[\alpha],\quad \langle -\alpha \vert (b+b^\dagger) \vert -\alpha \rangle = -2 \mathfrak{Re}[\alpha]$$ and that for large enough $\alpha$ (and therefore large enough coupling $\Gamma$) we have
\begin{equation}
    \langle \mp\alpha \vert b+b^\dagger \vert \pm \alpha \rangle \approx 0.
\end{equation}

The position quadrature operator $(b+b^\dagger)/\sqrt{2}$ has infinitely many eigenstates. However, \emph{within the subspace} spanned by $\{\vert \alpha \rangle, \vert - \alpha \rangle\}$, its action approximates that of implementing the operation given by 
\begin{equation}
    \frac{b+b^\dagger}{\sqrt{2}} \approx \sqrt{2}\mathfrak{Re}[\alpha]\Bigr(\vert \alpha \rangle \langle \alpha \vert - \vert -\alpha \rangle \langle -\alpha \vert\Bigr). 
\end{equation}
Thus, within that subspace we let the coupling $\Gamma R(t) \otimes (b + b^\dagger)$ be described by a spectral decomposition with spectrum $\{\pm 2\Gamma\,\mathfrak{Re}[\alpha]\}$ and projectors  
\begin{equation}
    P_-(t) \approx \vert 0_t \rangle \langle 0_t \vert \otimes \vert -\alpha \rangle \langle -\alpha \vert + \vert 1_t \rangle \langle 1_t \vert \otimes \vert \alpha \rangle \langle \alpha \vert
\end{equation}
and $P_+(t) = \mathbbm{1}-P_-(t)$. The orthogonality approximation is valid up to $\vert \langle - \alpha \vert \alpha \rangle \vert^2 = e^{-4 \vert \alpha \vert^2}$, which is approximately zero for moderately high values of $\vert \alpha\vert$. 

Assuming that we are in this regime we can apply the effective unitary evolution from Eq.~\eqref{eq:strong_coupling_effective} relative to the subspace $\mathcal{H}_{S} \otimes \mathcal{H}_{L}^{(1)}$. In this case, the generator of the quasistatic transporter $A(t)$ from Eq.~\eqref{eq:adiabatic_transporter} is given by 
\begin{equation}
    A(t) \approx A_S(t) \otimes S_\alpha 
\end{equation}
where we have defined
\begin{equation}
    S_\alpha \equiv \vert -\alpha \rangle \langle -\alpha \vert + \vert \alpha \rangle \langle \alpha \vert.
\end{equation}
Similarly, the Zeno Hamiltonian arising from a strong coupling regime from Eq.~\eqref{eq:strong_coupling_limit} is, in this case,  
\begin{align}
    H_{\rm Zeno}^{(\rm str)}(t) &= \sum_{\ell \in \{+,-\}}P_\ell(t) \left(H_S(t)\otimes \mathbbm{1}+\mathbbm{1}\otimes H_{L}^{(1)}\right)P_{\ell }(t)\nonumber \\
    &\approx \left(H_S(t) + \vert \alpha \vert^2\omega_0 \mathbbm{1}\right)\otimes S_\alpha.
\end{align}
In summary, we end up with an approximate effective Hamiltonian drive:
\begin{align}
    H_{SL^{(1)}\rm,eff}&(t) \approx \Gamma (t) R(t) \otimes (b+b^\dagger)\nonumber +\\ &+(H_S(t)+A_S(t)+\vert \alpha \vert^2 \omega_0 \mathbbm{1})\otimes S_\alpha. \label{eq:time-dependent-effective-damping}
\end{align}
From this, we see that we can motivate the ansatz from Eq.~\eqref{eq:Weber_ansatz} using the strong coupling approximation. Starting the lubricant system $\mathcal{H}_{L^{(1)}}$ in a pure state $\vert \phi \rangle \langle \phi \vert $ implies that (writing $\mathcal{U}[\cdot ] \equiv U(\cdot )U^\dagger$)
\begin{align}
\rho_{SL^{(1)}}(t) &= \mathcal{U}_{\rm tot}(t) \bigr[p_0\vert 0\phi \rangle \langle 0\phi \vert + p_1 \vert 1\phi \rangle \langle 1\phi \vert \bigr ]\nonumber \\
&=p_0 \mathcal{U}_{\rm tot}(t) \bigr [\vert 0\phi \rangle \langle 0\phi \vert \bigr ] + p_1 \mathcal{U}_{\rm tot}(t) \bigr [\vert 1\phi \rangle \langle 1\phi \vert \bigr ] \nonumber \\
&\approx  p_0 \mathcal{U}_{\rm eff}(t) \bigr [\vert 0\phi \rangle \langle 0\phi \vert \bigr ] + p_1 \mathcal{U}_{\rm eff}(t) \bigr [\vert 1\phi \rangle \langle 1\phi \vert \bigr ] \nonumber \\
&\approx p_0 \vert 0_{t}\rangle \langle 0_{ t}\vert \otimes \vert \phi_0\rangle \langle \phi_0\vert + p_1 \vert 1_{t}\rangle \langle 1_{t}\vert \otimes \vert \phi_1 \rangle \langle \phi_1 \vert. \label{eq: Ansaz_connecting_to_Weber_et_al}
\end{align}
In the last approximation we use the fact that, in the strong-coupling regime of interest, the effective evolution approximately preserves product structure between the working medium and the relevant coherent-state subspace of the lubricant. In the two-qubit model considered earlier this is reflected in the suppression of entanglement shown later in Fig.~\ref{fig:log_neg}. The point of Eq.~\eqref{eq: Ansaz_connecting_to_Weber_et_al} is therefore not that entanglement is exactly absent at finite $\Gamma$, but that it becomes negligible in the regime relevant for the comparison.

To conclude, we include frequent monitoring of the position quadrature of the harmonic oscillator, which we approximate as a measurement $\{\vert \alpha \rangle \langle \alpha \vert  , \vert - \alpha \rangle \langle -\alpha \vert  \}$. In the frequent monitoring limit, we have that the harmonic oscillator will stabilize in a specific coherent state and the joint evolution will be given by the Zeno Hamiltonian 
\begin{align}
    H_{\rm Zeno}^{(\alpha)}(t) &\approx 2 \Gamma \mathfrak{Re}[\alpha] \, R(t) \otimes \vert \alpha \rangle \langle \alpha \vert \,\,+ \nonumber \\
    &+ \,\bigr (H_S(t)+A_S(t) + \vert \alpha \vert^2 \omega_0 \mathbbm{1}\bigr )\otimes \vert \alpha \rangle \langle \alpha \vert. 
\end{align}

These calculations have a simple explanation from the perspective of the so-called \emph{coherently encoded qubits}~\cite{ralph2003coherentstates}. Within the subspace $\{\vert \alpha \rangle, \vert- \alpha \rangle\}$ the two states can be interpreted as defining a logical quantum bit via the specification $\vert 0 \rangle_{\rm logical} \equiv \vert \alpha \rangle$ and $\vert 1 \rangle_{\rm logical} = \vert -\alpha \rangle$. In this case, the quadrature operator acts as a phase flip operator, playing the role of the logical Pauli $Z_{\rm logical}$ matrix. The whole analysis then easily parallels that of the Zeno-assisted engine we considered in Sec.~\ref{sec:Zeno_assisted_quantum_heat_engine}, where the role played there by $X$ is played here by the logical operator $Z_{\rm logical}$.

\section{Thermodynamic footprint of a Zeno-assisted engine}\label{sec:thermodynamic_costs} 

Up to this point, we have shown that QZD can induce an effective shortcut to adiabaticity, thereby enabling the transitionless dynamics required to enhance a QHE. This leads to a natural question: Does this control advantage remain meaningful once the protocol is embedded in a \emph{realistic} thermodynamic account? In this section, we therefore examine where the idealized picture may break down and which costs are unavoidable in a physical implementation.

Two issues must be separated clearly. First, the previous sections tracked only energetic changes in the working medium and therefore ignored several implementation costs, such as the energy required to switch the system-lubricant interaction on and off, the cost of driving, the resources needed to perform frequent measurements, and the preparation of the lubricant in a low-entropy state. Second, even the way work is estimated must be reconsidered once the lubricant and the measuring process are included explicitly. Both points are addressed below.

To make the issue concrete, consider a single lubricated work stroke. The working medium begins with local Hamiltonian $H_{\mathrm{c}}$ or $H_{\mathrm{h}}$; a lubricant with its own local Hamiltonian $H_L$ is brought in; the interaction between the two is switched on; the joint system is driven for a duration $\tau_{\mathrm{comp}}$ or $\tau_{\mathrm{exp}}$; and, in the Zeno protocol, measurements are performed on the lubricant during that interval. At the end of the stroke the interaction is switched off again, the working medium has reached the new local Hamiltonian, the lubricant has local Hamiltonian $H_L$, and the joint state is some $\rho_{SL}(t)$. A faithful work accounting should therefore explain under which assumptions one is allowed to ignore the lubricant and retain only the reduced energetics of the working medium.

The situation is perhaps more drastic when implementing a Zeno drive. In the Zeno limit, the drive is a unitary process confined to a particular Zeno subspace. Outside of this limit, however, the process is \emph{not} adiabatic, which implies that Eq.~\eqref{eq:work} cannot be directly employed, and the associated dissipation must be taken into account. Moreover, and as can be seen in Fig.~\ref{fig:full_figure_test}, given a finite number of implemented measurements, we expect that for each sequence of selective measurements performed on the lubricant, the output work should \emph{fluctuate}, implying that a better description of work should be as an \emph{average} of all possible trajectories. In this section, we therefore revisit our assumptions about how work and heat are quantified, and examine whether they remain valid in our setting.

\subsection{What counts as a thermodynamic cost?}

Before analyzing specific costs, it is helpful to introduce a simple operational classification of what counts as a ``thermodynamic cost'' in our setting. The point of this classification is not to propose mutually exclusive thermodynamic primitives, but to organize the different ways in which a realistic implementation can reduce useful work output or power.

We will use \emph{work cost} for operations that require external energy input through a controlled intervention, such as switching interactions, driving the system, or re-preparing auxiliary degrees of freedom. We will use \emph{heat cost} for losses associated with imperfect energy exchange with the baths, for example incomplete thermalization or unwanted heat flow into auxiliary systems. Finally, we will use \emph{dissipative cost} for irreversibility generated by non-ideal dynamics, including entropy production, measurement-induced disturbance, and finite-time transitions that reduce the useful work later available from the cycle.

Taken together, these categories capture the dominant operational mechanisms by which an otherwise ideal QHE loses useful work or power.

\subsection{Revisiting the estimation of extracted work}\label{sec:revisiting}

Our earlier definition of work follows the standard semiclassical treatment of stroke-based QHEs~\cite{Dann2023}, in which one tracks the energetics of the working medium while treating the external drive as classical. This framework is appropriate for the bare Otto cycle, but once a lubricant and a measurement protocol are introduced it becomes necessary to check more carefully which energetic contributions can still be neglected and which cannot.

In our lubricated device, however, this assumption of adiabaticity may be challenged in at least two ways. In the strong-coupling regime, introducing a lubricant, coupling it to the working medium, and later tracing it out cannot occur without some form of dissipation---even in an idealized error-free setting. This is because the final state of the composite system may exhibit entanglement or classical correlations, and the reduced state of the lubricant may retain residual instantaneous coherence. In the Zeno-driving regime, the situation is different: In the ideal Zeno limit, the evolution is adiabatic, no entanglement is generated, and the lubricant returns to its initial state. Outside this limit, however, the measurements themselves may give rise to heat dissipation.

These considerations compel us to revisit how we have applied our definition of work to the lubricated case. In what follows, we therefore treat the combined system and lubricant as a new effective working medium $\mathcal{H}_{SL} = \mathcal{H}_{S} \otimes \mathcal{H}_L$, and consider two distinct scenarios that parallel the setup of Fig.~\ref{fig:full_figure_test}.

\subsubsection{Strong coupling}\label{sec:strong_coupling_COST}
\begin{figure}[t]
    \centering
    \includegraphics[width=\columnwidth]{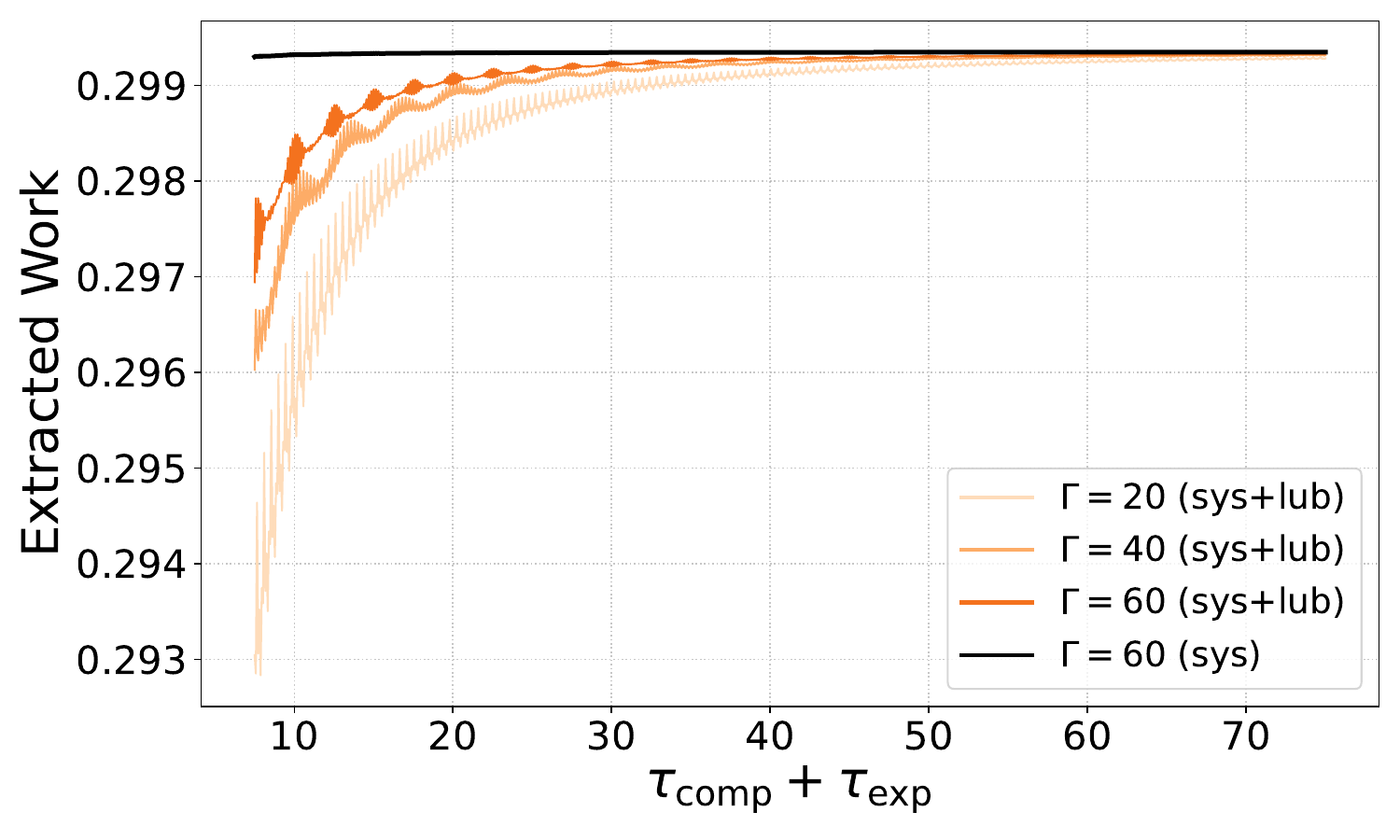}
    \caption{\textbf{Comparison between $-\Delta W_{\rm tot}$ and $-\Delta W_{\rm tot}^{\rm (s.c.)}$.} In the strong coupling limit $\Gamma \to \infty$ we have that the work contribution provided by $-\Delta W_{\rm tot}^{\rm (s.c.)}$ which includes the energetics of the lubricant and the full Hamiltonian $H_{\rm tot}(t)$ converges to $-\Delta W_{\rm tot}$, the transitionless work from Fig.~\ref{fig:full_figure_test}(c). The black line is calculated using Eq.~\eqref{eq:total_work_explicit}, while orange lines using Eq.~\eqref{eq:s.c.total_work_joint} and different values of strong coupling as shown in the legend.}
    \label{fig:str_composite_work}
\end{figure}

\begin{figure*}[t]
    \centering
    \includegraphics[width=1\linewidth]{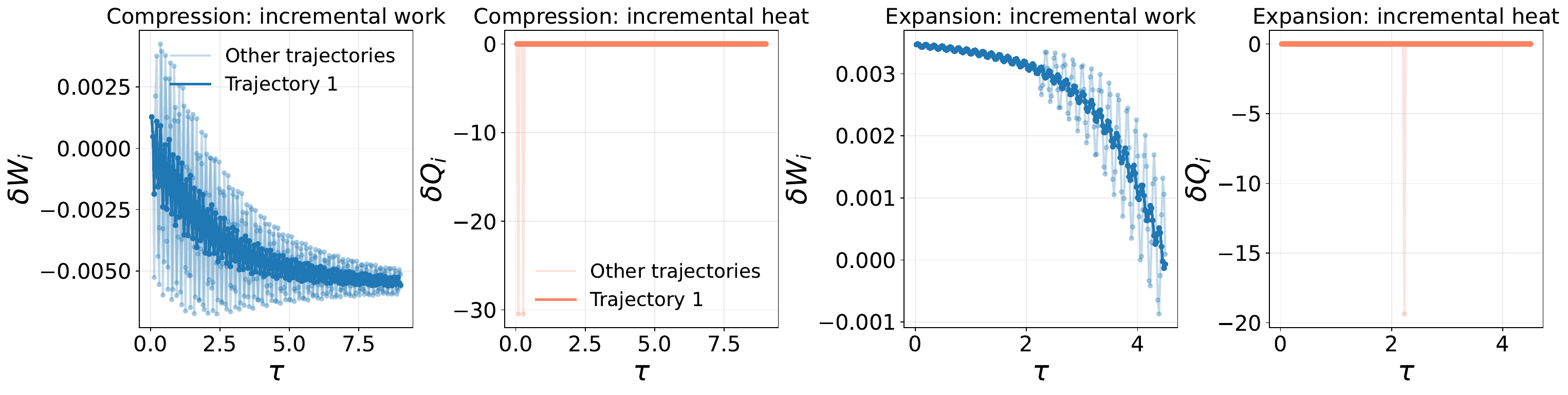}
    \caption{\textbf{Incremental work $\delta W_i$ and dissipated heat $\delta Q^{(\rm meas)}_i$.} Incremental energetic considerations during compression (left two panels) and expansion (right two panels) strokes. Each thin line represents one of the $N_{\rm traj} = 50$ trajectories; the highlighted dark blue (for the work) and dark red (for the heat) lines correspond to the first trajectory taken as a representative. Outside the Zeno limit, and for sufficiently many trajectories, unlikely trajectories where there is a jump between Zeno subspaces can occur, inducing a high dissipation cost proportional to the strong coupling  $\Gamma$.  Parameters: $\tau_{\rm comp} = 9, \tau_{\rm exp} = 4.5, n= 200, \Gamma = 20, N_{\rm traj} = 50, \omega = \omega_L = 1$, $T_{\rm c} = 0.5, T_{\rm h} = 3, \gamma_{\rm h} = \gamma_{\rm c}=0.5,$ and $ \Omega_0 = 3.01105$.}
    \label{fig:zeno_increments}
\end{figure*}
To start, we consider lubricating the device solely via the strong-coupling drive provided by $H_{\rm tot}(t)$. In this case, the evolution is fully unitary, and our revised notion of work becomes
\begin{align}
    \Delta W_{\rm tot}^{(\rm s.c.)} = &+ \langle H_{\rm tot}(t_1)\rangle_{\rho_{SL}(t_1)}- \langle H_{\rm tot}(t_0)\rangle_{\rho_{SL}(t_0)} \nonumber \\
&+\langle H_{\rm tot}(t_3)\rangle_{\rho_{SL}(t_3)}-\langle H_{\rm tot}(t_2)\rangle_{\rho_{SL}(t_2)},\label{eq:s.c.total_work_joint}
\end{align}
where
\begin{align}
\rho_{SL}(t_0) &= \left(p_0 \vert 0 \rangle \langle 0 \vert + p_1 \vert 1 \rangle \langle 1 \vert \right) \otimes \vert + \rangle \langle + \vert \label{eq: initial_state_compression}\\
\rho_{SL}(t_2)&=\left(q_0 \vert 0_{\rm h} \rangle \langle 0_{\rm h} \vert + q_1 \vert 1_{\rm h} \rangle \langle 1_{\rm h} \vert \right) \otimes \vert + \rangle \langle + \vert. \label{eq: initial_state_expansion}
\end{align}
Here, $p_0$ and $p_1$ denote the populations of the Gibbs state corresponding to the local Hamiltonian $H_{\rm c}$, while $q_0$ and $q_1$ are the analogous populations for $H_{\rm h}$. In general, the states $\rho_{SL}(t_1)$ and $\rho_{SL}(t_3)$ may contain non-trivial correlations between the system and lubricant.

We numerically estimate $\Delta W_{\rm tot}^{\rm (s.c.)}$ from Eq.~\eqref{eq:s.c.total_work_joint}, and plot the results in Fig.~\ref{fig:str_composite_work}. We see that for fixed values of $\tau_{\rm comp}$ and $\tau_{\rm exp}$, as $\Gamma \to \infty$ the net extracted work $\Delta W_{\rm tot}^{(\rm s.c.)}$ converges to the extracted work $\Delta W_{\rm tot}$ considered in Sec.~\ref{sec:Zeno_assisted_quantum_heat_engine}, which in turn approaches the transitionless work. \footnote{As one increases the coupling strength $\Gamma$, the duration of the stroke $\tau$ needs to scale adequately to ensure the dynamics approaches the quasistatic limit. The details are given in App.~\ref{app:generalization}, specifically in Eq.~\eqref{eq:precise_bound_operator_norm}.} With that, we conclude that in the strong coupling regime the lubricant system truly behaves as a near-ideal lubricant, introducing negligible changes to the total effective work extracted.

\subsubsection{Zeno drive}

For the Zeno-driven protocol,  Eq.~\eqref{eq:work} is no longer sufficient by itself because the work stroke is interrupted by measurements and the outcome of those measurements is stochastic. A more appropriate description is trajectory based: between two consecutive measurements the joint system evolves unitarily, while each measurement updates the state according to the outcome obtained in that run. Work and measurement-induced energy changes should then be assigned increment by increment and averaged over trajectories.

Therefore, for a given partitioned interval $t_i \equiv t_{(1)} \leq t_{(2)} \leq \ldots \leq t_{(n)} \equiv t_f$ of $n-1$ unitary pulses of duration $\delta t_{(k)} = t_{(k+1)}-t_{(k)}$,~\footnote{Notice that $t_0, t_1, t_2, \ldots$ denote the instants defining the QHE, whereas $t_{(1)}, t_{(2)}, \ldots, t_{(n)}$ denote the instants defining a particular partition of a work-stroke interval, given by $t_0 \leq t \leq t_1$ or $t_2 \leq t \leq t_3$.} we estimate the work for a given realization within each interval. To do so, we introduce a more explicit notation that accounts for whether a measurement has occurred. After a selective measurement at instant $t_{(k)}$ yielding outcome $\ell_k$, the joint state of the system and lubricant becomes
\begin{align*}
\rho_{SL}(t_{(k)})^{(\ell_k)} &= \frac{\mathbbm{1}\otimes \vert \ell_k \rangle \langle \ell_k \vert\rho_{SL}(t_{(k)}) \mathbbm{1}\otimes \vert \ell_k \rangle \langle \ell_k \vert}{\mathrm{Tr}\bigr[\rho_{SL}(t_{(k)}) \mathbbm{1}\otimes \vert \ell_k \rangle \langle \ell_k \vert \bigr]}
\end{align*}
Starting from $\rho^{(\ell_k)}_{SL}(t_{(k)})$, the system evolves unitarily to $\rho_{SL}(t_{(k+1)})$. The work increment $\delta W_k$ for that interval is then given by
\begin{equation}
\delta W_k(\ell_k) = \langle H_{\rm tot}(t_{(k+1)})\rangle_{\rho_{SL}(t_{(k+1)})}-\langle H_{\rm tot}(t_{(k)})\rangle_{\rho_{SL}^{(\ell_k)}(t_{(k)})}.
\end{equation}

At each instant $t_{(k)}$ a measurement is performed on the lubricant. Depending
on its outcome, we assume that a form of heat is dissipated, which we express as
\begin{align}
\delta Q_{k}^{(\rm meas)}(\ell_{k}) =
\mathrm{Tr}\left[ H_{\rm tot}(t_{(k)}) \left(\rho_{SL}(t_{(k)})-\rho_{SL}^{(\ell_{k})}(t_{(k)})\right) \right].
\end{align}
Here $\rho_{SL}(t_{(k)})$ denotes the state produced by the pulse
$U_{\rm tot}(t_{(k)},t_{(k-1)})$, i.e.\ immediately before the measurement at
$t_{(k)}$. For $k=1$ no pulse precedes the measurement and the lubricant is
prepared in $\vert + \rangle$, so the measurement leaves the state unchanged and
$\delta Q_1^{(\rm meas)} = 0$.

Outside the ideal Zeno limit, this contribution can be significant because rare jumps between Zeno subspaces occur at an energy scale set by the strong coupling $\Gamma$. As we will see, in the ideal limit, these jumps become negligible and so does the associated dissipative contribution. The complete energetic account for a given trajectory $\vec \ell = (\ell_1,\dots,\ell_{n})$ is therefore obtained by summing over all intervals,
\begin{equation}
\Delta W(t_f,t_i \mid \vec \ell \,\,) = \sum_{k=1}^{n-1} \delta W_k(\ell_k).
\end{equation}
In turn, the total heat dissipated is similarly given by
\begin{equation}
    \Delta Q^{(\rm meas)}(t_f,t_i \mid \vec \ell\,\,) = \sum_{k=1}^n\delta Q_k^{(\rm meas)}(\ell_k).
\end{equation}
The actual work output is then obtained by averaging over all possible trajectories.
\begin{equation}\label{eq:fluctuating_Zeno_work}
\Delta W^{\rm(Zeno)}(t_f,t_i) = \sum_{\vec \ell } p(\,\vec \ell\, \,) \, \Delta W(t_f,t_i \mid \vec \ell \,\,).
\end{equation}
Above, $p(\, \vec \ell \,\, )$ is the probability of obtaining a certain trajectory $\vec \ell$. In our setting, this is 
\begin{equation}\label{eq:trajectory_probability}
    p(\, \vec \ell \, \,) = \prod_{k=1}^{n} \mathrm{Tr}\Bigr [\rho_{SL}(t_{(k)}) \mathbbm{1} \otimes \vert \ell_k \rangle \langle \ell_k \vert \Bigr ].
\end{equation}

In the Zeno limit, the dominant contribution comes from the trajectory that remains in the same Zeno subspace throughout the stroke. All other trajectories become overwhelmingly unlikely as $n \to \infty$.~\footnote{The probability of leaving the Zeno subspace per measurement step scales as \(O(1/n^2)\)~\cite{facchi2008quantum}, where $n$ is the number of measurements. Hence the probability of remaining in the Zeno subspace we started at is approximately \[p_{\mathrm{no\ jump}} \simeq\prod_{k=2}^{n}\left(1-\frac{c_k}{n^2}\right) \simeq\exp\!\left[-\frac{1}{n^2}\sum_{k=2}^{n} c_k \right]\to 1,\] assuming bounded $c_k$'s which depend on the operator norm of $H_{\rm tot}(t_{(k)})$. Therefore, the probability of observing one or more jumps during the whole interval $\tau$ scales as \(O(1/n)\to 0\).}~At the same time, the average measurement-induced energy change tends to zero. This is why the trajectory-based description reduces, in the ideal limit, to the simpler transitionless-work picture used earlier. Figure~\ref{fig:zeno_increments} presents the specific estimation of the increments $\delta W_k(\ell_k)$ and $\delta Q_{k}(\ell_k)$ for 50 trajectories outside the ideal Zeno limit. The dissipations are high as they are proportional to $\Gamma$, yet occur only when there is a jump between different Zeno subspaces. Otherwise, the dissipative increments are significantly small for large $n$~(scaling as $O(1/n)$.) in comparison to the extracted work. 

We now estimate  Eq.~\eqref{eq:fluctuating_Zeno_work} numerically. Because the number of possible trajectories grows as $2^n$, an exhaustive evaluation quickly becomes impractical in the regime of frequent measurements. Apart from impracticality coming from an incredibly large number of trajectories, the vast majority of these trajectories have a near zero contribution. Therefore, we sample a certain finite number $N_{\rm traj}$  of trajectories $(\vec \ell_s)_{s=1}^{N_{\rm traj}}$ where we alternate between unitary evolutions and sampling the state update relative to the Born rule predictions appearing in Eq.~\eqref{eq:trajectory_probability}.  We then approximate $\Delta W^{(\rm Zeno)}(t_f,t_i)$ in Eq.~\eqref{eq:fluctuating_Zeno_work} as
\begin{align}
    \Delta W^{(\rm Zeno)}(t_f,t_i) &\approx \Delta W_{\rm mean}^{(\rm Zeno)}(t_f,t_i) \\
    &= \frac{1}{N_{\rm traj}}\sum_{s=1}^{N_{\rm traj}} \Delta W(t_f,t_i \mid \vec \ell_{s}).
\end{align}
We estimate the dissipative contribution in the same way, by averaging the measurement-induced energy change over the sampled trajectories
\begin{align}
    \Delta Q^{(\rm meas)}(t_f,t_i) & \approx \Delta Q^{(\rm meas)}_{\rm mean}(t_f,t_i) \\
    &= \frac{1}{N_{\rm traj}}\sum_{s=1}^{N_{\rm traj}} \Delta Q^{(\rm meas)}(t_f,t_i \mid \vec \ell_{s}). 
\end{align}
One should keep in mind that this procedure can overestimate the average dissipation if the sample is too small: jumps between Zeno subspaces are rare, but when they occur they contribute an $O(\Gamma)$ energy change and can therefore have a disproportionate effect on the sample mean.

The results are shown in Fig.~\ref{fig:otto_averaged_zeno} for 50 trajectories. We plot the averaged work extracted under a Zeno drive with $\Gamma = 20$ as the number of measurements is increased. Each point is averaged over 50 trajectories. In the simulation, the duration of the compression stroke is varied between $5$ and $10$, while the expansion stroke duration is set to $\tau_{\rm comp}/2$, as before. In this regime, the heat contribution approaches zero in the Zeno limit $n \to \infty$.

\begin{figure}[t]
    \centering
    \includegraphics[width=\columnwidth]{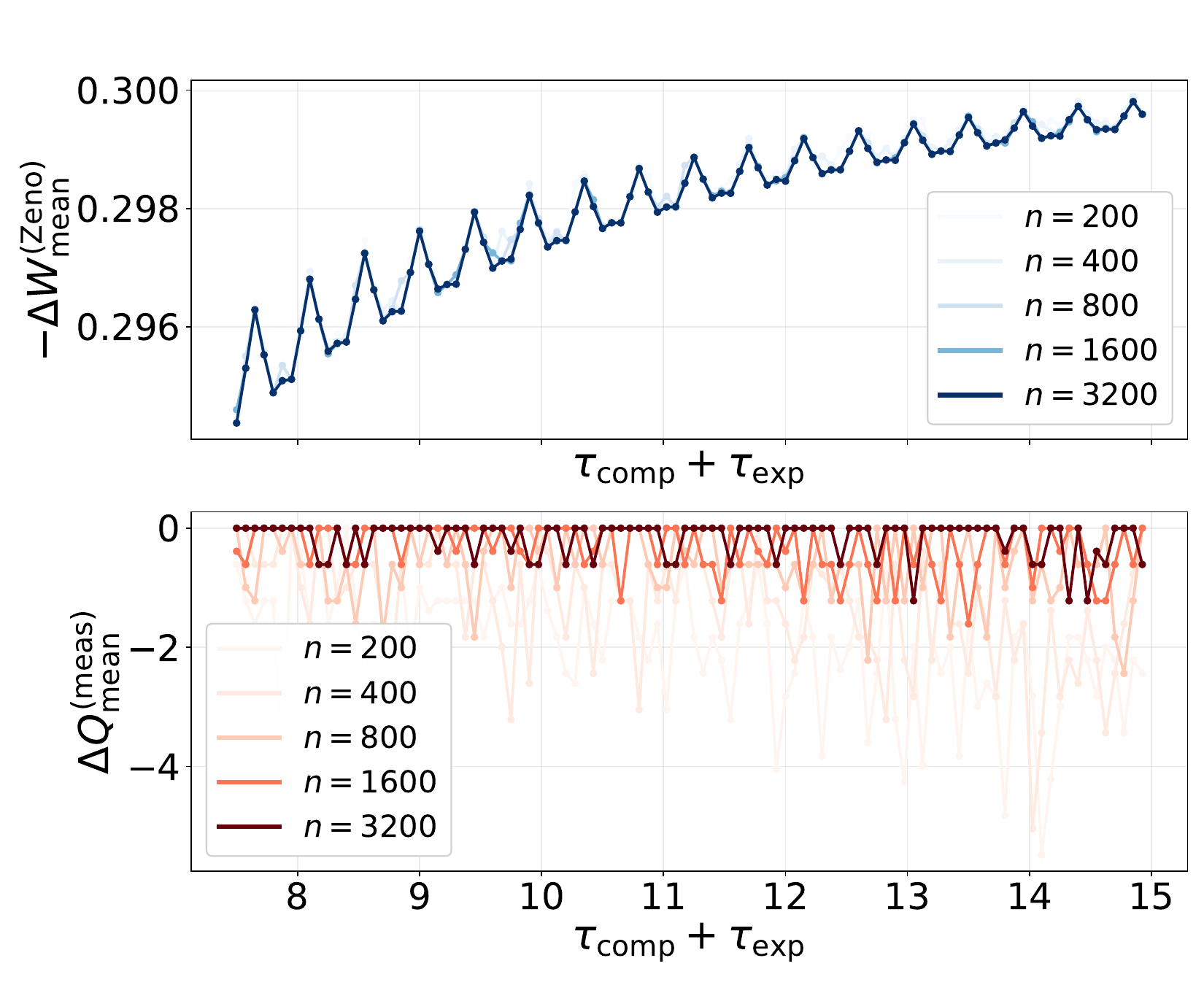}
    \caption{\textbf{Averaged extracted work and dissipated heat under a Zeno drive.} The results are shown for 50 trajectories. We plot the averaged work extracted under a Zeno drive with coupling strength $\Gamma = 20$ as the number of measurements $n$ is increased. Each point is averaged over $N_{\rm traj}=50$ trajectories. In the simulation, the compression stroke duration $\tau_{\rm comp}$ is varied between $5$ and $10$, while the expansion stroke duration is set to $\tau_{\rm exp} = \tau_{\rm comp}/2$, as before. In the Zeno limit $n \to \infty$, the dissipated heat approaches zero.  Parameters: $\Gamma = 20, N_{\rm traj} = 50$, $\tau_{\rm comp} = 2 \tau_{\rm exp}$, $\omega=\omega_L=1$,  $\Omega_0 = 3.01105$, $T_{\rm c} = 0.5$, $ T_{\rm h}=3$, and  $\gamma_{\rm h} = \gamma_{\rm c}=0.5$. }
    \label{fig:otto_averaged_zeno}
\end{figure}

\subsection{Costs associated with the lubricant qubit}\label{sec:lubricant_qubit_costs}

After revisiting our approach for estimating work, we turn our attention to the analysis of energetic costs which come from the introduction and manipulation of the lubricant system when operating the engine. In the following, we specifically investigate the work cost of turning the strong interaction on and off, and whether there is a decoupling penalty due to correlations between the working medium and the lubricant at the end of work strokes. We also examine the cost of resetting the lubricant state at the beginning of each work stroke, and the energetic cost of attaining Zeno stabilization on the lubricant during the Zeno-driven work strokes by accounting for the energetic cost of implementing a frequent monitoring of the lubricant's state.

\subsubsection{Strong interaction and correlations between the system and lubricant}

In this subsection, we show that the coupling and decoupling costs associated with lubricating the QHE according to the strong-coupling alone, such as the one considered in Fig.~\ref{fig:full_figure_test}(a)-(c), are significant, thus motivating the frequent monitoring of the lubricant. We start by noticing that, in the limit $\Gamma\to\infty$ the final composite state $\rho_{SL}(t_1)$ exhibits vanishingly small entanglement generation, as can be seen from Fig.~\ref{fig:log_neg}. We quantify entanglement~\cite{vedral1997quantifying,vedral1998entanglement} using the logarithmic negativity~\cite{plenio2005logarithmic}
\begin{equation}
    E_N(\rho_{SL}(t_1)) = \log_2 \Vert \rho^{T_L}_{SL}(t_1)\Vert_1 
\end{equation}
where $T_L$ here denotes the partial transpose relative to the lubricant system, and $\Vert \cdot \Vert_1$ denotes the trace norm. In our setting, since both the lubricant and the system are single-qubit states $E_N(\rho_{SL}(t_1)) = 0$ implies that $\rho_{SL}(t_1)$ is separable. Starting with $\rho_{SL}(t_0)$ as given by Eq.~\eqref{eq: initial_state_compression}, Fig.~\ref{fig:log_neg} shows the entanglement of the final state after evolving the system according to $U_{\rm tot}(t_1)$ for different durations of $\tau_{\rm comp}$. For sufficiently large $\Gamma$, entanglement is strongly suppressed even for short compression times.

\begin{figure}[t]
    \centering
    \includegraphics[width=\columnwidth]{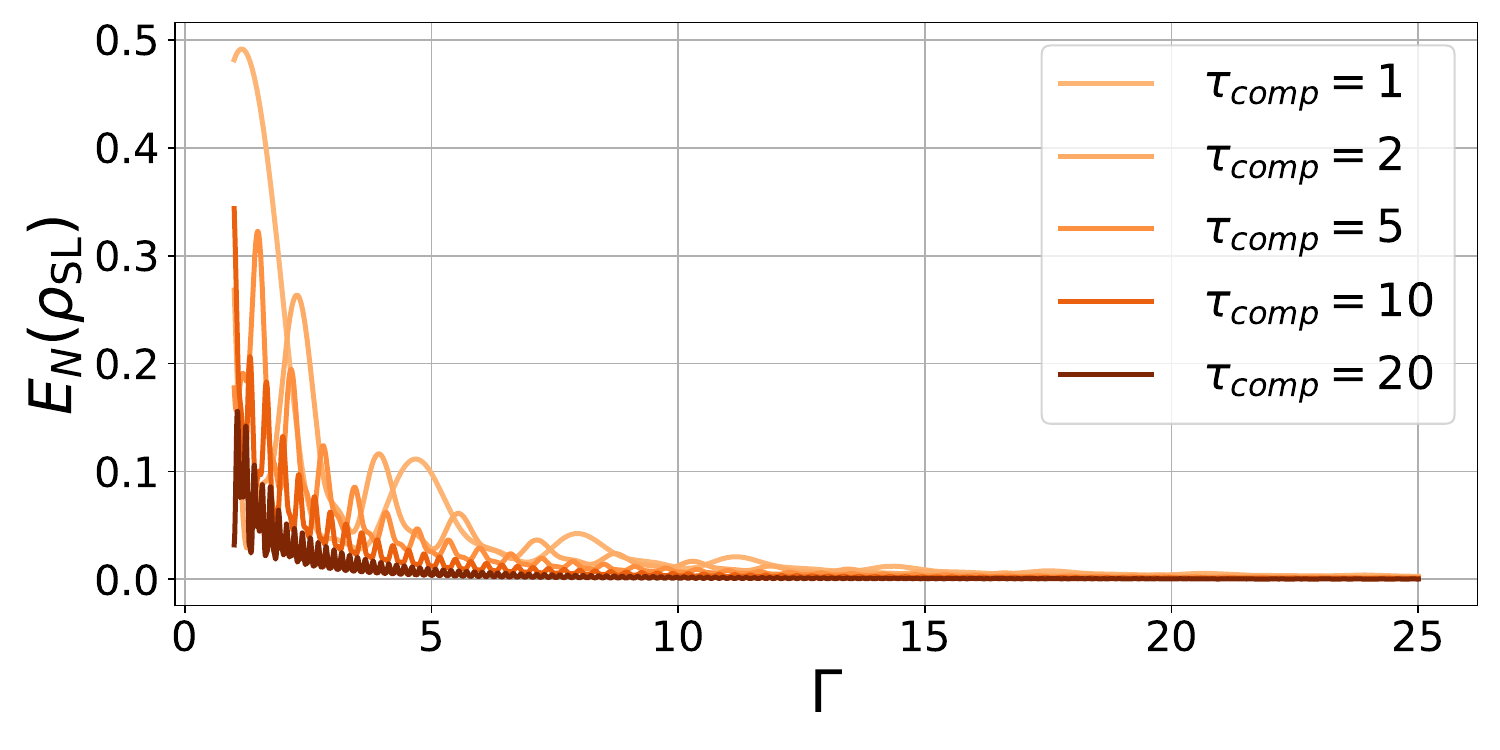}
    \caption{\textbf{Logarithmic negativity of the final state $\rho_{SL}(t_1)$ after the compression stroke.} Different colored curves stand for the varying duration of the compression stroke, as given in the legend. The coupling strength $\Gamma$ starts with $\Gamma = 1$ and changes in $0.01$ increments.  Parameters: $\omega=\omega_L=1$, $\Omega_0=5, T_{\rm c}=0.5$, and no measurements are present.}\label{fig:log_neg}
\end{figure}

However, entanglement is not the only relevant carrier of information. Provided that the working system starts in a mixed state, lubricant and working medium might share classical correlations, making the final state separable, but not necessarily a product state. Therefore, if the working fluid and the lubricant are correlated at the end of the work strokes, there is a work cost associated with decoupling them. Intuitively, this work cost can be viewed as the additional energy trapped in the correlations generated by the interaction~\cite{muller2018correlating,molitor2020stroboscopic}, and is given by
\begin{align}
    &\Delta W_{SL}^{\rm dec}(t_f,t_i)=  \nonumber \\  &\, \,\,\mathrm{Tr} \Big[H_{SL}(t_f) \rho_{SL}(t_f)\Big] - \mathrm{Tr}\Big [H_{SL}(t_i)\big(\rho_{S}(t_i)\otimes \rho_L(t_i)\big) \Big], \label{eq:on_off_cost}
\end{align}
where $H_{SL}(t)$ is the interaction term given by Eq.~\eqref{eq:H_SL(t)}, $\rho_{SL}(t_f) = U_{\rm tot}(t_f)\bigr (\rho_S(0) \otimes \vert + \rangle \langle + \vert \bigr) U_{\rm tot}(t_f)^\dagger$ is the joint state after respective work stroke, and $\rho_{L}(t_i)= \vert +\rangle \langle + \vert$ at the beginning of each work stroke. Here, $U_{\rm tot}(t)$ is generated by $H_{\rm tot}$ from Eq.~\eqref{eq:total_hamiltonian_SL}.

We are investigating two lubrication protocols  in parallel, (i) when only strong coupling between two qubits is present, and (ii) when the lubricant qubit is experiencing Zeno monitoring. Starting with the former, which has been analyzed in Subsection~\ref{sec:strong_coupling_COST}, we use Eq.~\eqref{eq:on_off_cost} to calculate the decoupling work cost, and show the results in Fig.~\ref{fig:dec_cost}. Even though the amount of entanglement decreases as $\Gamma \to \infty$, the presence of classical correlations makes $\Delta W^{\rm dec}_{SL}$ grow linearly in magnitude with the coupling strength $\Gamma$. Since the strong-coupling protocol requires $\Gamma$ to be large, a diverging decoupling penalty renders this approach ineffective. This provides a useful example of how a protocol that appears energetically favorable at first sight---as can be seen in Figs.~\ref{fig:full_figure_test}(a)-(c)---may exhibit significant drawbacks once more realistic energetic costs are taken into account.

\begin{figure}[t]
    \centering
    \includegraphics[width=\columnwidth]{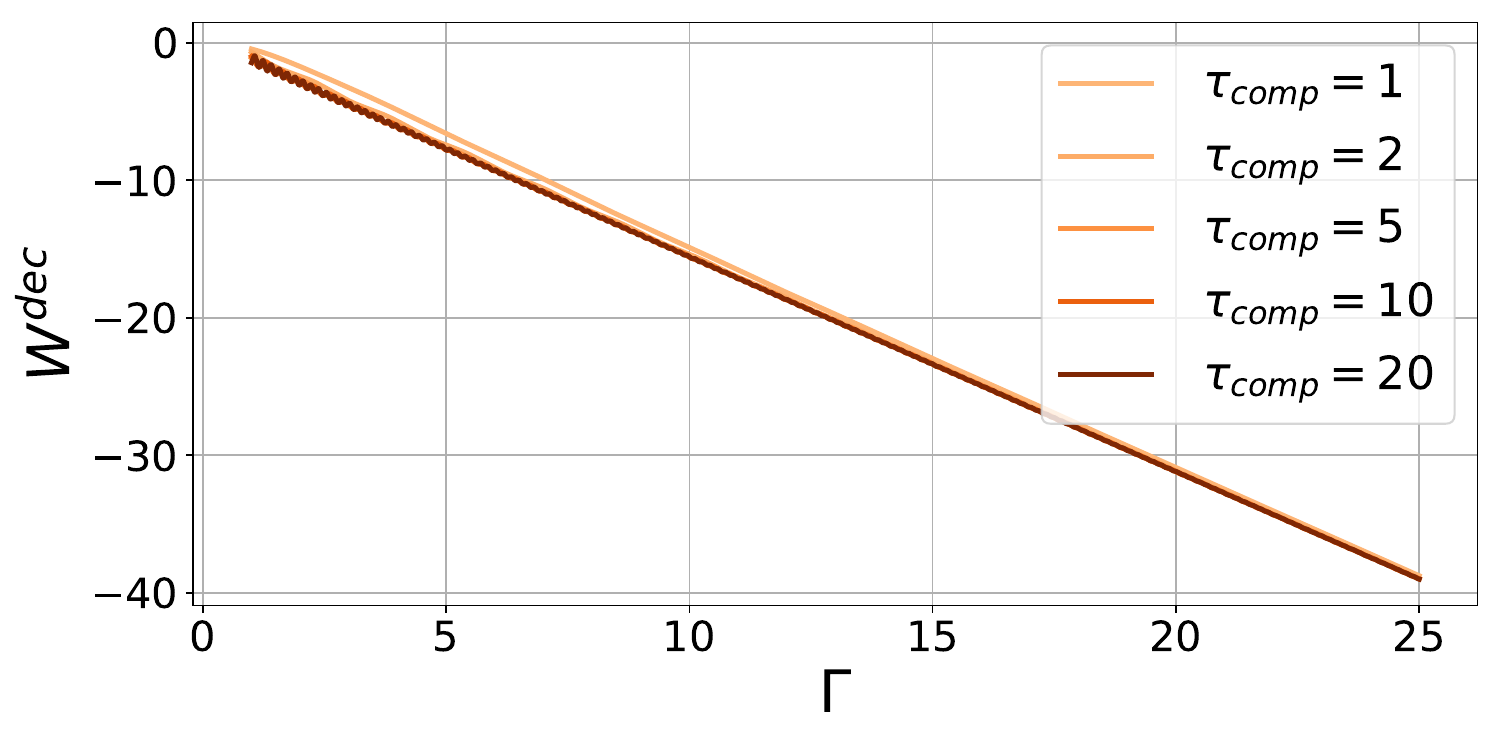}
    \caption{\textbf{Decoupling work cost for the compression stroke.} Using Eq.~\eqref{eq:on_off_cost}, we calculate the energy trapped in the interaction between the working system and the lubricant as a function of the coupling strength $\Gamma$ that starts with $\Gamma=1$ and changes in $0.01$ increments. Different colored curves stand for the varying duration of the compression stroke, as given in the legend. Parameters: $\omega=\omega_L=1$, $\Omega_0=5, T_{\rm c}=0.5$, and no measurements are present.}\label{fig:dec_cost}
\end{figure}

Figure~\ref{fig:dec_cost} shows why strong coupling alone is not a satisfactory thermodynamic solution: although it suppresses coherence generation, the decoupling cost grows with $\Gamma$. This motivates the addition of Zeno stabilization through frequent measurements. In the ideal Zeno regime, the final state after a work stroke is again a product state, so no interaction energy remains trapped in residual correlations when the coupling is switched off. In that regime, the decoupling penalty is therefore removed. 

The same conclusion can be reached analytically. In the ideal Zeno regime the
joint state stays in product form at the start and at the end of the compression stroke,
\begin{align*}
    \rho_{SL}(0) &= \underbrace{\left(p_0 \vert 0 \rangle \langle 0 \vert
      + p_1 \vert 1 \rangle \langle 1 \vert \right)}_{\equiv\, \rho_S(t=0)} \otimes \vert + \rangle \langle + \vert, \\
    \rho_{SL}(t_1) &= \underbrace{\left(p_0 \vert 0_{t_1} \rangle \langle 0_{t_1} \vert
      + p_1 \vert 1_{t_1} \rangle \langle 1_{t_1} \vert \right)}_{\equiv \, \rho_S(t=t_1)} \otimes \vert + \rangle \langle + \vert,
\end{align*}
where $\vert 0_{t}\rangle$ and $\vert 1_{t}\rangle$ are given by Eqs.~\eqref{eq:0t} and~\eqref{eq:1t}. Since $\langle + \vert X \vert + \rangle = 1$ and $\rho_{SL}(t) = \rho_S(t) \otimes \vert + \rangle \langle + \vert$, each term of Eq.~\eqref{eq:on_off_cost} reduces to
\begin{equation}\label{eq:dec_cost_reduces}
    \mathrm{Tr}\big[H_{SL}(t)\,\rho_S(t) \otimes \vert + \rangle \langle + \vert\big]
    = \Gamma \, \mathrm{Tr}\big[R(t)\,\rho_S(t)\big].
\end{equation}
Moreover, the ideal Zeno drive is an effective transitionless drive which keeps the populations fixed in the instantaneous
eigenbasis, so $\mathrm{Tr}[R(t)\rho_S(t)] = p_0 - p_1$ at both instants $t=0$ and $t=t_1$.
\begin{equation}
    \Delta W_{SL}^{\rm dec}(t_1,0) = \Gamma(p_0-p_1) - \Gamma(p_0-p_1) = 0.
\end{equation}
The same argument applies to the expansion stroke, with $\rho_S(t_2)$ diagonal
in $\{\vert 0_{\rm h} \rangle, \vert 1_{\rm h} \rangle\}$, which is the
instantaneous eigenbasis at $t_2$, and with populations $q_0$ and $q_1$:
\begin{equation}
    \Delta W_{SL}^{\rm dec}(t_3,t_2) = \Gamma(q_0-q_1) - \Gamma(q_0-q_1) = 0.
\end{equation}

\subsubsection{On the relation between decoupling and on/off work cost}

In our protocol, the interaction is switched off during the thermalization strokes so that the baths act only on the working medium. This avoids spurious heat exchange with the lubricant and keeps the thermalization model close to the standard Otto-cycle setting. If the interaction were left on, correlations between working medium and lubricant could alter both the thermalization time and the heat currents, as observed in  Ref.~\cite{weber2023thermodynamiccostspuredephasing}. Equation~\eqref{eq:on_off_cost} also quantifies the energetic cost associated with switching the strong interaction on and off, as shown below.

Turning the interaction on reduces the accessible parts of the total Hilbert space, which can be interpreted as projection to the subspace defined by the interaction Hamiltonian $H_{\rm int}(t)$. To estimate the associated work cost, we calculate the energy difference as the interaction term is added at the beginning of work strokes
\begin{align}
    &W_{\rm on}(t_i)= \nonumber \\
    &\,\,\, \mathrm{Tr}\Big[ H_{\rm tot}(t_i) \rho_{SL}(t_i) \Big] - \mathrm{Tr}\Big[ \big( H_S(t_i) + H_L \big) \rho_{SL}(t_i) \Big],
\end{align}
and the difference in energy as the interaction is removed at the end of work strokes
\begin{align}
    &W_{\rm off}(t_f) = \nonumber \\ 
    &\,\,\, \mathrm{Tr}\Big[ \big( H_S(t_f) + H_L \big)\rho_{SL}(t_f) \Big] - \mathrm{Tr}\Big[ H_{\rm tot}(t_f) \rho_{SL}(t_f) \Big],
\end{align}
with respect to the joint initial and final states. The total Hamiltonian is from Eq.~\eqref{eq:total_hamiltonian_SL}, while $H_S(t) + H_L = H_S(t) \otimes \mathbbm{1} + \mathbbm{1} \otimes H_L$. Using linearity of the trace and the initial product state, the above equations are reduced to
\begin{align}
    W_{\rm on}(t_i) &= \mathrm{Tr}\Big[ H_{SL}(t_i) \big( \rho_{S}(t_i)\otimes \rho_L(t_i) \big) \Big] \\
    W_{\rm off}(t_f) &= -\mathrm{Tr}\Big[ H_{SL}(t_f) \rho_{SL}(t_f) \Big],
\end{align}
yielding the negative of Eq.~\eqref{eq:on_off_cost} by addition. The opposite minus sign can be understood as follows: Eq.~\eqref{eq:on_off_cost} quantifies energy trapped in correlations, while the negative of it, obtained by analyzing the on/off cost, is energy required to \emph{create} those correlations. Since we have shown that in the ideal Zeno limit the decoupling cost vanishes, we also have that the on/off work cost vanishes in the same limit. 

Although a Zeno drive helps resolve the penalty due to strong coupling, it also suggests that one may be shifting the decoupling cost toward two other costs: the cost of implementing the Zeno monitoring and the dissipative cost of entropy production due to the inherent irreversibility of the measurement process. We now proceed to consider both in detail.

\subsubsection{Frequent monitoring} 

A substantial literature has been devoted to understanding the energetic requirements of implementing quantum measurements~\cite{takahiro2009minimal,jacobs2009secondlaw,Reeb2014,abdelkhalek2016energy, Esposito2011,kieu2006quantum}, and the ways in which they can be exploited in thermodynamic settings~\cite{buffoni2019quantummeasurementcooling}. From that perspective, ideal projective measurements are often treated as limiting operations whose exact implementation would require unbounded resources, because preparing pure states needed for a perfectly sharp measurement is itself costly~\cite{Guryanova2020,masanes2017general,scharlau2018quantumhornslemma,taranto2023landauer,clivaz2019unifying}. These observations are important, but they concern asymptotic costs (i.e.~the cost required for a \emph{perfect} pure state preparation), which do not by themselves rule out the use of measurement-based protocols in finite-resource settings. That is because in these cases the precise \emph{rates} at which approximately pure states are prepared play a more relevant quantitative role.

Our goal here is more modest and more practical. We do not attempt to model the microscopic details of the measuring device. Instead, we follow Ref.~\cite{abdelkhalek2016energy}, and use an effective energetic accounting that separates the measurement-induced change in the system from the cost of resetting the measurement register. This allows us to estimate how the monitoring cost scales in the regime relevant for the Zeno effect.

Following Ref.~\cite{abdelkhalek2016energy}, we assume that the measurement apparatus consists of two conceptual stages: a projective measurement described by a PVM $\{P_\ell\}_\ell$, and a reset step that restores the measurement register so that the process can be repeated. The reset is treated as an erasure process coupled to a bath---the \emph{resetting bath}---at inverse temperature $\beta$. In this framework, the total energetic cost of a projective measurement contains both a measurement-induced energy change and a Landauer-type contribution associated with resetting the register. Thus, the total energy change due to a projective measurement reads
\begin{equation}\label{eq:proj_M_cost}
    \Delta E_{\rm proj} = \Delta Q^{(\rm meas)} + \frac{1}{\beta}H(\{p_\ell\}_\ell),
\end{equation}
where $$\Delta Q^{(\rm meas)} = \mathrm{Tr}\bigr [H(\rho' - \rho) \bigr]$$ is the heat cost contribution due to the mapping onto a final state $$\rho' = \sum_\ell P_\ell \rho P_\ell,$$
and $\beta^{-1} H(\{p_\ell\}_\ell)$ is the dissipation cost quantified by the Shannon entropy $H$ of the distribution of Born rule probabilities $p_\ell = \mathrm{Tr}[P_\ell \rho P_\ell]$ from the measurement, weighted by the temperature of the resetting bath. 

Note that this is a different picture than the one we have  considered in Sec.~\ref{sec:revisiting}, where we have estimated the heat difference relative to the selective projected state $P_\ell \rho P_\ell/p_\ell$. The authors from Ref.~\cite{abdelkhalek2016energy} are not assigning a separate energy change to a particular outcome $\ell$ as they are interested in defining the energetic cost of the measurement step as a physical process, averaged over outcomes. Since we have already observed in Sec.~\ref{sec:revisiting} that a selective description has effectively no heat cost in the limit $n\to \infty$ we now consider the heat cost relative to this distinct approach.

To estimate this cost in our case we proceed similarly to Sec.~\ref{sec:revisiting}. We begin by applying the linear approximation to the total Hamiltonian $H_{\rm tot}$ from Eq.~\eqref{eq:total_hamiltonian_SL}, which generates system-lubricant evolution between two consecutive measurements. We keep the contributions up to the order $\delta t^2$, where $\delta t = t_{(k+1)}-t_{(k)},$ for all $k=1,\ldots,n$, is the time duration between measurements, leading to the approximate evolution
\begin{align}
    &\rho_{SL}(t_{(k)}) = \rho_{SL}(t_{(k-1)}) - i \delta t \bigr [H_{\rm tot}(t_{(k)}), \rho_{SL}(t_{(k-1)}) \bigr] \nonumber \\&- \frac{\delta t^2}{2}\Bigr [H_{\rm tot}(t_{(k)}), \bigr[H_{\rm tot}(t_{(k)}), \rho_{SL}(t_{(k-1)}) \bigr] \Bigr] + O(\delta t^3) \nonumber.
\end{align}
Note that the above approximation holds under the assumption that $\delta t \ll 1/\max(\omega,\omega_L,\Omega_0,\Gamma) = 1/\Gamma$. The lubricant state immediately before a measurement is $\rho_L \equiv \rho_L({ t_{(k)}})=\mathrm{Tr}_S\bigr [\rho_{SL}({t_{(k)}})\bigr ]$. In our case, we implement a local measurement of the form $\mathbbm{1} \otimes \vert \ell \rangle \langle \ell \vert$ with $\ell \in \{+,-\}$. Since the aim of the protocol is to keep the lubricant in the $\vert + \rangle$ state, the final state $\rho'_{SL}(t_{(k)})$ relevant to our calculation is given by 
\begin{equation}
\rho'_{SL}(t_{(k)}) = \sum_{\ell \in \{+,-\}} \mathbbm{1}\otimes \vert \ell \rangle \langle \ell \vert \, \rho_{SL}(t_{(k)}) \mathbbm{1}\otimes \vert \ell \rangle \langle \ell \vert.
\end{equation}
In this case, the change in heat due to the $k$-th measurement reads
\begin{equation}
    \delta Q_k^{\rm (meas)} = \mathrm{Tr}\left[H_{\rm tot}(t_{(k)}) (\rho_{SL}'(t_{(k)}) - \rho_{SL}(t_{(k)}))\right].
\end{equation}
Assuming the first change from $t_{(1)} = 0$ until $t_{(2)} = t_{(1)}+\delta t \sim \tau_{\rm comp}/n$, after a series of manipulations we end up with (keeping only up to second order terms) 
\begin{equation}
    \delta Q_2^{(\rm meas)} \approx  \frac{1}{2}(2p_1-1)\delta t^2 \Gamma \omega_L^2 = \frac{1}{2}(p_1-p_0)\frac{\tau_{\rm comp}^2\Gamma \omega_L^2}{(n-1)^2}.
\end{equation}
As in Sec.~\ref{sec:revisiting}, the measurement at $t_{(1)}$ acts on the prepared state $\vert + \rangle \langle + \vert$ and therefore contributes
$\delta Q_1^{(\rm meas)} = 0$.

Assuming that $\delta Q_2^{(\rm meas)} \approx \delta Q_k^{(\rm meas)}$ for the remaining $k=2,\ldots,n$, we end up with a total contribution of
\begin{equation}
    \Delta Q^{(\rm meas)} = \sum_{k=2}^n \delta Q_k^{(\rm meas)} \approx \frac{1}{2}(p_1-p_0)\frac{\tau_{\rm comp}^2\Gamma \omega_L^2}{n-1},
\end{equation}
which goes to zero as $\delta t \,\Gamma \ll 1$, i.e. as $n \to \infty$ at a faster rate than $\Gamma$.

It remains to estimate the entropy-related part of the monitoring cost, namely the contribution proportional to $\beta^{-1}\,H(\{p_\ell\}_\ell)$. The total contribution is given by $\beta^{-1}\,\sum_{i=2}^nH(\{p_{\ell_i}\}_{\ell_i})$ with each term dependent on incremental changes. We proceed in the same spirit as above by considering the first measurement interval $t^{(1)}=0$ to $t^{(2)}=t^{(1)}+\delta t$, expanding for small $\delta t$, and then assuming that the resulting leading-order expression is representative of the subsequent intervals. In this case, provided that we start with the lubricant in state $\vert + \rangle \langle + \vert$ we have that 
\begin{equation}\label{eq:Zeno_probability_jump}
    p_{\ell_2 \,= \,+} = 1 - \frac{\delta t^2 \omega_L^2}{4}+O(\delta t^3) \approx 1-\varepsilon_{\ell_2}, 
\end{equation}
leading to 
\begin{align*}
    H(\{p_{\ell_2\,=\,+},p_{\ell_2 \,=\,-}\}) &=- \sum_{\ell_2 \in \{+,-\}}p_{\ell_2} \ln p_{\ell_2}\\
    &= - (1-\varepsilon_{\ell_2})\ln(1-\varepsilon_{\ell_2})-\varepsilon_{\ell_2}\ln(\varepsilon_{\ell_2})\\
    &\approx \varepsilon_{\ell_2}(1-\ln(\varepsilon_{\ell_2})) \\
    &=\frac{\delta t^2 \omega_L^2}{4}\left(1- \ln \left(\frac{\delta t^2 \omega_L^2}{4}\right)\right),
\end{align*}
which converges to zero as $\varepsilon_{\ell_2}\to 0$. Assuming that for all $k$ the entropy term is approximately the same we end up with the total contribution of Zeno monitoring to be
\begin{align*}
&\sum_{i=2}^n H(\{p_{\ell_i}\}_{\ell_i}) \approx \\&(n-1)\, \frac{\tau_{\rm comp}^2 \omega_L^2}{4(n-1)^2}\left(1- \ln \left(\frac{\tau_{\rm comp}^2 \omega_L^2}{4(n-1)^2}\right)\right)\\
&\stackrel{n\to \infty}{\to}0.
\end{align*}

\subsubsection{Cost of preparing the lubricant}

The analysis above shows that the cost of monitoring the lubricant vanishes in
the Zeno limit, provided the lubricant is in the $\ket{+}$ state to begin with.
It remains to account for the cost of that preparation. We take the starting
point to be the thermal state $\rho_L^{\rm th}$ at inverse temperature $\beta$,
since this is the state the lubricant reaches for free by equilibrating with its
environment~\cite{lostaglio2019introductory}; any other starting state has higher free energy and would therefore
be cheaper to convert, so this choice is conservative.

Closely following the discussion present in Ref.~\cite{abdelkhalek2016energy}, we model the preparation as the following sequence. First, an $X$-basis
projective measurement is performed on $\rho_L^{\rm th}$, yielding $\pm$ with
equal probability. Second, if the outcome is $-$, a rotation about $z$ maps
$\vert - \rangle$ to $\vert + \rangle$; this unitary commutes with $H_L$ and is
therefore energetically free. Third, the register storing the outcome is reset so
that the apparatus can be reused. Applying Eq.~\eqref{eq:proj_M_cost} to this
sequence gives
\begin{align}
    \Delta E_{\rm proj}^{\rm th\rightarrow \ket{+}}
    &= \frac{\omega_L}{2}(p_1-p_0) + \frac{1}{\beta}\ln{2} \nonumber \\
    &= \frac{\omega_L}{2}\tanh{\left(\frac{\beta \, \omega_L}{2}\right)}
     + \frac{1}{\beta}\ln{2},
\end{align}
where $p_0$ and $p_1$ are here the thermal populations of the \emph{lubricant}
relative to $H_L$.

We note that in this case the ambiguity between the selective accounting of
Sec.~\ref{sec:revisiting} and the non-selective one of
Eq.~\eqref{eq:proj_M_cost} is immaterial. Because $H_L \propto Z$ while the
measurement is in the $X$ basis, both post-measurement states carry the same
energy, $\langle + \vert H_L \vert + \rangle = \langle - \vert H_L \vert -
\rangle = 0$, so the selective heat is independent of the outcome and coincides
with its average.

This cost can be a one-off cost if Zeno monitoring is continuous while the engine is working. In this case, $\Delta E_{\rm proj}^{\rm th\rightarrow \ket{+}}\approx0.7274$ with $\beta=2$ and $\omega_L=1$. This cost decreases as $\beta \rightarrow \infty$, leaving us with $\Delta E_{\rm proj}^{\rm th\rightarrow \ket{+}} = \frac{\omega_L}{2}$, and further minimized by reducing energy gap of the lubricant $\omega_L$.

\subsubsection{Entropy production}

Beyond the explicit energetic changes associated with measurement and strong coupling, the protocol also produces entropy through the irreversible removal of correlations between the working medium and the lubricant. In the previous subsection we accounted for entropy production associated with the measuring device following Ref.~\cite{abdelkhalek2016energy}. We now estimate the internal entropy production generated within the system-lubricant composite itself~\cite{landi2021irreversible}. 

In our protocol, each measurement disrupts the correlations that have built up during the preceding unitary pulse. Following Refs.~\cite{landi2021irreversible,Elouard2017}, we quantify the associated stochastic entropy production by comparing the probability of a forward measurement trajectory with that of the corresponding backward trajectory. 

Let $\rho_{SL}(t)$ denote, as before, the state of the composite system. We are performing selective measurements on the lubricant at regular time intervals $\delta t_{(k)}$ in the $X$ basis $\{\ket +,\ket -\}$. This procedure leads to a stochastic trajectory $\vec \ell = (\ell_1, \ldots , \ell_n)$. As previously mentioned, between measurements, the joint state $\rho_{SL}$ evolves according to a unitary generated by the total Hamiltonian from Eq.~\eqref{eq:total_hamiltonian_SL}. Assuming that within that interval $t_{(k)} \leq t \leq t_{(k+1)}$, the initial state is given by $\rho_{SL}^{(\ell_k)}(t_{(k)})$ the probability that after the unitary pulse of $\delta t$ we observe $\ell_{(k+1)}$ is given by
\begin{align}
    p({\ell_{k+1} \vert \ell_{k}}) =  \mathrm{Tr} \left[{U_k \, \rho_{SL}^{(\ell_k)}(t_{(k)})\, U_k^\dagger \,  \mathbbm{1}\otimes \vert \ell_{k+1} \rangle \langle \ell_{k+1} \vert} \right]
\end{align}
where we have momentarily used the notation $U_k \equiv U_{\rm tot}(t_{(k+1)},t_{(k)})$ to simplify the expression. In words, this gives the conditional probability of observing $\ell_{k+1}$ at the end of the interval $t_{(k)} \leq t \leq t_{(k+1)}$ given that the outcome observed at the start was $\ell_k$. The  conditional probability of obtaining a certain trajectory, is thus given by 
\begin{equation}\label{eq:forward probability}
    \mathcal{P}_F[\,\vec \ell \,\,] = p({\ell_n \vert \ell_{n-1}}) \cdots \,\,p({\ell_2 \vert \ell_1})\,\,p({\ell_1})
\end{equation}
where $p({\ell_1}) = 1$ if $\ell_1 = +$ in our case, since the initial state of the lubricant system is $\vert +\rangle \langle + \vert $. The probability that we observe $\ell_n$ after these $n-1$ steps is then given by
\begin{equation}
    p(\ell_n) = \sum_{\ell_1,\,\ldots,\,\ell_{n-1}} \mathcal{P}_F[\,\vec \ell \,\,].
\end{equation}

We denoted the conditional probability in Eq.~\eqref{eq:forward probability} as $\mathcal{P}_F[\,\vec \ell \,\,]$ since this is interpreted as the probability of the ``forward'' trajectory.  To quantify the amount of entropy produced by following the trajectory $\vec \ell$, we consider the reverse protocol. Starting from the final state after the observation of $\ell_n$ we evolve backwards by implementing the same procedure as above but changing $U_k \to U_k^\dagger$. This time-reversed trajectory, that we denote as $\overleftarrow{\ell}$ leads to the conditional probability given by 
\begin{equation}
    \mathcal{P}_B[\, \overleftarrow \ell \,\,] = p(\ell_{1} \vert \ell_{2})\, \cdots \,p(\ell_{n-1} \vert \ell_{n})\,  p(\ell_n).
\end{equation}
The stochastic entropy production defined via forward and backwards  trajectories~\cite{crooks1998nonequilibrium,evans1993probability,landi2021irreversible}
\begin{equation}
    \sigma(\ell_n) = \ln \left(\frac{\mathcal{P}_F[ \, \vec \ell \,\,]}{\mathcal{P}_B[\,\overleftarrow{\ell}\,]} \right)
\end{equation}
simplifies, in this case, to~\cite{Elouard2017,landi2021irreversible}
\begin{equation}\label{eq:trajectory_entropy_prod}
    \sigma(\ell_n) = \ln \left(\frac{p(\ell_1)}{p(\ell_n)}\right). 
\end{equation}
As we start with the lubricant in state $\vert + \rangle \langle + \vert$ we have that $p(\ell_1) = 1$, and therefore the entropy produced along the trajectory is, assuming that the result is $\ell_n = +$, given by 
\begin{align}\label{eq:Zeno_entropy_prod}
    \sigma (\ell_n = +) &= -\ln(p(\ell_n)) \\
    &= - \ln\left( \sum_{\ell_1,\ldots,\,\ell_{n-1}} p(\ell_n|\ell_{n-1})\cdots \,p(\ell_2|\ell_1)p(\ell_1) \right)\nonumber \\
    &\approx - \ln [(1-\varepsilon_{\ell_2})^{n-1}] \nonumber \\
    &\stackrel{n\to \infty}{\to}0 \nonumber,
\end{align}
where we approximated the probability of each jump to be the same and arbitrarily small in the Zeno limit. Note, importantly, that above we have used the fact that in the Zeno limit $\varepsilon_{\ell_2}=O(1/n^2)$, i.e. we have used the fundamental quadratic short-time behavior that is the key ingredient to the Zeno effect. 

\subsection{Cost of imperfect thermalization}\label{subsec:imperfect_therm}

So far we have concentrated on accelerating the work strokes. Once lubrication makes those strokes effectively short, the thermalization stages become the dominant contribution to the cycle time in a passive thermalization scheme. It is therefore necessary to examine explicitly when finite-time thermalization becomes the main bottleneck for power generation.

In the weak-coupling regime, the qubit asymptotically approaches a Gibbs state in the long-time limit, whereas in the strong-coupling regime the joint system-bath state relaxes towards a \emph{global} Gibbs state~\cite{breuer2002theory, rivas2012open, Trushechkin2022, manzano2020lindblad}. This global Gibbs state may incur additional correlations between the bath and the working qubit, which in turn lead to high decoupling costs. For an analysis of heat conductance in the limit of strong couplings between the working system and a heat bath we refer to Refs.~\cite{rivas2020strong,ivander2022strong,talkner2020strongcouplingthermo}, and to Ref.~\cite{burgarth2019generalized} for a description of the evolution in the context of strong damping Zeno limits $\gamma_{\rm h}, \gamma_{\rm c} \to \infty$. 

In the remainder of this subsection we adopt the standard weak-coupling description of thermalization. In that regime the main limitation is not an additional switching penalty but the long time required for the working medium to relax close to the relevant Gibbs state. Under the rotating-wave approximation in the system-bath interaction, the local Gibbs state of the working qubit, $\rho_{\beta}^{\rm th} = e^{-\beta H_S}/\mathrm{Tr}[e^{-\beta H_S}]$ is a steady state of the reduced dynamics. In deriving the Lindblad equation, two assumptions are particularly relevant for our analysis: (i) the initial system-bath state is uncorrelated, and (ii) due to the \emph{Born approximation}~\cite{breuer2002theory}, the system-bath state remains approximately factorized at all times, i.e.
\begin{equation}
    \rho_{SB}(t) \approx \rho_S(t) \otimes \rho_{B,\beta}^{\rm th}
\end{equation}
(where $\rho_{B,\beta}^{\rm th}$ is the thermal state of the bath), up to terms of order $O(\gamma^3)$, where $\gamma$ is the coupling strength. This approximation reflects the fact that correlations built up between the system and the bath remain negligible precisely because the coupling is weak, which is satisfied in the regime where $\gamma \ll \omega $. As a consequence, the final state is effectively a product state, so no additional cost is required to decouple the two subsystems.

One can characterize thermalization either through population relaxation or through the contraction of the full state toward the steady state. The former arises from the rate equation for the populations, and the latter follows from a rigorous definition of the mixing time based on the Liouvillian gap~\cite{rivas2012open, breuer2002theory}. For the single-qubit model considered here, these notions lead to the same characteristic decay rate, so we use the trace-distance relaxation timescale in the discussion below. Strictly speaking, this is a relaxation time rather than an exact finite-time thermalization time, since the approach to the Gibbs state is asymptotic.

Recall that, for any primitive GKLS semigroup $\{e^{t \mathcal{L}}\}_t$ the distance of a state at time $t$ from a steady state of the dynamics is
\begin{equation}
    \Vert  \rho(t)-\rho^{\rm ss} \Vert _1 \leq Ce^{-\lambda_{\rm gap}t},
\end{equation}
where $C$ is a constant, and $$\lambda_{\rm gap}= \min_{\lambda \in \mathrm{spec}(\mathcal{L})}(-\mathfrak{Re}(\lambda))$$ is known as the \emph{Liouvillian  gap}~\cite{mori2020resolving}. For the Lindbladians $\mathcal{L}_{\rm c}$ and $\mathcal{L}_{\rm h}$ discussed in Sec.~\ref{sec:background_section} a direct calculation (applied to our single-qubit case) yields the gaps
$$\lambda_{\rm gap}^{\rm (cold)} = \frac{\gamma_{\rm c}}{2}(2 n_{\rm c}+1), \quad \lambda_{\rm gap}^{\rm (hot)} = \frac{\gamma_{\rm h}}{2}(2 n_{\rm h}+1),$$
where $n_{{\rm c}} = (e^{\omega \beta_{\rm c}}-1)^{-1}$ and $n_{\rm h} = (e^{\Omega \beta_{\rm h}}-1)^{-1}$ are the thermal photon number at the inverse temperature of respective baths.

Under these considerations, the thermalization time for the cold thermalization stroke then is  approximately given by
\begin{equation}
    \tau_{\rm cold} \approx \frac{1}{\lambda_{\rm gap}^{\rm (cold)}} \approx \frac{2}{\gamma_{\rm c}(2n_{\rm c}+1)},
\end{equation}
and similarly for the hot thermalization stroke
\begin{equation}
    \tau_{\rm hot} \approx \frac{1}{\lambda_{\rm gap}^{\rm (hot)}} \approx \frac{2}{\gamma_{\rm h}(2n_{\rm h}+1)}.
\end{equation}

\begin{figure}[t]
    \centering
    \includegraphics[width=\columnwidth]{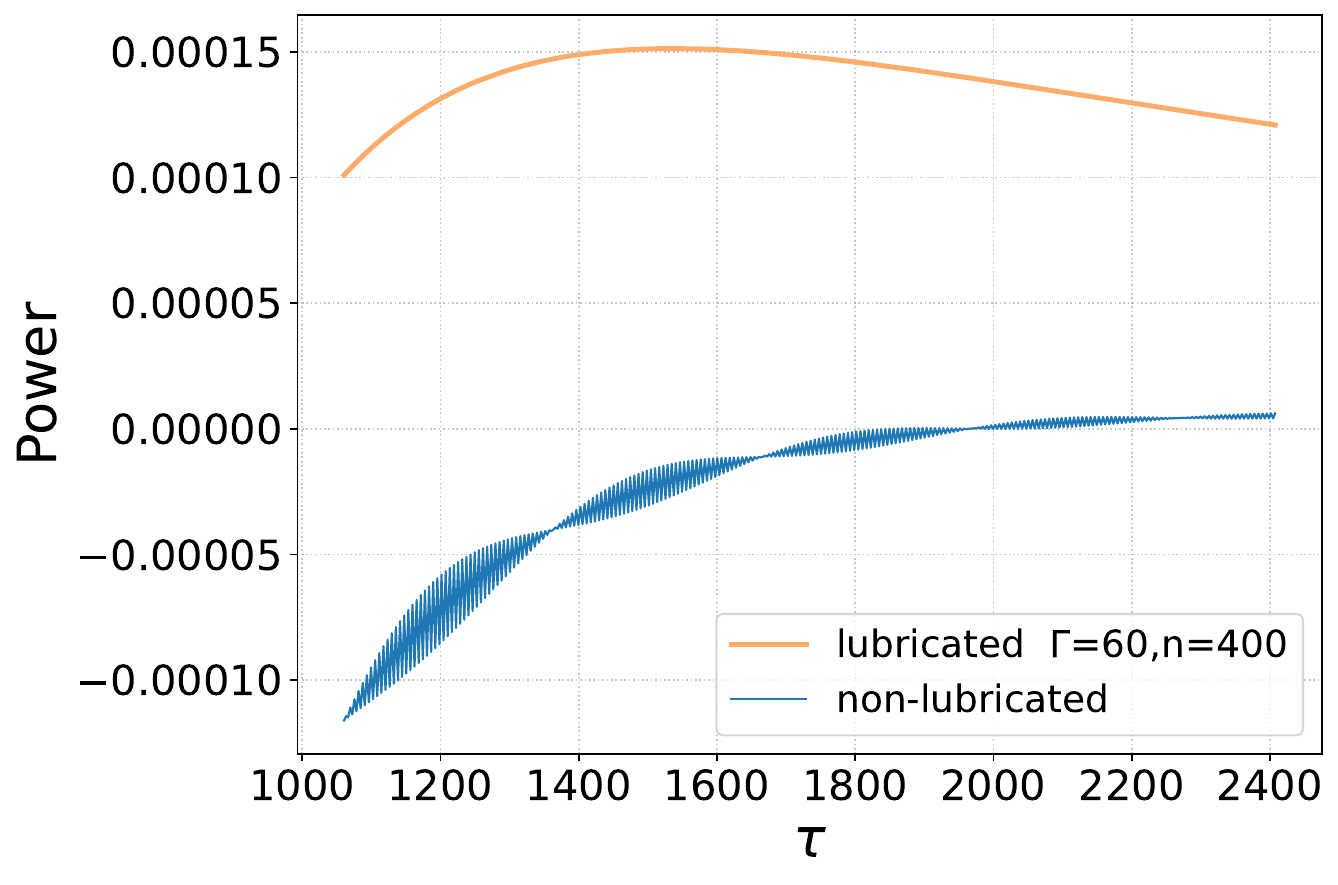}
    \caption{\textbf{Power generation in the regime of weak-coupling between working system and thermal baths.} Within our lubrication scheme, the work strokes can be made significantly shorter than the thermalization strokes. In our passive thermalization setting, choosing $\gamma_{\rm h}=\gamma_{\rm c}=0.005$ makes the thermalization timescale orders of magnitude larger than the duration of the work strokes. The cycle duration is varied by changing $\tau_{\rm hot}$ in the range $300 \leq \tau_{\rm hot} \leq 800$, with $\tau_{\rm cold}=2\tau_{\rm hot}$. For $\tau \lesssim 2000$, the power remains negative in the non-lubricated case, since more work is performed during the expansion stroke than is extracted during compression. Parameters: $\tau_{\rm comp} = 5$, $\tau_{\rm exp} = 2.5$, $\Gamma = 60$, $n = 400$, $\omega = \omega_L = 1$, $\Omega_0 = 3.01105$, $T_{\rm c} = 0.5$, $T_{\rm h} = 3$, and $\gamma_{\rm c} = \gamma_{\rm h} = 0.005$.}\label{fig:imperfect_thermalization}
\end{figure}

We note, however, that in our case the weak-coupling limit requires $\gamma_{\rm h}, \gamma_{\rm c} \ll 1$, i.e.\ significantly smaller than the value $\gamma_{\rm h}=\gamma_{\rm c}=0.5$ we have considered in our previous simulations. In Sec.~\ref{sec:Zeno_assisted_quantum_heat_engine}, we chose these values so as to obtain comparable cycle times for the heat and work strokes, which makes the improvements more easily visible in the simulations. If instead we take, for instance, $\gamma_{\rm h},\gamma_{\rm c}\approx 0.005$, then $\tau_{\rm cold}$ and $\tau_{\rm hot}$ must be two orders of magnitude larger. This is shown in Fig.~\ref{fig:imperfect_thermalization}. As seen there, lubrication continues to significantly improve power generation. Here, we consider the case in which the initial and final states during the work strokes are qubit states that have been only imperfectly thermalized. In this regime, the non-lubricated engine cannot generate positive power until much longer thermalization times, whereas the lubricated engine can.

\subsection{Driving cost}

In our model, the time-dependent Hamiltonian $H_{\mathrm{tot}}(t)$ in Eq.~\eqref{eq:total_hamiltonian_SL}  may be understood as a modification of $H_{S}(t)$ arising from an externally implemented control field interacting with the working system and the lubricant. Our earlier semiclassical formulation of work, from Eq.~\eqref{eq:total_work_explicit}, ignored the energetic cost of that control, as is standard in many treatments of QHEs ~\cite{Dann2023}. Here, however, we estimate the cost of maintaining the strong coupling between the working qubit and the lubricant captured by $H_{SL}(t)$, the interaction term in Eq.~\eqref{eq:total_hamiltonian_SL}.~\footnote{Note that we do not consider the cost of driving associated to the term $\frac{f(t)}{2}X$ from $H_S(t)$. This term is simply taken as \emph{contributing} to the net work change from Eq.~\eqref{eq:total_work_explicit}. Similarly, if we interpret the working system and the lubricant as a new composite lubricant system, thus accounting for the work generated using Eq.~\eqref{eq:s.c.total_work_joint}, a semiclassical formulation of work would not consider this as a cost but merely as being part of the Hamiltonian terms contributing to the net work changes.}

Naturally, the control cost has been considered in the context of counter-diabatic drive~\cite{guery-odelin_shortcuts_2019, zheng2016counterdiabaticcost,Abah2019,torrontegui2017energy,Campbell2017,Kiely2022,santos2015superadiabatic,abah2018performance}, where a \emph{control-cost functional} is defined
\begin{equation}\label{eq:cost-functional}
    C_{\rm drive}(\tau) = \nu \int_0^{\tau} \left\Vert  H_{\rm drive}(t) \right \Vert_{\rm F}  \,\mathrm d t,
\end{equation}
such that $\nu$ is a setup-dependent constant~\footnote{The value of $\nu$ can depend on the physical implementation, field geometry, impedance, losses, and possible recovery of supplied energy, as it describes the relation between the mathematical modeling of the drive and the actual energetic cost of physically maintaining the drive during the stroke.}, $\tau$ is the duration of the drive, $\Vert B\Vert_{\rm F} =\sqrt{\mathrm{Tr}\left[B^\dagger B\right]}$ denotes the Frobenius norm, and $H_{\rm drive}$ is a control operator of interest. 

We note that the control-cost functional in Eq.~\eqref{eq:cost-functional} is nonzero even for an operator proportional to identity.  This relates to the question of whether a thermodynamic cost refers to the work done on the quantum system---which assigns zero cost to a global phase shift---or the energy required to maintain the control field, which is nonzero whenever the drive is active regardless of its effect on the state, as with the cost functional. The second interpretation is adopted here: the field costs energy to maintain regardless of its effect on the quantum state. Thus, $C_{\rm drive}$ can be understood as a semiclassical hardware-dependent estimate of the power required to maintain the strong externally generated system-lubricant coupling (associated, in some settings, to the entropy produced by the classical apparatus generating the control fields~\cite{Kiely2022}), not as a fundamental functional estimate.

This distinction is also reflected in how the cost functional is used in the quantum-control and shortcut-to-adiabaticity literature~\cite{Campbell2017,Kiely2022,santos2015superadiabatic}. In the context of counter-diabatic driving, this quantity estimates the power required for achieving a transitionless dynamics, while taking the initial time-dependent drive to be free. The Hamiltonian is usually written as $H_{\rm CD}(t) = H_0(t) + \tilde H(t)$, so only the cost of implementing $\tilde H$ with classical control fields is taken into account. In our case, the transitionless (adiabatic) dynamics \emph{emerges} from the Zeno limit. Because of that, the meaningful choice here is not to consider the counter-diabatic drive, but the interaction term responsible for its emergence given by $H_{SL}(t)$.

In our case, the working qubit and lubricant experience $H_{\rm tot}(t)$, which implies that the resulting control drive term is given by the time-dependent contributions in this Hamiltonian, namely 
\begin{align}
    H_{SL}(t) &= \Gamma  \cos(\theta_t) Z \otimes X+ \Gamma \sin(\theta_t) X \otimes X. 
\end{align}
We apply Eq.~\eqref{eq:cost-functional} to our protocol taking into account that the Zeno driving evolution happens in pulses interrupted by measurements on the lubricant. Hence, the total duration of the compression stroke is divided in $n-1$ intervals, leading to $\delta t=\tau_{\rm comp}/(n-1)$
\begin{align}
    C_{\rm drive}(\tau_{\rm comp}) &= \nu\sum_{k=1}^{n-1} \int_{t_{(k)}}^{t_{(k)}+\delta t} \mathrm dt \,  \Vert H_{SL}(t)\Vert_{\rm F}  \nonumber \\
    &= \nu\sum_{k=1}^{n-1} \int_{t_{(k)}}^{t_{(k)}+\delta t} \mathrm dt \,  2 \Gamma \nonumber \\
    &= 2\nu\,\Gamma (n-1) \delta t \nonumber\\
    &= 2\nu\,\Gamma \tau_{\rm comp}.
\end{align}
The same result is obtained for the expansion stroke, substituting $\tau_{\rm comp}$ by $\tau_{\rm exp}$. This control-cost is such that $\nu$ has dimensions of inverse time (equivalently, of energy in our units). We can relate it to the cost of the power required to generate the strong coupling drive by dividing it by the total cycle time. In that case, we note that the contributions per cycle are given by
$$C_{\rm drive}(\tau_{\rm comp})+C_{\rm drive}(\tau_{\rm exp}) \approx 2 \nu \Gamma (\tau_{\rm comp}+\tau_{\rm exp}),$$ 
which needs to be deducted from the theoretical prediction in $P_{\rm tot} = -\Delta W_{\rm tot}/\tau$ (using the work from Eq.~\eqref{eq:total_work_implicit}) leading to
\begin{equation}\label{eq:power_driving_cost}
    P_{\rm net} 
    = P_{\rm tot} - \frac{2 \nu \Gamma (\tau_{\rm comp}+\tau_{\rm exp})}{\tau}.
\end{equation}
This points to a trade-off between power generated and the cost of implementing the drive.

Among the costs considered in this section, the external driving cost is the most serious practical limitation. Successful lubrication requires $\Gamma$ to be large, while the control overhead scales as $O(\Gamma)$. This creates a nontrivial trade-off between improving the work stroke and paying for the control field that implements it.

One possible mitigation is hardware dependent: the prefactor $\nu$ in Eq.~\eqref{eq:power_driving_cost} is setup-specific and may be small in favorable experimental platforms. A more promising route is to work in architectures where the relevant strong coupling is native to the device rather than generated dynamically by an external field. Examples include nearby NV centers with geometry-fixed dipole coupling or strongly interacting Rydberg platforms~\cite{carmele2014opto}. In such cases the overhead need not scale as the cost of a continuously applied dynamical drive, and the conceptual advantages of the Zeno-assisted protocol may be retained.

\begin{mainresultsbox}{\hypertarget{box_2}{BOX 2: Summary of thermodynamic costs}}
\vspace{0.15cm}
{\parbox[h]{\dimexpr\linewidth\relax}{

\Checkmark~\textbf{Coupling and decoupling:} Vanishes in the ideal Zeno limit, where the joint state remains a product state and the populations are preserved. Away from that limit, strong coupling alone leaves a penalty growing with
$\Gamma$.\newline 

\Checkmark~\textbf{Switching on and off:} Equal in magnitude to the decoupling cost, and therefore vanishing under the same conditions.\newline 

\Checkmark~\textbf{Entropy production:} The stochastic entropy production along a measurement trajectory vanishes as $n \to \infty$, as expected from the evolution being unitary within a fixed Zeno subspace.\newline 

\Checkmark~\textbf{Zeno jumps:} Each jump between Zeno subspaces dissipates an energy of order $\Gamma$, but such jumps are suppressed as $O(1/n)$, so their averaged contribution vanishes in the Zeno limit.\newline 

\Checkmark~\textbf{Lubricant preparation:} Bounded by
$\omega_L/2 + \beta^{-1}\ln 2$ and incurred once (in the ideal Zeno limit) rather than per cycle, so it is amortized over repeated operation. It decreases with increasing $\beta$ and decreasing the lubricant gap $\omega_L$.\newline 

\Checkmark~\textbf{Imperfect thermalization:} Not specific to lubrication, but it limits the attainable power once the work strokes are short. Lubrication still improves performance markedly; accelerating the heat strokes is left to future
work.\newline 

\XSolidBrush~\textbf{External driving:} If the strong coupling is generated by an external field, and the lubricant is not counted as part of the working medium, maintaining the drive costs $O(\Gamma)$. This overhead is hardware dependent and absent when the coupling is native to the platform.\newline 

\usym{2753}~\textbf{Costs not considered:} Limitations from quantum speed limits and from Zeno-adapted thermodynamic uncertainty relations, as well as practical overheads such as non-ideal or finite-duration measurements, are left for future
work. \newline 
}}\par
\vspace{0.4cm}
\end{mainresultsbox}

\section{Discussion and outlook}\label{sec:discussion_and_outlook} 

In this work we have introduced a Zeno-assisted lubrication protocol for stroke-based finite-time quantum heat engines. The central idea is to suppress the coherence-generating transitions responsible for quantum friction by coupling the working medium to an auxiliary lubricant and using strong coupling together with frequent monitoring to confine the dynamics to a suitable Zeno subspace. In the ideal Zeno limit, we show the resulting effective Hamiltonian contains the counter-diabatic term required for transitionless driving during the work strokes.

We have also examined the extent to which this idealized control advantage survives once a more complete thermodynamic accounting is performed. This required revisiting both the definition of extracted work and the additional costs associated with switching, monitoring, control, and imperfect thermalization. Our findings are summarized in~\secondbox. The resulting picture is more nuanced: the protocol does provide a clear control-theoretic shortcut to adiabaticity, but its practical value depends strongly on which physical resources are counted and on how the protocol is implemented.

Within the set of costs analyzed here, the dominant practical limitation is the cost of implementing the strong-coupling drive through an external control field. Other challenges---such as isolating the lubricant during the heat strokes, performing sufficiently sharp and frequent measurements, and controlling timing errors---are also relevant, but they are more platform specific. This suggests that the protocol is most promising in devices where strong coupling is naturally available or where measurement and control overhead can be made comparatively small.

It is also important to stress that the present protocol does not eliminate the cost of suppressing nonquasistatic excitations; rather, it transfers the burden from slow quasistatic driving to an externally engineered Zeno stabilization mechanism involving strong coupling and monitoring. In this sense, the protocol should be understood not as a free circumvention of finite-time thermodynamic constraints, but as a redistribution of control and thermodynamic resources. In other words, our protocol does not simply propose a Zeno-assisted engine but articulates a coherent framework for understanding Zeno stabilization as an emergent route to counter-diabatic thermodynamic control. The Zeno control should itself be regarded as a thermodynamic resource. From this perspective, the protocol exemplifies how structured control resources can reshape the accessible thermodynamic trajectories of quantum devices.

\subsection{Relation with previous work}

Our results can be placed within the broader context of strong-coupling quantum thermodynamics~\cite{talkner2020strongcouplingthermo}. Much of that literature has focused on strong coupling between a working medium and a bath, and on the consequences of such coupling for efficiency, power, and nonequilibrium thermodynamic structure~\cite{rivas2020strong,gonzalez2024hamiltonian}. Here the strong interaction plays a different role: it is used as a control resource that, together with monitoring, generates an effective shortcut to adiabaticity. The auxiliary system is therefore not merely dressing the thermodynamics; it is actively shaping the work-stroke dynamics.

Other works have studied work extraction in strongly coupled settings, but typically not in the presence of a time-dependent coupling generated by the drive itself~\cite{perarnau2018strong}, and without the realization that in certain regimes these limits result in the formation of an effective QZD. Importantly, Ref.~\cite{perarnau2018strong} showed that, for cyclic work extraction under their assumptions, strong coupling can have a detrimental effect on the work output. Our results do not contradict that conclusion, since the assumptions are different: in our protocol both the strong coupling term and the effective unitary acting on the working medium are time dependent, and the role of the auxiliary system is not merely to dress the thermodynamics but to generate an effective shortcut to adiabaticity. 

A significant conceptual link also exists between our protocol and dephasing-assisted QHEs. In particular, we have compared our results with the dephasing-bath model studied by some of us in Ref.~\cite{weber2023thermodynamiccostspuredephasing}, and found that the optimal regimes numerically identified there can also be interpreted as arising from strong-coupling manifestations of the QZD. This connection provides a unifying perspective on two seemingly distinct lubrication strategies---one based on dissipative engineering and the other on quantum measurement---in specific regimes and suggests that both can be understood as different realizations of constrained dynamics in an appropriate subspace.

Our protocol is also closely related to the notion of \emph{Zeno dragging}, namely the use of frequent measurements or kicks to steer a system adiabatically within a constrained subspace~\cite{hacohen2018incoherent,lewalle2024optimal}. This approach provides a closely related dual picture, which is obtained when one uses frequent \emph{time-dependent} strong measurements on a composite system and its lubricant rather than on the working medium itself~\cite{burgarth2013nonabelian}. This duality suggests that a broader family of Zeno-based control protocols may be available for finite-time quantum thermodynamics. Recently, Barontini~\cite{barontini2025quantumzeno} has considered this form of Zeno-assisted implementation where, instead of lubricating an existing engine, they substitute the usual unitary drive with a Zeno-dragging protocol. Also there, one finds that such types of engines incur promising improvements on the power and efficiency with negligible costs. 

On a different note, our findings may be of independent interest to the quantum-control and shortcuts-of-adiabaticity literature as we provide one such shortcut induced by the Zeno effect. Again, our approach uses different manifestations of the Zeno effect, and a related perspective may be taken where a form of shortcut can be obtained via Zeno-dragging (i.e. implementing time-dependent projections) in such a manner that they carry out the generator of the quasistatic transporter as recently discovered by Ref.~\cite{delcampo2026shortcutsadiabaticityadaptivequantum}.

\subsection{Future directions}
A natural next question is whether other manifestations of QZD can be combined in a similarly useful way. In particular, strong damping can also induce effectively unitary Zeno dynamics~\cite{burgarth2019generalized}. This raises the possibility of replacing explicit projective monitoring by suitably engineered continuous dissipation, thereby connecting the present protocol more directly to dephasing-assisted lubrication schemes such as Ref.~\cite{weber2023thermodynamiccostspuredephasing}.

One may then interpret the resulting dynamics as exhibiting two distinct manifestations of continuous QZD: one arising from strong coupling and one arising from time-dependent strong damping~\cite{burgarth2019generalized,diMeglio2024timedependent}. From this viewpoint, the strong Markovian environment acts as an effective measuring apparatus, and induces Zeno subspaces defined by the peripheral projections of the Lindbladian~\cite{burgarth2019generalized}. An interesting open question is whether these two manifestations of the QZD can also be mathematically shown to produce a counter-diabatic drive that is useful for lubricating work strokes. For the ideas discussed here to apply to the type of QHEs considered in Sec.~\ref{sec:relationship_between_decoherence_and_zeno}, it would be necessary to generalize the theorems of Refs.~\cite{burgarth2019generalized,burgarth2022oneboundtorulethem} to the case of \emph{time-dependent} and \emph{unbounded} operators generating quantum dynamical semigroups more typical in open quantum systems. Recently, Ref.~\cite{burgarth2026rotatingwavesecularapproximationsopen} has structured the first results in this direction investigating the time-dependent case. 

More broadly, the present framework invites several (some mentioned in the last point of~\secondbox). One obvious direction is to move beyond a single-qubit working medium and study multilevel, multipartite, or many-body working substances, where both the structure of the Zeno subspaces and the finite-time control problem become richer~\cite{hacohen2018incoherent,burgarth2014exponential}. Another is to relax idealizations of the monitoring process---for example by allowing nonprojective measurements, finite measurement duration, or explicit timekeeping costs~\cite{woods2023autonomous,lautenbacher2025physically}. These generalizations are likely to be essential for any realistic comparison between different Zeno-based thermodynamic control strategies.

It will also be important to analyze our protocol in the context of thermodynamic uncertainty relations~\cite{Horowitz2019,Hasegawa2019,Van_Vu2025,Prech2025}. Numerically, we observed that the fluctuations decrease as the Zeno regime becomes stronger (as expected, since in the Zeno limit the evolutions are unitary within a certain Zeno subspace), which may indicate an unexplored link between enhanced stabilization and entropy production. Quantifying this trade-off would be valuable for assessing the fundamental performance of Zeno-assisted thermal machines and for determining the consistency bounds governing their operation. More generally, our findings suggest that a careful study of the interplay between QZD, fluctuation theorems~\cite{campisi2010fluctuation,campisi2011influence}, and thermodynamic uncertainty relations~\cite{tanvanvu2023thermodynamic,Van_Vu2025} is a particularly promising direction for future work. 

A related concern arises from the notion of a quantum speed limit~\cite{deffner2017qsls}. Concretely, these impose fundamental bounds relating the norm of a time-dependent Hamiltonian with how fast we can drive a quantum system. These have deep connections with thermodynamic uncertainty relations~\cite{tanvanvu2023thermodynamic,vovantuan2020unified,yunoki2025quantumspeedlimitquantum} and will effectively impact fundamental power generation in Zeno-assisted QHEs such as those we have introduced.

In conclusion, the quantum Zeno dynamics offers a clean and conceptually distinctive route to lubricating finite-time quantum heat engines. By recasting friction suppression as a subspace-control problem, it connects quantum thermodynamic performance with strong coupling, monitoring, and shortcut-to-adiabaticity techniques in a unified way. Whether this route is practically advantageous will depend on the physical platform and on the cost model adopted, but as a theoretical framework it opens a promising line of inquiry at the interface of measurement, control, and quantum thermodynamics.\newline

\section*{Author contributions and AI use disclosure}

SM and RW contributed equally to this work. SFH and MBP supervised the project and reviewed the manuscript. AI tools were used for language editing and polishing, \texttt{LaTeX} assistance, identifying typographical and other minor errors, surveying the literature, assisting with coding, and generating and refining plots. AI tools were also used during the development of some initial calculations and to assist with routine or tedious mathematical manipulations. The main theoretical results of the paper were developed without AI assistance. The same applies to the theoretical results presented in Appendix~\ref{app:friction_terms_work}. 
The initial manuscript and all initial text were written by the authors and subsequently refined through human-AI collaboration. All figures were created by the authors without AI-assistance and are not AI-generated. We also used \href{https://www.wolfram.com/mathematica/}{\texttt{Mathematica}} for analytical calculations. 

\section*{Replication material}

All code used for implementing the simulations, as well as the required data for replicating our findings, are openly available and  can be found in a \href{https://doi.org/10.5281/zenodo.22734267}{\texttt{Zenodo}} repository~\cite{zenodo2026zenoassisted}.

\begin{acknowledgments}
We would like to thank Raphael Weber, Rudi B. P. Pietsch, Adolfo del Campo, Som Kanjilal, Franklin L. S. Rodrigues, Filipa Peres, Lennart Bosch, Michele Campisi, Giovanni Di Meglio, Nicolás T. Domingues, and Henning Kirchberg for fruitful discussions. The authors acknowledge support from the European Research Council (ERC) under the European Union's Horizon 2020 research and innovation programme (Grant Agreement No. 856432, HyperQ), the EU-Project SPINUS (Grant No. 101135699) and from the Alexander von Humboldt Foundation.
\end{acknowledgments}

\bibliography{bibliography}

\begin{appendix}

\section{Quantum friction contribution to work}\label{app:friction_terms_work}

Let us consider that the state of our working medium is pure, and evolves according to the Schrödinger equation
\begin{equation}
    i\frac{\mathrm d}{\mathrm dt}\vert \psi(t)\rangle = H(t) \vert \psi(t) \rangle.  
\end{equation}
Here, $t \mapsto H(t): \mathbbm{C}^2 \to \mathbbm{C}^2$ is a time-dependent self-adjoint operator for every $t$. With respect to its instantaneous eigenbasis $\{\vert 0_t\rangle, \vert 1_t\rangle \}$, where 
\begin{equation}
    H(t) \vert i_t\rangle = (-1)^i\frac{\varepsilon(t)}{2}\vert i_t\rangle 
\end{equation}
we can write the state of the system as \begin{equation*}
\vert \psi(t) \rangle = c_0(t)\vert 0_t\rangle + c_1(t) \vert 1_t\rangle. 
\end{equation*}

The evolution of the coherences $c_0(t)$ and $c_1(t)$ is given by the equations
\begin{align}\label{eq:coupled_coef}
    \dot c_1 &= - \langle 1_t\vert \dot 0_t\rangle c_0-\langle 1_t\vert \dot 1_t\rangle c_1+i\frac{\varepsilon(t)}{2}c_1, \\
    \dot c_0 &= - \langle 0_t\vert \dot 1_t\rangle c_1-\langle 0_t\vert \dot 0_t\rangle c_0 -i\frac{\varepsilon(t)}{2}c_0.
\end{align}
These equations follow directly from the Schrödinger equation after expanding the state in the instantaneous eigenbasis. They show that the terms proportional to
\begin{equation}
    \langle 1_t\vert \dot 0_t\rangle = -\frac{\dot \theta_t}{2}
\end{equation}
couple the two amplitudes. Above, we have used Eq.~\eqref{eq:instantaneous_states} from the main text.  This coupling is precisely what generates coherence. In the quasistatic regime, $\dot \theta_t \approx 0$, these terms become negligible. 

Let us consider a time-dependent two-level Hamiltonian with $H_S(t)$ which is a combination of constant Pauli term $Z$ and time-dependent Pauli $X$ term, so that the instantaneous eigenbasis $\{\vert 0_t \rangle ,\vert 1_t\rangle \}$ can be viewed as a rotation relative to some angle $\theta_t$ as we have considered in the main text. We begin from the instantaneous spectral decomposition
\begin{equation}
    H_{S}(t) =  \frac{\varepsilon(t)}{2}R(t)
\end{equation}
noting that both eigenvalues $\left\{\pm \frac{\varepsilon(t)}{2} \right\}$ and eigenvectors $\{\vert 0_t \rangle, \vert 1_t \rangle\}$ are time-dependent. Let us momentarily denote $\varepsilon_0(t) = \frac{\varepsilon(t)}{2}$  and $\varepsilon_1(t) =-\frac{\varepsilon(t)}{2}$. 
 
Taking a time derivative on both sides of $H_S(t)\ket{i_t} = \varepsilon_i(t)\ket{i_t}$ one obtains the following:
\begin{align}
    \frac{\mathrm d}{\mathrm dt}\big( H_{S}(t)\ket{i_t} \big) &= {\dot H_{S}(t) \ket{i_t}} + H_{S}(t) \frac{\mathrm d}{\mathrm dt}\ket{i_t} \nonumber \\
    =\frac{\mathrm d}{\mathrm dt} \big(\varepsilon_i(t)\ket{i_t} \big) &= \dot \varepsilon_i(t) \ket{i_t} + \varepsilon_i(t) \frac{\mathrm d}{\mathrm dt} \ket{i_t}.
\end{align}
We now isolate the ${\dot H_{S}(t)}$ term 
\begin{equation}\label{eq:term_H_dot_comp}
    \dot H_{S}(t)\,|i_t\rangle = \dot\varepsilon_i(t)|i_t\rangle + \varepsilon_i(t)\frac{\mathrm d}{\mathrm dt}|i_t\rangle - H_{S}\frac{\mathrm d}{\mathrm d t}|i_t\rangle,
\end{equation}
and turn our attention to the expression for obtaining the work during the compression stroke which is an integral over $\mathrm{Tr}[\rho \dot H_{S}]$ as given by Eq.~\eqref{eq:work}. We can thus calculate this term as
\begin{widetext}
\begin{align*}
    \mathrm{Tr}[\rho \dot H_{S}]
        &= \sum_{i}\langle i_t| \rho \dot H_{S}|i_t\rangle = \sum_{i,j}\langle i_t| \rho |j_t\rangle \langle j_t| \dot H_{S}|i_t\rangle = \sum_{i,j} \rho_{i_t j_t} \langle j_t| \dot H_{S}|i_t\rangle \\
        &\stackrel{\text{Eq.}~\eqref{eq:term_H_dot_comp}}{=} \sum_{i,j} \rho_{i_tj_t} \langle j_t| \Big(\dot\varepsilon_i(t)|i_t\rangle + \varepsilon_i(t)\frac{\mathrm d}{\mathrm dt}|i_t\rangle - H_{S}\frac{\mathrm d}{\mathrm d t}|i_t\rangle\Big) \\
        &= \sum_{i,j} \rho_{i_t,j_t}  \Big( \delta_{ij}\dot\varepsilon_i(t) + \langle j_t| \varepsilon_i(t) \frac{\mathrm d}{\mathrm dt}| i_t\rangle - \langle j_t|H_{S}\frac{\mathrm d}{\mathrm dt}|i_t\rangle \Big) \\
        &= \sum_{i,j} \rho_{i_tj_t} \Big( \delta_{ij}\dot\varepsilon_i(t) + \langle j_t| \varepsilon_i(t) \frac{\mathrm{d}}{\mathrm dt}|i_t\rangle - \langle j_t| \varepsilon_j(t) \frac{\mathrm d}{\mathrm{d}t}|i_t\rangle \Big) \\
        &= \sum_{i,j} \rho_{i_tj_t} \Big( \delta_{ij}\dot\varepsilon_i(t) + (\varepsilon_{i}(t) - \varepsilon_{j}(t)) \langle j_t|\frac{\mathrm d}{\mathrm dt}\vert i_t\rangle \Big) = \sum_{i,j} \rho_{i_tj_t} \Big( \dot \varepsilon_{i}(t) \delta_{ij} + \varepsilon_{ij}(t) \langle j_t| \frac{\mathrm d}{\mathrm dt}\vert i_t\rangle \Big),
\end{align*}
\end{widetext}
where $\varepsilon_{ij}(t) \equiv \varepsilon_i(t) - \varepsilon_j(t)$ is the difference between the eigenvalues. For our case, we have that $\varepsilon_i(t) = (-1)^i\varepsilon(t)/2$ and therefore  
\begin{align}
    \varepsilon_{ij}(t) &= ((-1)^i -(-1)^j)\frac{\varepsilon(t)}{2} \nonumber \\
    &= \left\{\begin{array}{ll}
       0  & \text{ if } \,\,i=j \\
       (-1)^i\varepsilon(t)  & \text{ if } \,\, i\neq j
    \end{array} \right.
\end{align}
In what follows, we assume that there are no level crossings, meaning that for all $t$ we have $\varepsilon_{ij}(t) \neq 0$ if $i \neq j$. We split the sum into two cases: 1) $i=j$ for which the second term is zero, as we are subtracting the same eigenvalues, and 2) $i\neq j$ for which the first term is zero since $\ket{i_t}$ and $\ket{j_t}$ are orthogonal. Hence, we arrive at the following expression
\begin{align}
    \mathrm{Tr}[\rho \dot H_{S}] &= +\rho_{0_t0_t} \frac{\dot{\varepsilon}(t)}{2}
     - \rho_{1_t1_t} \frac{\dot \varepsilon(t)}{2}\nonumber \\& +\rho_{0_t1_t}\varepsilon(t)\langle 1_t \vert \dot 0_t \rangle -\rho_{1_t0_t}\varepsilon(t) \langle 0_t \vert \dot 1_t\rangle \nonumber \\
     &= \frac{(\rho_{0_t0_t} - \rho_{1_t1_t})}{2} \dot \varepsilon(t) -\frac{(\rho_{0_t1_t} +\rho_{1_t0_t})}{2}\varepsilon(t)\dot \theta_t
\end{align}
from which we obtain
\begin{equation}
    \mathrm{Tr}[\rho \dot H_{S}]=\frac{\dot{\varepsilon}(t)}{2}(\rho_{0_t 0_t}-\rho_{1_t 1_t})-\varepsilon(t)\,\dot \theta_t \,\mathfrak{Re}[\rho_{0_t 1_t}]. 
\end{equation}
Above, we have used that $\langle 1_t \vert \dot 0_t \rangle = - \dot \theta_t/2 = - \langle 0_t \vert \dot 1_t \rangle $, with $\theta_t$ being the instantaneous rotation angle in $xz$-plane of $\{\vert 0_t \rangle, \vert 1_t \rangle\}$ relative to $\{\vert 0 \rangle, \vert 1\rangle \}$, and $\mathfrak{Re}[\rho_{0_t1_t}]=(\rho_{0_t1_t}+\rho_{0_t1_t}^*)/2$. Note that we have been using the notation $\rho_{ij} \equiv \langle i\vert \rho \vert j \rangle $. 

Now, we turn our attention to the reduction in work due to the friction, which is described by the last term in the above expression. Let us then write
\begin{equation}
    \delta W_{\rm fric}(t) := -\varepsilon(t)\,\dot \theta_t \,\mathfrak{Re}[\langle 0_t \vert \rho_S(t) \vert 1_t \rangle ]
\end{equation}
for a given quantum state $\rho_S$. To get expressions for compression and expansion strokes found in the main text, we simply plug in the corresponding time-dependent eigenvalue gaps (e.g. given by Eq.~\eqref{eq:eigenvalues_compression}), time-dependent eigenvectors, as well as $\dot \theta_t$, and then integrate along the specific interval within the Otto cycle.

Let us consider first the contribution during the compression stroke. The frictional part of the work during the compression stroke (setting $t_0 = 0$) is
\begin{equation}
    \delta W_{\rm comp}^{\rm (fric)}(t) = -\frac{\omega \Omega_0 \mathfrak{Re} \left[\langle 0_t \vert \rho_{S}(t)\vert 1_t \rangle \right]}{\sqrt{(\omega \tau_{\rm comp})^2 + (\Omega_0 t)^2}}.
\end{equation}
Similarly, for the expansion stroke, the expression for the frictional part reads
\begin{equation}
    \delta W_{\rm exp}^{\rm (fric)} = \frac{\omega \Omega_0 \mathfrak{Re}\left[\langle 0_t \vert \rho_{S}(t)\vert 1_t\rangle \right]}{\sqrt{(\omega \tau_{\rm exp})^2 + (\Omega_0 (\tau_{\rm exp} - t + t_2))^2}}.
\end{equation}
The total contribution of a finite-time description of work during the compression and expansion strokes is then given by
\begin{align}
     \Delta W(t_1,0) &= \int_0^{t_1} \frac{\dot{\varepsilon}_{\rm comp}(t)}{2}(\,\rho_{0_t0_t}-\rho_{1_t1_t})\mathrm dt \\
     &- \int_0^{t_1}  \frac{\omega \Omega_0 \, \mathfrak{Re}\!\big[\langle 0_t\vert \rho_S(t)\vert 1_t \rangle ]} {\sqrt{\omega^2 \tau_{\rm comp}^2+\Omega_0^2t^2} }  \mathrm dt,
\end{align}
and
\begin{align}
     \Delta W(t_3,t_2) &= \int_{t_2}^{t_3} \frac{\dot{\varepsilon}_{\rm exp}(t)}{2}(\,\rho_{{0}_t {0}_t}-\rho_{{1}_t {1}_t})\mathrm dt \\
     &+ \int_{t_2}^{t_3} \frac{\omega \Omega_0 \, \mathfrak{Re}\left[\langle  0_t \vert \rho_{S}(t)\vert  1_t\rangle \right]}{\sqrt{(\omega \tau_{\rm exp})^2 + (\Omega_0 (\tau_{\rm exp} - t + t_2))^2}} \mathrm dt,
\end{align}
as we wanted to show. 

\section{Extensions of the shortcut to adiabaticity beyond the single-qubit drive}\label{app:generalization}

In  Section~\ref{sec:Zeno_assisted_quantum_heat_engine} we showed, for the specific two-level drive used in the Otto cycle, how strong coupling plus frequent monitoring produces an effective shortcut to adiabaticity. The purpose of this appendix is to state the more general mechanism clearly: the same construction applies to arbitrary finite-dimensional working systems, provided the assumptions of Theorem~\ref{theorem:strong_coupling_theorem}  below (adapted from Ref.~\cite{burgarth2022oneboundtorulethem}) are satisfied. 

\begin{theorem}[Adapted from Ref.~\cite{burgarth2022oneboundtorulethem}]\label{theorem:strong_coupling_theorem}
    Let $t \in [0,\tau] \mapsto H_{\Gamma}(t)$ be an integrable Hamiltonian, with $H_\Gamma(t)$ self-adjoint and bounded for all $t \in [0,\tau]$ and all $\Gamma > 0$. Let 
    \begin{equation}
        U_\Gamma(t) := \mathcal{T}\left[\exp\left(-i\int_0^t \mathrm{d}s \, H_\Gamma(s)\right)\right]
    \end{equation}
    be the evolution generated by $H_\Gamma(t)$. Assume that for all $t \in [0,\tau]$ there exists the limit 
    \begin{equation}
        H_0(t) = \lim_{\Gamma \to \infty}\frac{1}{\Gamma}H_\Gamma(t),
    \end{equation}
    where $H_0(t)$ is self-adjoint, has the finite spectral representation 
    \begin{equation}
        H_0(t) = \sum_{\ell = 1}^m E_\ell (t)P_\ell (t),
    \end{equation}
    with $\{E_\ell(t)\} \subseteq \mathbbm{R}$, while $P_\ell(t)=P_\ell(t)^\dagger$, $P_\ell(t)P_\kappa(t)=\delta_{\ell \kappa}P_\ell(t)$ for all $\ell,\kappa$. Moreover, assume that $E_\ell(t)$ is $C^1$ (i.e., every $E_\ell(t)$ has a time derivative and that derivative is a continuous function), and $P_\ell(t)$ is $C^2$ (i.e., every $P_\ell(t)$ has a first and second derivative, and the second derivative is a continuous function), and that there are no level crossings, i.e.,
    \begin{equation}
        \vert \omega_{\kappa,\ell}(t)\vert := \vert E_\ell(t)-E_\kappa(t)\vert > 0,
    \end{equation}
    for all $t \in [0,\tau]$ and $\kappa \neq \ell$. Now, assume that 
    \begin{equation}
        G_\Gamma(t) := H_\Gamma(t)-\Gamma \,H_0(t)
    \end{equation}
    is differentiable and $\Vert G_\Gamma \Vert_{\infty,\tau}, \Vert \dot G_{\Gamma}\Vert_{\infty,\tau} = o(\sqrt{\Gamma})$. Then, we have that $U_\Gamma(t)$ uniformly converges to 
    \begin{equation}
       \mathcal{T}\left[\exp\left(-i\int_0^t \mathrm{d}s \big(\Gamma H_0(s)+G_{\Gamma,\rm Zeno}(s)+A(s)\big) \right)\right] 
    \end{equation}
    as $\Gamma \to \infty$ for any $t \in [0,\tau]$, where $A(t)$ is the generator of the quasistatic  transporter
    \begin{equation}
        A(t) = \frac{i}{2}\sum_{\ell=1}^m [\dot P_\ell(t),P_\ell (t)]
    \end{equation}
    and
    \begin{equation}
        G_{\Gamma,\rm Zeno}(t) = \sum_{\ell=1}^m P_\ell (t)G_\Gamma (t) P_\ell(t)
    \end{equation}
    is the time-dependent Zeno Hamiltonian of $G_\Gamma(t)$ relative to $\{P_\ell(t)\}$. The convergence error is bounded by
    \begin{widetext}
    \begin{align}\label{eq:precise_bound_operator_norm}
        &\left \Vert U_\Gamma(t)-\mathcal{T}\left[\exp\left(-i\int_0^t \mathrm{d}s \left(\Gamma H_0(s)+G_{\Gamma,\rm Zeno}(s)+A(s)\right)\right)\right] \right\Vert \leq \frac{\sqrt{m}}{\Gamma \eta}(1 + \tau \Vert A \Vert_{\infty, \tau}+2\Vert G_{\Gamma}\Vert_{\infty,\tau} ) \times \nonumber \\& \times \Bigr [ \left(2 + \frac{\eta'}{\eta}\tau\right)(\Vert A \Vert_{\infty, \tau}+\Vert G_{\Gamma}\Vert_{\infty,\tau})+\tau(\Vert \dot A \Vert_{\infty, \tau}+\Vert\dot  G_{\Gamma}\Vert_{\infty,\tau}+2\Vert A\Vert_{\infty,\tau} \Vert G_{\Gamma}\Vert_{\infty,\tau})\Bigr], 
    \end{align}
    \end{widetext}
    where $\eta$ and $\eta'$ are the minimal spectral gap and the maximal spectral slope, defined respectively as 
    \begin{equation}
        \eta = \min_{\kappa,\ell:\kappa\neq \ell} \,\min_{t \in [0,\tau]} \vert \omega_{\kappa,\ell}(t)\vert,\quad \eta' =  \max_{\kappa,\ell:\kappa\neq \ell} \,\max_{t \in [0,\tau]} \vert \dot \omega_{\kappa,\ell}(t)\vert.
    \end{equation}
\end{theorem}

Above, we have used the norm $\Vert \cdot \Vert_{\infty,T}$ is defined as $\Vert A \Vert_{\infty,T} := \sup_{t \in [0,T]} \Vert A \Vert$, where $\Vert A \Vert$ is the operator norm. The proof of this theorem can be found in Ref.~\cite{burgarth2022oneboundtorulethem}. Here we use it as a technical input and specialize it to the class of Hamiltonians relevant for our lubrication protocol. The resulting corollary makes explicit that the shortcut-to-adiabaticity mechanism is not restricted to the single-qubit model studied in the main text.

\begin{corollary}[Shortcuts to adiabaticity from strong coupling]
    Let $t \in [0,\tau] \mapsto  H_\Gamma(t)$ be a time-dependent Hamiltonian given by 
    \begin{equation}\label{eq:starting_Hamiltonian_of_the_corollary}
        H_\Gamma(t) = \Gamma H_{\rm int}^{S}(t) \otimes H_{\rm int}^{L}+H_S(t)\otimes \mathbbm{1}+ \mathbbm{1} \otimes H_L
    \end{equation}
    which satisfies all the conditions from Theorem~\ref{theorem:strong_coupling_theorem}, where $H_S(t)$ is an arbitrary time-dependent Hamiltonian for a $d_S$-dimensional system $S$ and $H_L$ is an arbitrary time-independent Hamiltonian for a $d_L$-dimensional system $L$. Further, suppose that $[H_S(t),H_{\rm int}^S(t)] = 0$ for all $t \in [0,\tau]$, that $H_{\rm int}^S(t)$ and $H_{\rm int}^L$ are invertible, and, writing the spectral decomposition $$H_{\mathrm{int}}^L = \sum_{m=1}^{d_L} q_{m}^{\rm (int)} Q_m,$$ that 
    \begin{align}
        \sum_{m=1}^{d_L} Q_m \, H_L \, Q_m = 0
    \end{align}
    holds. Then, in the limit of $\Gamma \to \infty$ the unitary evolution $U_\Gamma(t)$ converges uniformly to  $U_{\rm eff}(t)$ generated by the Hamiltonian 
    \begin{equation}
        H_{\rm eff}(t) = \Gamma H_{\rm int}^S(t) \otimes H_{\rm int}^L + \big(H_S(t) + A_S(t)\big)\otimes \mathbbm{1},
    \end{equation}
    where $A_S(t)$ is the quasistatic transporter relative to the spectral decomposition of $H_S(t)$.
\end{corollary}

\begin{proof}
    The proof follows closely the strategy employed in the main text. We start noticing that for the specific Hamiltonian $H_\Gamma(t)$ given by Eq.~\eqref{eq:starting_Hamiltonian_of_the_corollary} the Hamiltonian $H_0(t) = \lim_{\Gamma \to \infty}\frac{1}{\Gamma} H_\Gamma (t)$ is precisely
    \begin{equation}
        H_0(t) = H_{\rm int}^S(t) \otimes H_{\rm int}^L \,.
    \end{equation}
    Consider the spectral decomposition $H_S(t) = \sum_{n=1}^{d_S} \varepsilon_n(t) P_n(t)$
    of the system's local Hamiltonian. Since, by assumption,
    $[H_{\rm int}^S(t),H_S(t)] = 0$ for all $t$ we can write
    \begin{equation}
        H_{\rm int}^S(t) = \sum_{n=1}^{d_S} \varepsilon^{(\rm int)}_n(t) P_n(t),
    \end{equation}
    and, together with $H_{\rm int}^L = \sum_{m=1}^{d_L} q_m^{(\rm int)} Q_m$, we obtain
    the finite spectral representation
    \begin{align}
    H_0(t) &= \sum_{n,m}E_{n,m}(t)P_{n,m}(t)\\
    &= \sum_{n,m} \varepsilon_n^{\rm (int)}(t)\,q_m^{(\rm int)}\, P_n(t) \otimes Q_m.
    \end{align}    

    If the products $\varepsilon_n^{(\rm int)}(t) q_m^{(\rm int)}$ are not all distinct, the spectral projectors of $H_0(t)$ are the coarser sums $P_\ell(t) = \sum_{(n,m) \in S_\ell} P_n(t)\otimes Q_m$ over each degenerate block $S_\ell$. Two distinct pairs $(n,m) \neq (n',m')$ can only lie in the same block with $n=n'$ if $\varepsilon_n^{(\rm int)}(t) = 0$, and with $m = m'$ if $q_m^{(\rm int)} = 0$; both are excluded by invertibility. Hence every pair sharing a block has $n \neq n'$ and $m \neq m'$, and all cross terms below vanish through $P_n(t)P_{n'}(t) = 0$ or $Q_mQ_{m'} = 0$. The computation may therefore be carried out with the finer family $\{P_n(t)\otimes Q_m\}_{n,m}$.

\begin{figure}[t]
    \centering    \includegraphics[width=\columnwidth]{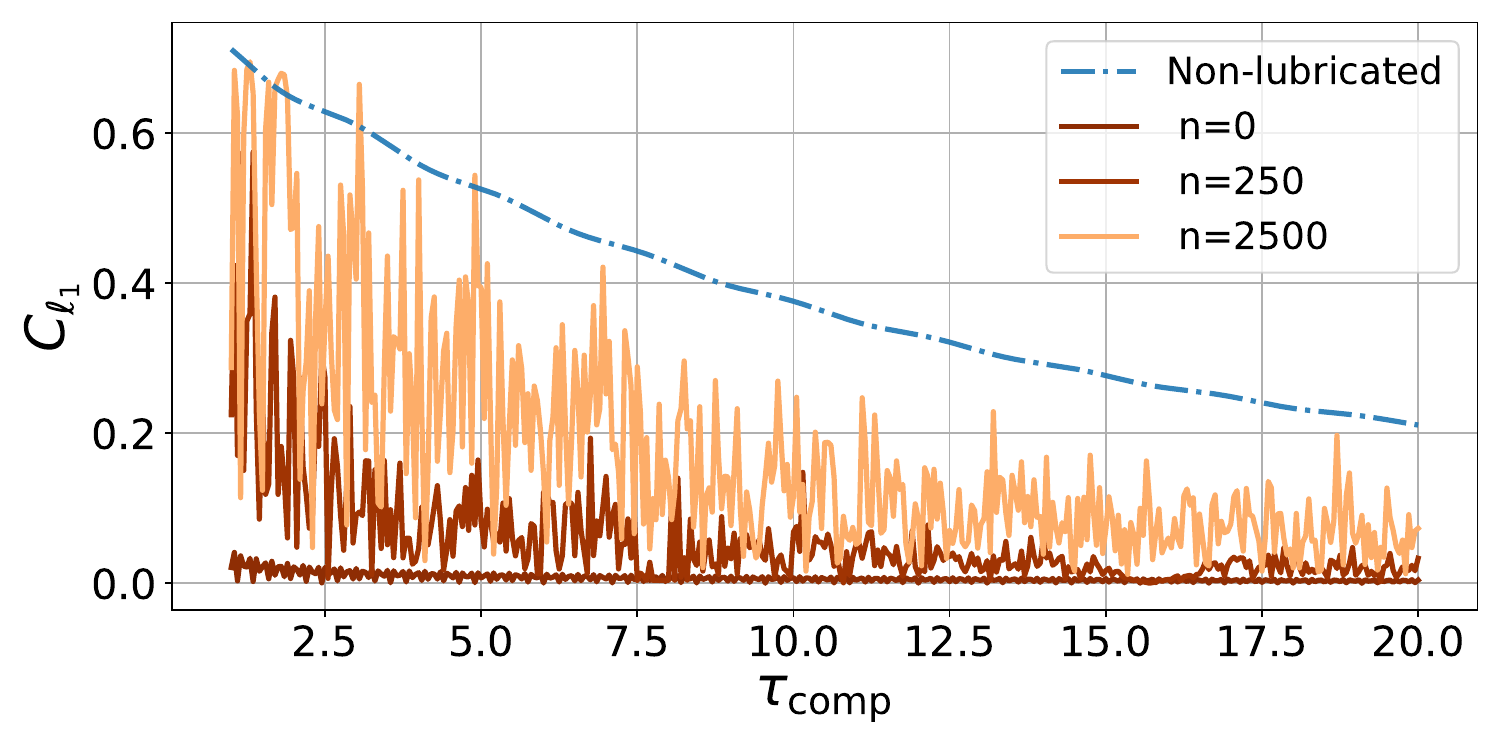}
    \caption{\textbf{The effect of frequently monitoring the lubricant system with the computational basis.} If $\Gamma/n \ll 1$ the unitary pulses $U_{\rm tot}(\delta t)$ are far from the strong-coupling regime, and thus far from the Zeno limit. The curve is approximated from simulated values at intervals of 0.05 in $\tau_{\rm comp}$. Parameters: $\Gamma = 50$, $\omega=\omega_L=1$, $\Omega_0=5$, $T_{\rm c}=0.5$.}    \label{fig:measurements_make_things_worse}
\end{figure}

    From Theorem~\ref{theorem:strong_coupling_theorem} we know that the evolution will converge uniformly to $U_{\rm eff}(t)$ generated by the effective Hamiltonian 
    \begin{equation}
        H_{\rm eff}(t) = \Gamma H_{\rm int}^S(t) \otimes H_{\rm int}^L+ A(t)+G_{\Gamma,\rm Zeno}(t),
    \end{equation}
    where we calculate $A(t)$ and $G_{\Gamma,\rm Zeno}(t)$ from the projectors $\{P_n(t) \otimes Q_m\}_{n,m}$. In this case, we find 
    \begin{align}
        &G_{\Gamma, \rm Zeno}(t) =  \nonumber \\ 
        &=\sum_{n,m}P_n(t) \otimes Q_m \left( H_S(t)\otimes \mathbbm{1}+\mathbbm{1}\otimes H_L\right)P_n(t) \otimes Q_m \nonumber \\
        &= H_S(t)\otimes \mathbbm{1} 
    \end{align}
    where we have used the assumption that $\sum_m Q_m H_L Q_m = 0$. Moreover, we have also that calculating $A(t)$ one obtains
    \begin{align}
        A(t) &= \frac{i}{2}\sum_{n,m}\left[ \frac{\mathrm{d}}{\mathrm{d}t} \left (P_n(t) \otimes Q_m \right) , P_n(t) \otimes Q_m\right] \\
        &= \frac{i}{2}\sum_{n,m} \left [\dot P_n(t), P_n(t) \right] \otimes Q_m \\
        &= \frac{i}{2}\sum_n \left [\dot P_n(t), P_n(t) \right] \otimes \mathbbm{1}.
    \end{align}
    This concludes the proof.
\end{proof}

With the above, it is now clear that any $d_S$-dimensional system can be lubricated via our method, for example, by considering an auxiliary lubricant system as we have done in the main text. \newline

\section{Frequent monitoring and the role of commutativity}\label{app:commutativity}

In the main text we first took the strong-coupling approximation for the evolution generated by $H_{\mathrm{tot}}(t)$ in Eq.~\eqref{eq:total_hamiltonian_SL} and then applied frequent monitoring of the lubricant in a basis that commutes with the interaction term $\Gamma R(t)\otimes X$. This appendix clarifies why that choice is not important in our construction, and when the order of the two limits matters.

Let us instead consider taking the frequent-monitoring limit first, which leads to an evolution generated by the effective Zeno Hamiltonian
\begin{align}
    {H'}_{\rm eff}^{(\ell)}(t) &= \mathbbm{1}\otimes \vert \ell \rangle \langle \ell \vert H_{\rm tot}(t) \mathbbm{1}\otimes \vert \ell \rangle \langle \ell \vert\\
    &=\ell \, \Gamma R(t) \otimes \vert \ell \rangle \langle \ell \vert + H_{S}(t)\otimes \vert \ell \rangle \langle \ell \vert,
\end{align}
since $\langle \ell \vert Z \vert \ell \rangle = 0$ for $\ell \in \{+,-\}$. Applying the strong-coupling Theorem~\ref{theorem:strong_coupling_theorem} to this Hamiltonian, the effective evolution in each Zeno subspace is then generated by
\begin{equation}
    {H'}_{\rm Zeno}^{(\ell)}(t) = \ell \Gamma R(t) \otimes \vert \ell \rangle \langle \ell \vert + \bigr(H_{S}(t)+A_{S}(t) \bigr) \otimes \vert \ell \rangle \langle \ell \vert .
\end{equation}
We thus conclude that ${H'}_{\rm Zeno}^{(\ell)}(t)$ just obtained coincides with $H_{\rm Zeno}^{(\ell)}(t)$ derived in the main text.

The situation changes if the lubricant is monitored in a basis that anticommutes with the interaction term. For example, if one performs frequent monitoring in the computational basis $\{\vert i \rangle \}_{i=0,1}$ of the lubricant, the two limits no longer match. Theoretically, this is seen because in that case the monitoring suppresses the very interaction responsible for the shortcut to adiabaticity. One is then combining two limits that work against each other: if $\Gamma/n$ remains finite, jumps between Zeno subspaces persist; if $\Gamma/n\ll 1$, the interaction is effectively averaged away and the protocol reduces to the non-lubricated dynamics. Figure~\ref{fig:measurements_make_things_worse} illustrates this competition.

\end{appendix}

\end{document}

%% file: preamble.tex
\usepackage{amsthm}
\usepackage{mathtools}
\usepackage{physics}
\usepackage{bbm}
\usepackage{amsfonts}
\usepackage{amssymb}
\usepackage{graphicx}
\usepackage{xfrac}
\usepackage[T1]{fontenc}
\usepackage{bbold}
\usepackage[utf8]{inputenc}
\usepackage{babel}
\usepackage{bbding}
\usepackage{xcolor}
\usepackage{quantikz}
\usepackage{booktabs}
\usepackage{tabularx}
\usepackage{cancel}
\usepackage{utfsym}

\usepackage[table,xcdraw]{xcolor}
\definecolor{color_scheme}{RGB}{141, 186, 128}
\usepackage[colorlinks=true,linktocpage=true]{hyperref}
\hypersetup{
    allcolors  = {color_problem_box!60!black},
}

\theoremstyle{plain}
\newtheorem{theorem}{Theorem}

\newtheorem{corollary}{Corollary}

\theoremstyle{definition}

\usepackage{amsmath,amsthm,amssymb,amsfonts}

\usepackage{graphicx}
\usepackage{multirow}
\usepackage{url}
\usepackage[usenames,dvipsnames]{xcolor}
\usepackage{bbold}
\usepackage[normalem]{ulem}

\usepackage[most]{tcolorbox} 

\newtcolorbox{question}{
  colback=white,
  colframe=blue!50!black,
  fonttitle=\bfseries,
  title=Question,
  sharp corners,
  boxrule=0.8pt
}

\definecolor{color_problem_box}{RGB}{141, 186, 128}

\newtcolorbox{mainresultsbox}[2][]{%
  attach boxed title to top center = {yshift=-8pt},
  colback      = color_problem_box!5!white,
  colframe     = green!50!black,
  halign       = flush left,
  fonttitle    = \bfseries\sffamily,
  colbacktitle = color_problem_box!95!black,
  title        = #2,
  enhanced,
  float        = t,
}

\newcommand{\firstbox}{\hyperlink{box_1}{\color{color_scheme!60!black}{{Box 1}}}}

\newcommand{\secondbox}{\hyperlink{box_2}{\color{color_scheme!60!black}{{Box 2}}}}